\documentclass{lmcs}
\pdfoutput=1

\usepackage{lastpage}
\lmcsdoi{21}{4}{12}
\lmcsheading{}{\pageref{LastPage}}{}{}%
{Jul.~23,~2024}{Oct.~28,~2025}{}

\usepackage[noadjust,nocompress]{cite}
\usepackage{booktabs}
\usepackage{siunitx}
\usepackage{colortbl}
\usepackage{graphicx}
\usepackage{amsmath}
\usepackage{mathabx}
\usepackage{ebproof}
\usepackage[utf8]{inputenc}
\ebproofset{
  right label template={$\scriptstyle(\inserttext)$},
  center=false,
}

\usepackage{color}

\newcommand{\hori}{\mathbf{H}}
\newcommand{\dblRel}{\mathbb{R}\mathbf{el}}
\newcommand{\dblDist}{\mathbb{D}\mathbf{ist}}
\newcommand{\dblPRel}{\mathbb{P}\mathbf{RRel}}
\newcommand{\dblCG}{\mathbb{C}\mathbf{G}}
\newcommand{\dblTCG}{\mathbb{T}\mathbf{CG}}
\newcommand{\companion}[1]{\widehat{#1}}
\newcommand{\conjoint}[1]{\widecheck{#1}}
\newcommand{\dblVis}{\mathbb{V}\mathbf{is}}

\makeatletter
\DeclareFontFamily{OMX}{MnSymbolE}{}
\DeclareSymbolFont{MnLargeSymbols}{OMX}{MnSymbolE}{m}{n}
\SetSymbolFont{MnLargeSymbols}{bold}{OMX}{MnSymbolE}{b}{n}
\DeclareFontShape{OMX}{MnSymbolE}{m}{n}{
    <-6>  MnSymbolE5
   <6-7>  MnSymbolE6
   <7-8>  MnSymbolE7
   <8-9>  MnSymbolE8
   <9-10> MnSymbolE9
  <10-12> MnSymbolE10
  <12->   MnSymbolE12
}{}
\DeclareFontShape{OMX}{MnSymbolE}{b}{n}{
    <-6>  MnSymbolE-Bold5
   <6-7>  MnSymbolE-Bold6
   <7-8>  MnSymbolE-Bold7
   <8-9>  MnSymbolE-Bold8
   <9-10> MnSymbolE-Bold9
  <10-12> MnSymbolE-Bold10
  <12->   MnSymbolE-Bold12
}{}

\let\llangle\@undefined
\let\rrangle\@undefined
\DeclareMathDelimiter{\llangle}{\mathopen}%
                     {MnLargeSymbols}{'164}{MnLargeSymbols}{'164}
\DeclareMathDelimiter{\rrangle}{\mathclose}%
                     {MnLargeSymbols}{'171}{MnLargeSymbols}{'171}
\makeatother

\usepackage{url}

\usepackage{stmaryrd}
\usepackage{amsfonts}
\usepackage{amssymb}
\usepackage{amsthm}
\usepackage{mathtools}
\usepackage{cmll}
\usepackage[all]{xy}
\usepackage{tikz}
\usepackage{adjustbox}
\usepackage{hyperref}
\usepackage{mathrsfs}
\usetikzlibrary{calc,matrix,decorations.markings,decorations.pathreplacing,arrows,cd,positioning,shapes.misc}
\usetikzlibrary{decorations.pathmorphing}
\usetikzlibrary{backgrounds}
\tikzstyle{game-causality}=[dotted, thick]
\tikzstyle{strat-causality}=[->, thick,  -open triangle 60]
\tikzset{curve/.style={settings={#1},to path={(\tikztostart)
    .. controls ($(\tikztostart)!\pv{pos}!(\tikztotarget)!\pv{height}!270:(\tikztotarget)$)
    and ($(\tikztostart)!1-\pv{pos}!(\tikztotarget)!\pv{height}!270:(\tikztotarget)$)
    .. (\tikztotarget)\tikztonodes}},
    settings/.code={\tikzset{quiver/.cd,#1}
        \def\pv##1{\pgfkeysvalueof{/tikz/quiver/##1}}},
    quiver/.cd,pos/.initial=0.35,height/.initial=0}

\tikzset{
  prof/.style = {decoration = {markings, mark = at position 0.5 with { \node[transform shape, yscale=.4] {$|$}; } }, postaction = {decorate} },
}
\tikzstyle{conflict}=[decorate, decoration={snake,amplitude=.3mm,segment length=2mm},-]
\tikzstyle{neutralnode}=[fill=gray!25, draw, thick, inner sep=2pt,
minimum size=10pt, minimum width=10pt, rounded corners=2, ]
\tikzstyle{posnode}=[fill=blue!25, draw, thick, inner sep=2pt, minimum
size=10pt,rounded corners=2]
\tikzstyle{negnode}=[fill=red!25, draw, thick, inner sep=2pt, minimum size=10pt, rounded corners=2]
\tikzstyle{background rectangle}=  [fill=gray!10]
\tikzstyle{background rectangle 2}=  [fill=olive!10]
\tikzstyle{background rectangle 3}=  [fill=teal!8]

\newcommand{\Rel}{\mathbf{Rel}}
\newcommand{\Set}{\mathbf{Set}}
\newcommand{\PRel}{\mathbf{PRRel}}
\newcommand{\Sym}{\mathbf{Sym}}
\newcommand{\Dist}{\mathbf{Dist}}
\newcommand{\ES}{\mathbf{ES}}
\newcommand{\CG}{\mathbf{CG}}
\newcommand{\TCG}{\mathbf{TCG}}

\newcommand{\ESS}{\mathbf{ESS}}
\newcommand{\Esp}{\mathbf{Esp}}
\newcommand{\Vis}{\mathbf{Vis}}
\newcommand{\WVis}{\mathbf{Wis}}

\newcommand{\intr}[1]{\llbracket #1 \rrbracket}
\newcommand{\lin}{\multimap}
\newcommand{\Mf}{\mathcal{M}}
\newcommand{\iso}{\cong}
\newcommand{\bij}{\simeq}

\newcommand{\C}{\mathcal{C}}
\newcommand{\D}{\mathcal{D}}
\newcommand{\display}{\partial}

\newcommand{\op}{\mathrm{op}}
\newcommand{\ot}{\leftarrow}
\newcommand{\id}{\mathrm{id}}

\newcommand{\Fam}{\mathbf{Fam}}
\newcommand{\N}{\mathbb{N}}

\newcommand{\U}{\mathbf{U}}

\newcommand{\tensor}{\otimes}

\newcommand{\der}{\mathrm{der}}
\newcommand{\dder}{i}
\DeclarePairedDelimiter\tuple\langle\rangle
\def\profto{\mathrel{\mkern11mu\raisebox{.25mm}{$\shortmid$}\mkern-16mu{\longrightarrow}}}

\newcommand{\sym}{\cong}
\newcommand{\pair}[1]{\tuple{#1}}
\newcommand{\tto}{\Rightarrow}
\newcommand{\evm}{\mathbf{ev}}
\renewcommand{\cc}{\mathbf{\hspace{.5ex}c\!\!\!\!c\,}}
\newcommand{\cleq}{\vartriangleleft}
\newcommand{\CC}{\mathbf{CC}}
\DeclarePairedDelimiter\coll\|\|
\newcommand{\collexp}[1]{\coll{#1}_{\oc}}

\renewcommand{\tilde}[1]{\mathscr{S}(#1)}
\newcommand{\ntilde}[1]{\mathscr{S}_-(#1)}

\newcommand{\ptilde}[1]{\mathscr{S}_+(#1)}
\newcommand{\tildep}[1]{\mathscr{S}^+(#1)}
\newcommand{\dom}{\mathrm{dom}}
\newcommand{\cod}{\mathrm{cod}}
\newcommand{\pid}{\mathsf{pid}}
\newcommand{\pcomp}{\mathsf{pcomp}}
\newcommand{\comp}[1]{\mathscr{G}(#1)}
\newcommand{\join}{\mathsf{join}}
\newcommand{\runit}{\mathsf{runit}}
\newcommand{\lunit}{\mathsf{lunit}}
\newcommand{\w}{\mathsf{w}}
\renewcommand{\v}{\mathsf{v}}
\newcommand{\inj}{\mathsf{i}}
\newcommand{\q}{\mathsf{q}}
\newcommand{\unf}{\mathsf{unf}}
\newcommand{\fld}{\mathsf{fld}}

\newcommand{\e}{\mathsf{e}}
\newcommand{\pder}{\mathsf{pder}}
\newcommand{\pprom}{\mathsf{pprom}}
\newcommand{\just}{\mathsf{j}}
\newcommand{\init}{\mathsf{init}}
\newcommand{\card}[1]{|#1|}
\newcommand{\seely}{\mathsf{see}}

\newcommand{\ev}[1]{|#1|}
\newcommand{\conflict}{\mathrel{\#}}
\newcommand{\conf}[1]{\mathscr{C}(#1)}

\newcommand{\confp}[1]{\mathscr{C}^+(#1)}
\newcommand{\pol}{\mathrm{pol}}
\renewcommand{\qu}{\q}
\newcommand{\pr}{\partial}
\newcommand{\enb}{\vdash}
\newcommand{\imc}{\rightarrowtriangle}
\newcommand{\done}{\checkmark}
\newcommand{\gcc}{\mathsf{gcc}}
\newcommand{\mconflict}{\!\!\xymatrix@C=15pt{\, \ar@{~}[r]&\,}\!\!\!\!\!}
\newcommand{\nconf}[1]{\mathscr{C}^0(#1)}

\newdir{|>}{%
!/4.5pt/@{|}*:(1,-.2)@{>}*:(1,+.2)@_{>}}
\newdir{pb}{:(1,-1)@^{|-}}
\def\pb#1{\save[]+<16 pt,0 pt>:a(#1)\ar@{pb{}}[]\restore}

\newcommand{\tbool}{\mathbb{B}}
\newcommand{\tnat}{\mathbb{N}}
\newcommand{\tunit}{\mathbb{U}}

\newcommand*{\ty}{a}
\newcommand*{\seq}[1]{\langle #1 \rangle}
\newcommand*{\prom}[1]{\llangle #1 \rrangle}
\newcommand*{\seqdots}[3]{\seq{#1_{#2}, \dots, #1_{#3} } }
\newcommand*{\tyf}[1]{\textcolor{blue}{#1}}
\newcommand*{\tyl}{\vec{a}}
\newcommand*{\morpCone}{\eta}

\newcommand{\ttrue}{\mathbf{t\!t}}
\newcommand{\tfalse}{\mathbf{f\!f}}
\newcommand{\choice}{\mathbf{choice}}

\newcommand{\PCF}{\mathsf{PCF}}

\begin{document}
\title[From Thin Concurrent Games to Generalized Species of Structures]{From Thin Concurrent Games\\
to Generalized Species of Structures\\
(extended version)}

\author[P. Clairambault]{Pierre Clairambault\lmcsorcid{0000-0002-3285-6028}}[a]
\author[F. Olimpieri]{Federico Olimpieri\lmcsorcid{0000-0003-1485-5360}}[a]
\author[H. Paquet]{Hugo Paquet\lmcsorcid{0000-0002-8192-0321}}[b]

\address{Aix Marseille Univ, CNRS, LIS, Marseille, France}
\email{pierre.clairambault@lis-lab.fr, federico.olimpieri@lis-lab.fr}
\address{Inria, École Normale Supérieure – PSL, CNRS, Paris, France}
\email{Hugo.Paquet@inria.fr}

\begin{abstract}
  Two families of denotational models have emerged from the semantic
  analysis of linear logic: \emph{dynamic} models, typically presented
  as game semantics, and \emph{static} models, typically based on a
  category of relations. In this paper we introduce a formal bridge
  between a dynamic model and a static model: the model of \emph{thin
    concurrent games and strategies}, based on event structures, and
  the model of \emph{generalized species of structures}, based on
  distributors. A special focus of this paper is the two-dimensional nature of the
  dynamic-static relationship, which we formalize with double
  categories and bicategories.
  
  In the first part of the paper, we construct a symmetric monoidal
  oplax functor from linear concurrent strategies to distributors. We
  highlight two fundamental differences between the two models: the
  composition mechanism, and the representation of resource
  symmetries.

In the second part of the paper, we adapt established methods from game
semantics (visible strategies, payoff structure) to enforce a tighter
connection between the two models. We obtain a cartesian closed
pseudofunctor, which we exploit to shed new light on recent results in
the theory of the $\lambda$-calculus.  
\end{abstract}

\maketitle

\section{Introduction}

The discovery of linear logic has had a deep influence on programming
language semantics. The linear analysis of resources provides a refined
perspective that leads, for instance, to important notions of program
approximation \cite{DBLP:journals/pacmpl/BarbarossaM20} and
differentiation \cite{DBLP:journals/tcs/EhrhardR03}.
Denotational models for higher-order programming languages can be
constructed from this resource-aware perspective, exploiting the fact
that every model of linear logic is also a model of the simply-typed
$\lambda$-calculus. 

In this paper, we clarify the relationship between two such
denotational models:
\begin{itemize}
  \item  \emph{Thin concurrent games}, a framework for game semantics
    introduced by Castellan, Clairambault, and Winskel 
\cite{DBLP:journals/lmcs/CastellanCW19}, in which 
    programs are modeled as concurrent strategies.
  \item \emph{Generalized species of structures}, a 
    combinatorial model developed by Fiore, Gambino, Hyland, and
    Winskel \cite{fiore2008cartesian}, in which programs are interpreted as
    categorical distributors (or profunctors) over groupoids. 
\end{itemize}
\noindent 
We carry out this comparison in a two-dimensional setting also
including morphisms between strategies
and morphisms between distributors. In the language of bicategory
theory, our first key contribution is a symmetric monoidal, \emph{oplax} functor of
bicategories
\begin{equation}
  \label{eq:functor-informal}
  \begin{tikzcd}
    \parbox{3.5cm}{\centering games, strategies and maps} 
\qquad\arrow{r}{} & \qquad
    \parbox{5cm}{\centering groupoids,
    distributors and natural transformations}
\end{tikzcd}
\end{equation}
showing, in particular, that symmetries of strategies can be
explained via groupoid actions.

\subsection{Static and dynamic models}

This work fits in a long line of research on the relationship between
\emph{static} and \emph{dynamic} denotational models arising from
linear logic. 

In a static model, programs are represented by
their input/output behavior, or by collecting representations of completed
executions. The simplest example is given by the category of sets and
relations: this is the relational model of linear logic
(Section~\ref{subsec:relational_model} or \cite{DBLP:journals/tcs/Girard87}). 
In a dynamic model, programs are represented by their
interactive behavior with respect to every possible execution
environment. This includes game semantics
(\cite{DBLP:journals/iandc/HylandO00,DBLP:journals/iandc/AbramskyJM00}), which has
proved incredibly proficient at modeling various computational features
or effects
(\cite{DBLP:conf/lics/AbramskyHM98,DBLP:conf/lics/Laird97,DBLP:journals/apal/GhicaM08}). 

To illustrate the difference, 
the identity 
at type $1 \multimap 1$ in game semantics is a \emph{strategy}
gathering \emph{plays} which represent chronologically
an exchange of information and control flow between the {\color{blue}program} and its
{\color{red}environment}: 
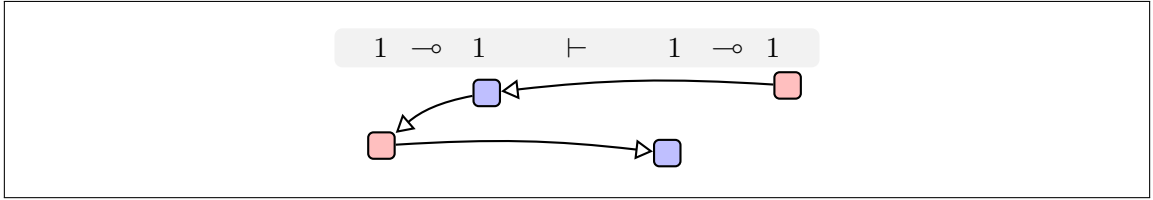
\[
  \begin{tikzpicture}
    \node (a) at (0.2, 3) {$\ \ \ 1$};
    \node (b) at (1, 3) {$\multimap$};
    \node (c) at (1.7, 3) {$1$};
    \node (d) at (3, 3) {$\vdash$};
    \node (e) at (4.3, 3) {$1$};
    \node (f) at (5, 3) {$\multimap$};
    \node (g) at (5.8, 3) {$1\ \ \ $};    
    \node[negnode] (1) at (5.8, 2.5) {};
    \node[posnode] (2) at (1.8, 2.4) {};
    \node[negnode] (3) at (0.4, 1.7) {};
    \node[posnode] (4) at (4.2, 1.6) {};
    \draw [bend right=5, strat-causality] (1) to (2);
    \draw [bend right=15, strat-causality] (2) to (3);
    \draw [bend left= 5, strat-causality] (3) to (4);
    \begin{pgfonlayer}{background}
    \draw[background rectangle ,rounded corners=3pt, draw=none]
(a.west) -- (a.north west) -- (g.north east) -- (g.south east) --
(a.south west) -- (a.west);
     \end{pgfonlayer}
  \end{tikzpicture}
\]

In contrast, in the relational model, $1 \multimap 1$ is
a one-element set containing a single input/output pair. The identity
relation over it can be seen as a 
collapsed version of the strategy above:
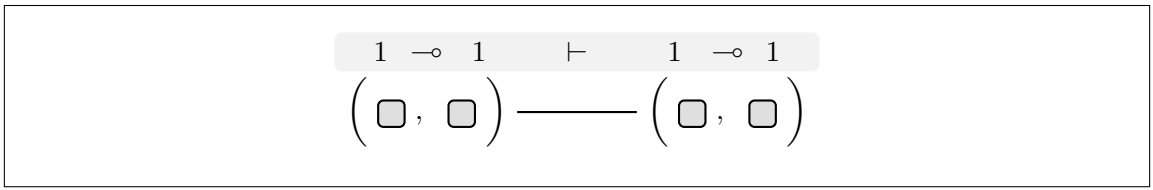
\newcommand{\onenode}{\raisebox{-0.3em}{\tikz\node[neutralnode] (f) at (0, 0) {};}}
  \[
  \begin{tikzpicture}
    \node (a) at (0.2, 3) {$\ \ \ 1$};
    \node (b) at (1, 3) {$\multimap$};
    \node (c) at (1.7, 3) {$1$};
    \node (d) at (3, 3) {$\vdash$};
    \node (e) at (4.3, 3) {$1$};
    \node (f) at (5, 3) {$\multimap$};
    \node (g) at (5.8, 3) {$1\ \ \ $};    
    \begin{pgfonlayer}{background}
    \draw[background rectangle ,rounded corners=3pt, draw=none]
(a.west) -- (a.north west) -- (g.north east) -- (g.south east) --
(a.south west) -- (a.west);
     \end{pgfonlayer}
     \node (first) at (1, 2.2) {$\biggl(\ \onenode\ ,\  \  \onenode\ \biggr)$};
     \node (last) at (5, 2.2) {$\biggl(\ \onenode\ ,\ \  \onenode\ \biggr)$};
       \draw[thick] (first) -- (last);
   \end{tikzpicture} 
\]

This suggests a simple equation
\[
\text{\emph{game semantics}} = \text{\emph{relational model}} + 
\text{\emph{time}}\,,
\]
so that going from game semantics to the relational model is a
simple matter of forgetting the temporal order of execution, an
intuition first explored in \cite{DBLP:conf/csl/BaillotDER97}.

But this naive intuition hides a fundamental difference
between the composition mechanisms in static and dynamic models:
strategies may \emph{deadlock}, while relations cannot. More precisely,
in game semantics, two strategies can synchronize by performing the
same actions \emph{in the same order}, whereas in the collapsed
version, only the actions matter and not the
order. Thus, as was quickly established
\cite{DBLP:conf/csl/BaillotDER97}, one cannot simply forget time in a
functorial way, and composition is usually only preserved in an oplax manner, as
in \eqref{eq:functor-informal}.

Static collapses of game semantics
require an adequate notion
of \emph{position} for a game. This is somewhat difficult to
define in traditional game semantics, but very natural in alternative
causal versions such as concurrent games, because we can look
at configurations of the underlying event
structure (Section~\ref{subsec:basic-cg}, see also
\cite{DBLP:journals/tcs/Mellies06, DBLP:conf/lics/Mellies05}).

The subtle relationship between static and dynamic models was refined
by many authors over almost three decades
(see \emph{e.g.} \cite{DBLP:conf/csl/BaillotDER97,
DBLP:conf/lics/Mellies05, DBLP:conf/tlca/Boudes09,
DBLP:journals/entcs/CalderonM10, DBLP:conf/lics/TsukadaO16}), to
identify settings in which functoriality can be restored -- in
particular, our work inherits insights from Melliès' line of work on
asynchronous games \cite{DBLP:conf/lics/Mellies05}. Leveraging this, we
show (in Section~\ref{sec:cc_pseudofunctor}) how the oplax functor
\eqref{eq:functor-informal} can be strictified to a
\emph{pseudofunctor}, that preserves composition up to isomorphism.

\subsection{Proof-relevant models and symmetries}

Our aim here is to take the static-dynamic relationship to a new level
that takes into account the symmetries implicit in the resource usage.
Symmetry plays an important role in game semantics
(Section~\ref{sec:symmetry}, or
\cite{DBLP:journals/iandc/AbramskyJM00,DBLP:journals/tcs/Mellies06,DBLP:conf/csl/CastellanCW14,mfps22}), but so far
connections only exist with static models whose symmetries are implicit
or quotiented, like the relational model.
We argue that generalized species, which  
represent combinatorial structures in terms of their
symmetries, provide a
convenient target for a static collapse of thin concurrent
games. 

The two models we consider are ``proof-relevant'' \cite{ol:proof}, in
the sense that the interpretation of a program provides, for each
possible execution, a set of proofs or \emph{witnesses} that this
execution can be realized.  This high degree of intensionality is
useful for modelling languages with non-deterministic features
\cite{DBLP:conf/lics/TsukadaO15,DBLP:conf/lics/CastellanCPW18}. In a
proof-relevant model, symmetries arise naturally in the linear
duplication of witnesses. In Section~\ref{subsec:prrel} we discuss the
limitations of a proof-relevant model without symmetries.

\subsection{Two-dimensional categories: bicategories and double categories.}
Proof-relevant models are naturally organized as two-dimensional
categories, because programs are interpreted as structured objects
(e.g.~strategies or distributors) between which there is a natural
notion of morphism. There are various kinds of 2-dimensional
categories. This line of research is primarily about semantics in
\emph{bicategories}, and indeed the final sections of this paper
(Section~\ref{sec:cc_pseudofunctor}, Section~\ref{sec:lambda}) are
about bicategories with structure and bicategorical functors. However,
in order to properly relate static and dynamic models, we find it
conceptually important and technically useful to regard them as richer
structures: \emph{pseudo double categories}. We do not assume any
prior knowledge of double categories, and we give more precise
motivation in Section~\ref{sec:double-categ-stat}.

\subsection{Bicategorical models of the $\lambda$-calculus}

As further motivation for this work, we note that the two-dimensional
and proof-relevant aspects are significant on the syntactic side.  The
interpretation of $\lambda$-terms in generalized species has a
presentation in terms of an \emph{intersection type system}
\cite{DBLP:conf/lics/TsukadaAO17,ol:intdist}, that takes into account the symmetries and can be
exploited to characterize computational properties and equational
theories of the $\lambda$-calculus.  More generally,
the structural 2-cells in a cartesian closed bicategory have a
syntactic interpretation as $\beta\eta$-rewriting steps in the
simply-typed $\lambda$-calculus (\cite{philip-marcelo:tt, seely:two}).

In Section~\ref{sec:cc_pseudofunctor} we connect to this line of work
by constructing a cartesian closed pseudofunctor, which preserves
the semantics of $\lambda$-terms in both typed and untyped settings.

\subsection{Outline of the paper and key contributions}
The paper is organized as follows.

\subsubsection*{Review of static models} In Section~\ref{sec:static} we recall
static semantics, including a
bicategory $\PRel$ of proof-relevant relations, and a bicategory
$\Dist$ of distributors.
One can view $\PRel$ as the sub-model of
$\Dist$ with no symmetries, so there is an embedding $\PRel \hookrightarrow
\Dist$. The bicategory $\Esp$ of generalized species is defined in
terms of $\Dist$.

In Section~\ref{sec:double-categ-stat}, we explain that these bicategories
arise from double categories whose structure can be exploited for our
purposes. We define the symmetric monoidal double categories
$\dblPRel$ and $\dblDist$ which embed $\PRel$ and $\Dist$. 

\subsubsection*{Collapsing concurrent games to static models}
In Section~\ref{sec:collapse_prel} we introduce the double category $\dblCG$ of ``plain'' concurrent
games without symmetry, and we show that a collapse operation gives
an oplax functor $\dblCG \to \dblPRel$
(Theorem~\ref{th:collPR}).

Then in Section~\ref{sec:symmetry} we add symmetry: we define
the double category $\dblTCG$ of thin concurrent games, with $\dblCG
\hookrightarrow \dblTCG$ as the sub-model with no symmetries. We show that every
strategy has a distributor of \emph{positive witnesses} (Proposition~\ref{prop:strategy-to-distributor}), and
that this extends to an oplax functor $\dblTCG \to \dblDist$
(Theorem~\ref{th:oplax_tcg}).
In summary, we have the following
situation involving embedding functors and collapse oplax functors, all
of which are additionally shown to preserve symmetric monoidal
structure:
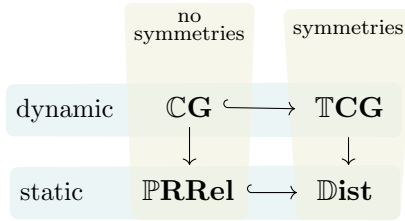
\[
\begin{tikzpicture}
  \matrix (m) [matrix of math nodes, row sep=1em, column sep=1em]
  {
  &[-1.8em] \substack{\text{no}\\\text{symmetries}}  &{\substack{\text{symmetries}}}  \\
\small\text{dynamic} \phantom{\bigl(\bigr)} &  \dblCG &\  \dblTCG\  \!\!\!\! \!\!\phantom{\bigl(\bigr)} \\
\small\text{static} \ \ \  \phantom{\bigl(\bigr)}  &  \dblPRel
\phantom{\bigl(\bigr)} \!\!\!\!\! &\  \dblDist\   \!\!\!\!\phantom{\bigl(\bigr)} \\
  };
  \draw[right hook->] (m-2-2) to  (m-2-3);
  \draw[->] (m-2-2) to (m-3-2);
   \draw[right hook->] (m-3-2) to[left] (m-3-3);
   \draw[->] (m-2-3) to  (m-3-3);
    \begin{pgfonlayer}{background}
    \draw[background rectangle 3,rounded corners=3pt, draw=none]
(m-2-1.west) -- (m-2-1.north west) -- (m-2-3.north east) -- (m-2-3.south east) --
(m-2-1.south west) -- (m-2-1.west);
     \end{pgfonlayer}
    \begin{pgfonlayer}{background}
    \draw[background rectangle 3,rounded corners=3pt, draw=none]
(m-3-1.west) -- (m-3-1.north west) -- (m-3-3.north east) -- (m-3-3.south east) --
(m-3-1.south west) -- (m-3-1.west);
     \end{pgfonlayer}
    \begin{pgfonlayer}{background}
    \draw[opacity=0.6, background rectangle 2,rounded corners=3pt, draw=none]
(m-1-2.north) -- (m-1-2.north east) -- (m-3-2.south east) -- (m-3-2.south) --
(m-3-2.south west) -- (m-1-2.north west) -- (m-1-2.north);
     \end{pgfonlayer}
    \begin{pgfonlayer}{background}
    \draw[opacity=0.6,background rectangle 2,rounded corners=3pt, draw=none]
(m-1-3.north) -- (m-1-3.north east) -- (m-3-3.south east) -- (m-3-3.south) --
(m-3-3.south west) -- (m-1-3.north west) -- (m-1-3.north);
     \end{pgfonlayer}
   \end{tikzpicture}
 \]

In Section~\ref{sec:collapse_linear} we include the mechanism of
\emph{visibility} \cite{DBLP:journals/corr/abs-2103-15453} to eliminate
deadlocks; this refines the above into \emph{pseudo-functors} between
symmetric monoidal double categories; yielding symmetric monoidal
pseudofunctors between the induced symmetric monoidal bicategories.

\subsubsection*{Cartesian closed structure and the $\lambda$-calculus}
Finally, in Section~\ref{sec:cc_pseudofunctor} we add the exponential
modality and focus on the cartesian closed structure. To ensure
preservation of the exponential modality, we introduce a refine
version of $\TCG$ called $\WVis$, using a \emph{payoff} mechanism from
game semantics (\cite{DBLP:conf/lics/Mellies05,
DBLP:journals/pacmpl/ClairambaultV20}).
Thus we obtain a pseudofunctor
$\WVis \longrightarrow \Dist$ (Theorem~\ref{th:pseudo_vws_dist})
and by also refining our categorical structure with a \emph{relative}
pseudo-comonad, we derive a pseudofunctor $\WVis_\oc \to \Esp$
(Theorem~\ref{th:main}) which we show is cartesian closed. 

Finally in Section~\ref{sec:lambda} we apply this result to the semantics of
untyped $\lambda $-calculus: we build a reflexive object in $\WVis$ and
show that, under our pseudofunctor, this is sent to an extensional
\emph{categorifed graph model} $\D^\ast$ in $ \Esp$
\cite{ol:proof}. 

\subsubsection*{Note on this long version.} This paper is an extension
of our earlier paper \cite{DBLP:conf/lics/ClairambaultOP23}, presented
at the LICS 2023 conference. The extended version includes all proofs
of technical results, and some entirely new contributions. In
particular, all sections up to Section~\ref{sec:cc_pseudofunctor} are
developed in a double-categorical setting which refines and
generalizes the bicategorical setting of
\cite{DBLP:conf/lics/ClairambaultOP23}, and provides simple proofs of
new structural results. See Section~\ref{sec:double-categ-stat} for
motivation. Finally, this long version also allows for further
illustration: we showcase the induced correspondence between game
semantics and generalized species over an example $\lambda$-term.

\section{A Tour of Static Semantics}
\label{sec:static}

In this section we present three static models: the basic relational
model $\Rel$ (Section~\ref{subsec:relational_model}), a proof-relevant
version of it which we call $\PRel$ (Section~\ref{subsec:prrel}), and the
model $\Dist$ of groupoids and distributors
(Section~\ref{subsec:dist}). 

\subsection{The relational model of linear logic}
\label{subsec:relational_model}

We start with the \emph{relational model}, which gives a denotational
interpretation in the category $\Rel$ of sets and relations. A type $A$
is interpreted as a \emph{set} $\intr{A}$, and a program $\vdash M : A$
is interpreted as a subset $\intr{M} \subseteq \intr{A}$. The set
$\intr{A}$ is often called the \emph{web} of $A$,
and we think of its elements as representations of
completed program executions. The subset $\intr{M}$ contains the
 executions that the program $M$ can realize.

\begin{exa}
The ground type for booleans is interpreted as $\intr{\tbool} =
\{\ttrue, \tfalse\}$, and the constant $\vdash \ttrue : \tbool$ as
$\intr{\ttrue} = \{\ttrue\}$.
\end{exa}

The interpretation of a program $M$ is computed compositionally,
following the methodology of denotational semantics, using
the categorical structure we describe below.

\subsubsection{Basic categorical structure}
The category $\Rel$ is defined to have sets as objects, and as
morphisms the \emph{relations} from $A$ to $B$, \emph{i.e.} the subsets $R
\subseteq A \times B$, with the usual notions of identity and
composition for relations.
Its monoidal product is 
the cartesian product
of sets. If $R_i \in \Rel[A_i, B_i]$ for $i = 1, 2$, then the relation
$R_1 \times R_2 \in
\Rel[A_1 \times A_2, B_1 \times B_2]$ is defined to contain the pairs $((a_1,
a_2),
(b_1, b_2))$ with $(a_i, b_i) \in R_i$. The unit $I$ is a fixed
singleton set, say $\{*\}$. This monoidal structure is closed, with
 linear arrow  $A \lin B = A \times B$.

Moreover $\Rel$ has finite cartesian products: the
product of sets $A$ and $B$ is given by the disjoint union
$A + B = \{1\} \times A \uplus \{2\} \times B$, and the empty set is a
terminal object. 

\subsubsection{The exponential modality} The exponential modality of
$\Rel$ is based on \emph{finite multisets}. If $A$ is a set, we write
$\Mf(A)$ for the set of finite multisets on $A$. To denote specific
multisets we use a list-like notation, as in \emph{e.g.} $[0, 1, 1] \in
\Mf(\mathbb{N}).$

Given a set $A$, its \textbf{bang} $\oc A$ is the set $\Mf(A)$.
This extends to a comonad $\oc$ on $\Rel$, satisfying the required
conditions for a model of intuitionistic linear logic: the \emph{Seely isomorphisms}
\[
\Mf(A + B) \iso \Mf(A) \times \Mf(B) \qquad \Mf(\emptyset) \cong I
\]
make $\oc$ a symmetric monoidal functor from $(\Rel, +, \emptyset)$
to $(\Rel, \times, I)$, satisfying a coherence axiom  \cite[\S 7.3]{panorama}.
From this, we obtain that the
Kleisli category $\Rel_{\oc}$ is cartesian closed and thus a model of 
the simply-typed $\lambda$-calculus.

\subsubsection{Conditionals and non-determinism} The category $\Rel_\oc$ also supports
further primitives in a call-by-name setting. 

\begin{exa}\label{ex:prelim_term}
Considering the term $\vdash M : \tbool \to \tbool$ of $\PCF$
\[
\begin{array}{l}
\vdash \lambda
x^\tbool.\,\mathtt{if}\,x\,\mathtt{then}\,x\\
\hspace{50pt}\mathtt{else}\,\mathtt{if}\,x\,\mathtt{then}\,\tfalse\,\mathtt{else}\,\ttrue
: \tbool
\to \tbool\,,
\end{array}
\]
then $\intr{M} = \{([\ttrue, \ttrue], \ttrue), ([\ttrue,
\tfalse], \tfalse), ([\tfalse, \tfalse], \ttrue)\}$, representing the
possible executions given a multiset of values for $x$.
\end{exa}

As a further example, $\Rel$ supports the interpretation of
non-deterministic computation: consider $\vdash \choice : \tbool$ a
non-deterministic primitive that may evaluate to either $\ttrue$ or
$\tfalse$. Then we can set  $\intr{\choice} = \{\ttrue, \tfalse\}$
so that we have
\begin{equation}
\label{eq:nondet}
  \intr{\mathtt{if}\,\choice\,\mathtt{then}\,\ttrue\,\mathtt{else}\,\ttrue}
  = \{\ttrue\}\,.
\end{equation}

The relational model is a cornerstone of static semantics, and 
the foundation of many recent developments in 
denotational semantics \cite{er:point, carv:ex, DBLP:conf/lics/LairdMMP13}.
In this paper we are concerned with its \emph{proof-relevant}
extensions. Roughly speaking, one motivation is to
keep separate different execution paths that lead to the same value, as
with value $\ttrue$ in \eqref{eq:nondet}.

\subsection{Proof-relevant relations}
\label{subsec:prrel}

To showcase this, we consider a notion of proof-relevant relation between sets
(e.g.~\cite{normal-functors, gambino-joyal}).
The idea is to record not only the executions that a program may achieve, but also
the distinct ways in which each execution is realized. We replace relations $\intr{M} \subseteq
A\times B$ with \emph{proof-relevant} relations
\[
\intr{M} : A\times B \to \Set\,, 
\]
so that each point of the web has an associated set of \emph{witnesses}. 
In this model, for instance,
$\intr{\mathtt{if}\,\choice\,\mathtt{then}\,\ttrue\,\mathtt{else}\,\ttrue}(\ttrue)
$ from \eqref{eq:nondet} should be a set $\{*_1, *_2\}$ containing two witnesses,
because there are two possible
paths to the value $\ttrue$. 

Formally, the model is organized as a categorical structure with sets as objects,
functors $\alpha : A \times B \to \Set$ (with $A\times B$ viewed as a
discrete category
) as morphisms, composed with
\begin{eqnarray}
(\beta\circ \alpha)(a, c) &=& \sum_{b \in B} \alpha(a, b) \times \beta(b, c)
\label{eq:comp_prel}
\end{eqnarray}
and with identity morphisms given by $\id_A(a, a') = \{*\}$ if
$a = a'$ and empty otherwise\footnote{A more
standard presentation of this model is via the bicategory of \emph{spans} of
sets, with sets as objects and spans $A \ot S \to B$ as morphisms.}. An important observation is that this
does not form a category. Categorical laws are only
\emph{isomorphisms}, with for instance $(\id_B \circ \alpha)(a, b) =
\alpha(a, b) \times \{*\} \cong \alpha(a, b)$.
We obtain a \emph{bicategory}: a two-dimensional structure incorporating
\emph{$2$-cells} -- morphisms between morphisms -- with 
categorical laws holding only up to coherent invertible $2$-cells.

We call this bicategory $\PRel$. This model shares (in a
bicategorical sense) much of the structure of $\Rel$, and may be used
to interpret \emph{e.g.} the linear $\lambda$-calculus.
We use $\PRel$ as a static collapse of a basic dynamic
model in Section~\ref{subsubsec:plain_oplax}. 

\paragraph*{Limitations of a proof-relevant model without symmetry}
Unfortunately, the finite multiset functor on $\Rel$ does not seem to
extend to $\PRel$. Intuitively, the objective of keeping track of
individual execution witnesses is in tension with the quotient involved
in constructing finite multisets, which blurs out the identity of
individual resource accesses\footnote{In technical terms, the functor $\Mf(-)$ on $\Set$ is not
cartesian -- does not preserve pullbacks -- and so does not preserve
the composition of spans.}. Proof-relevant models that do support an
exponential modality do so by replacing finite multisets with a
categorification, such as finite lists related by explicit permutations
-- \emph{symmetries}.

\subsection{Distributors and generalized species of structures}
\label{subsec:dist}

This categorification yields, among other possible approches, the
bicategory of \emph{distributors}.
\emph{Distributors} are symmetry-aware proof-relevant relations -- here
we consider distributors on 
\emph{groupoids},
\emph{i.e.} small categories in which every morphism is invertible.

\subsubsection{The bicategory of groupoids and distributors}

If $A$ and $B$ are groupoids, a \textbf{distributor} from $A$ to $B$ (also
known as a \emph{profunctor} or \emph{bimodule}) is a functor
\[
\alpha : A^{\op} \times B \to \Set\,.
\]

Thus, for every $a \in A$ and $b \in B$ we have a set $\alpha(a, b)$
of witnesses,
but unlike in $\PRel$ we also have symmetries, in the form of an
action by morphisms in $A$ and $B$.
If $x \in \alpha(a, b)$ and $g \in B(b, b')$, we
write $g \cdot x$ for the functorial action $\alpha(\id, g)(x) \in \alpha(a,
b')$. Similarly, if $f \in A(a', a)$, we write $x\cdot f \in \alpha(a', b)$ for $\alpha(f, \id)$. The actions
must commute, so we can write $g \cdot x \cdot f$ for $(g \cdot x)
\cdot f = g \cdot (x \cdot f) \in \alpha(a', b')$.

We define a bicategory $\Dist$ with groupoids as objects,
distributors as morphisms, and
natural transformations as 2-cells (\cite{yoneda1960ext,benabou1973distributeurs}). 
The \textbf{identity distributor} on $A$ is 
\[
\id_A = A[-, -] : A^\op \times A \to \Set\,,
\]
the hom-set functor.
The \textbf{composition} of two distributors $\alpha : A^\op
\times B \to \Set$ and $\beta : B^\op \times C \to \Set$ is obtained as
a categorified version of \eqref{eq:comp_prel}, defined in terms of a
coend:
\begin{eqnarray}
(\beta \bullet \alpha)(a, c) &=& \int^{b \in B} \alpha(a, b) \times \beta(b,
c)\,.\label{eq:coend}
\end{eqnarray}

Concretely, $(\beta \bullet \alpha)(a, c)$ consists in pairs $(x,
y)$, where $x \in \alpha(a, b)$ and $y \in \beta(b, c)$ for some $b \in B$,
quotiented by $(g \cdot x, y) \sim (x, y \cdot g)$ for $x
\in \alpha(a, b)$, $g \in B(b, b')$ and $y \in \beta(b', c)$ -- by
convention, we use $y \bullet x \in (\beta \bullet \alpha)(a, c)$ as a
notation for the (equivalence class of) the pair $(x, y)$.

The bicategory $\Dist$ has a symmetric monoidal structure given by the cartesian
product $A \times B$ of groupoids, extended pointwise to
distributors. There is a closed structure given by $A \lin B = A^\op
\times B$. 
Finally, $\Dist$ has cartesian products given by the
disjoint union $A + B$ of groupoids. 

\subsubsection{The exponential modality}
In this model with explicit symmetries, the exponential modality is not given
by finite multisets, but instead by finite \emph{lists} with explicit permutations.
\begin{defi}
For a groupoid $A$, there is a groupoid $\Sym(A)$ with as objects the finite
lists $(a_1 \dots a_n)$ of objects of $A$, and as morphisms
$(a_1 \dots a_n) \longrightarrow (a'_1 \dots a'_m)$ the pairs $(\pi, (f_i)_{1\leq
  i \leq n})$, where $\pi : \{1, \dots, n\} \cong \{1, \dots, m\}$ is a
bijection and
$f_i \in A(a_i, a'_{\pi(i)})$ for each $i= 1, \dots, n$.
\end{defi}

More abstractly, $\Sym(A)$ is the \emph{free symmetric (strict)
monoidal category} over $A$.  
This extends to a pseudo-comonad on $\Dist$, where \emph{pseudo} means
that the comonad laws only hold up to coherent invertible 2-cells
\cite{lack2000coherent, fiore2008cartesian}.

\subsubsection{Generalized species of structures}
The Kleisli bicategory $\Dist_\Sym$ is denoted $\Esp$, and the
morphisms in $\Esp$ are called \emph{generalized species of
  structures} \cite{fiore2008cartesian}. Concretely, $\Esp$ has the same objects as
$\Dist$; morphisms are generalized species\footnote{If $A = B = 1$,
  then a generalized species
corresponds to a combinatorial species in the
classical sense of Joyal \cite{joyal1981theorie}. This can be
further generalized to arbitrary small categories $A$ and $B$
\cite{fiore2008cartesian}, but we
do not need this generality.},
defined as distributors from $\Sym(A)$ to $B$; and 2-cells
are natural transformations. Equipped with 
\[
\Sym(A + B) \simeq \Sym(A) \times \Sym(B)  
\qquad\qquad 
\Sym(\emptyset) \simeq I
\]
the Seely \emph{equivalences}, $\Dist$
is a bicategorical model of linear logic. In particular,
the bicategory $\Esp$ is cartesian closed.

\subsubsection{Relationship with $\PRel$}
Distributors conservatively extend the proof-relevant relations of 
Section~\ref{subsec:prrel}: if we regard sets as discrete groupoids, we get  an embedding $\PRel
\hookrightarrow \Dist$ that preserves the symmetric monoidal closed
structure and the cartesian structure.
Explicit symmetries appear essential in defining an exponential
modality in a proof-relevant model: even when $A$ is a discrete groupoid,
$\Sym(A)$ is not discrete. 

\section{Double Categories for Static Models}
\label{sec:double-categ-stat}

The models $\PRel$ and $\Dist$ are bicategories: they have objects,
morphisms, and 2-cells. This 2-dimensional structure is essential,
because the laws for composition are only satisfied up to invertible
2-cells. Although $\Rel$ is a category, we can regard it as a
degenerate bicategory where the laws for composition hold
strictly, and with 2-cells given by the usual order on relations: if $R, R' \in \Rel(A, B)$, then we have a 2-cell
$R \Rightarrow R'$ whenever $R \subseteq R'$ as subsets of
$A \times B$. (This is a proof-irrelevant counterpart of 
the 2-dimensional structure in $\PRel$ and $\Dist$.)

The point of this section is to establish that each of these
bicategories arises as a fragment of a larger categorical structure
known as a pseudo double category, which consists of objects, two
kinds of morphisms (known as horizontal and vertical), and
2-cells. Horizontal morphisms compose weakly (as in a bicategory) and
vertical morphisms compose strictly (as in a category).
The motivation for considering this more general structure is that,
for many bicategories, there exists a simpler class of morphisms
between the objects for which composition is strictly associative and
unital. This can be exploited to simplify certain structural proofs,
as we will do. 

We first give the definition of a pseudo double category. We usually
drop the adjective ``pseudo'' in the paper.  Pseudo-double categories are originally due to
Ehresmann \cite{ehresmann1963categories}.  For a detailed reference
on double category theory that includes monoidal double categories,
functors between them, and monads, see 
\cite[Sections~2--5]{ggv24}. For the theory of double
monads, see also \cite{cruttwell2010unified}.

\begin{defi}
  \label{def:doublecat}
  A \textbf{(pseudo) double category} (e.g.~\cite[Def.~2.1]{ggv24}) consists of:
  \begin{itemize}
  \item a class of objects;
  \item \textbf{vertical morphisms} between the objects,
    typically displayed as $A \to B$, with
    the structure of a category (that is, identity morphisms, and a composition
    operation satisfying axioms);
  \item \textbf{horizontal morphisms} between the objects,
    typically displayed as $A \profto B$, with identity morphisms and
    a composition operation;
  \item square-shaped 2-cells between pairs of vertical morphisms and
    pairs of horizontal morphisms, as follows:
    \[
\begin{tikzcd}
	A & B \\
	C & D
	\arrow[""{name=0, anchor=center, inner sep=0}, "R", "\shortmid"{marking}, from=1-1, to=1-2]
	\arrow["f"', from=1-1, to=2-1]
	\arrow["g", from=1-2, to=2-2]
	\arrow[""{name=1, anchor=center, inner sep=0}, "S"', "\shortmid"{marking}, from=2-1, to=2-2]
	\arrow[shorten <=12pt, shorten >=12pt, Rightarrow, from=0, to=1]
\end{tikzcd}
\]
which can be composed horizontally and vertically. When the vertical
boundaries $f$ and $g$ are identity morphisms, a 2-cell is
called \textbf{globular}. 
\end{itemize}
\noindent 
Horizontal morphisms can be composed, but associativity and unit laws only
hold weakly, just like in a
bicategory: up to coherent, invertible, globular 2-cells, whose
names we omit.  
\end{defi}

An immediate observation is that a pseudo double category $\mathbb{D}$ always has an
underlying bicategory over the same objects, consisting of horizontal morphisms
and globular 2-cells. We call this the \textbf{horizontal
  bicategory} $\hori(\mathbb{D})$. It is well-known that our
bicategories $\Rel, \PRel$ and $\Dist$ all arise in this way from
pseudo double categories with a natural class of vertical morphisms, as
we explain now. 

\subsection{Examples: static models as double categories.} To define a
double category from a bicategory, we must define a class of vertical
morphisms and an appropriate notion of 2-cells. In $\Rel$ and
$\PRel$, the objects are sets, and the vertical morphisms are taken to
be functions. In $\Dist$, the objects are groupoids, and the vertical
morphisms will be functors.

\begin{exa}[The double category of relations] The double category
  $\dblRel$ has the following components:
  \begin{itemize}
  \item objects: sets;
  \item vertical morphisms: functions;
  \item horizontal morphisms: relations; 
  \item and a 2-cell of type \[ \begin{tikzcd}
	A & B \\
	C & D
	\arrow[""{name=0, anchor=center, inner sep=0}, "R", "\shortmid"{marking}, from=1-1, to=1-2]
	\arrow["f"', from=1-1, to=2-1]
	\arrow["g", from=1-2, to=2-2]
	\arrow[""{name=1, anchor=center, inner sep=0}, "S"', "\shortmid"{marking}, from=2-1, to=2-2]
	\arrow[shorten <=12pt, shorten >=12pt, Rightarrow, from=0, to=1]
      \end{tikzcd}\]
    exists when $(a, b) \in R$ implies $(f(a), g(b)) \in S$ for every
    $a, b$. 
  \end{itemize}
\end{exa}
\noindent 
When $f$ and $g$ are identity functions, there
is a 2-cell as above just when $R \subseteq S$. So the horizontal
bicategory $\hori(\dblRel)$ is exactly the bicategory $\Rel$ of sets
and relations. 

We now turn to proof-relevant relations and distributors. Since $\PRel$ can be regarded as
the full sub-bicategory of $\Dist$ over the discrete groupoids, we
start with the more general model. 

\begin{exa}[The double category of groupoids and distributors] 
 The double category
  $\dblDist$ has the following components:
  \begin{itemize}
  \item objects: groupoids;
  \item vertical morphisms $A \to B$: functors;
  \item horizontal morphisms $A \profto B$: distributors $A^\op \times
    B \to \Set$; 
  \item 2-cells of type \[ \begin{tikzcd}
	A & B \\
	C & D
	\arrow[""{name=0, anchor=center, inner sep=0}, "\alpha", "\shortmid"{marking}, from=1-1, to=1-2]
	\arrow["F"', from=1-1, to=2-1]
	\arrow["G", from=1-2, to=2-2]
	\arrow[""{name=1, anchor=center, inner sep=0}, "\beta"', "\shortmid"{marking}, from=2-1, to=2-2]
	\arrow[shorten <=12pt, shorten >=12pt, Rightarrow, from=0, to=1]
      \end{tikzcd}\]
    are natural transformations $\alpha(a, b) \tto \beta(F(a), G(b))$
    between functors $A^\op \times B \to \Set$. 
  \end{itemize}
\end{exa}
\noindent 
It is clear that $\hori(\dblDist) = \Dist$. We can consider the
discrete sub-model, with no symmetries:

\begin{exa}[The double category of proof-relevant relations]
  If we restrict the objects in $\dblDist$ to discrete skeletal
  groupoids (\emph{i.e.}~sets), we obtain a double category called
  $\dblPRel$. Up to isomorphism, this double
  category has sets as objects, and functions as vertical
  morphisms. Then $\hori(\dblPRel)$ corresponds to our bicategory $\PRel$. 
\end{exa}

A \textbf{double functor} between double categories defines an
appropriate mapping of objects, vertical and horizontal morphisms, and
2-cells, with coherence data similar to that for a functor of
bicategories. Note that a double functor can be lax, oplax, or pseudo,
depending on the nature of the globular compositor 2-cell. Between double functors, there are two possible kinds of
natural transformations: vertical and horizontal \cite[Def.~3.4 and
Def.~3.5]{ggv24}.

\subsection{Companions and conjoints.} Our three examples share a
common feature: a vertical morphism always induces a pair of
horizontal morphisms. For example, in $\dblRel$, a function
$f : A \to B$ induces relations
$\companion{f} = \{ (a, f(a)) \mid a \in A \} : A \profto B$ and
$\conjoint{f} = \{ (f(a), a) \mid a \in A\} : B \profto A$. In
$\dblDist$ (and so also in $\dblPRel$) we have a similar construction:
if $A$ and $B$ are groupoids and $F : A \to B$ is a functor, there are
distributors $\companion{F} : A \profto B$ and
$\conjoint{F} : B \profto A$, defined below
\begin{align*}
  \companion{F} &: A^\op \times B \to \Set   &   \conjoint{F} &: B^\op
                                                  \times A \to \Set \\
  &(a, b) \mapsto B[F(a), b] & & (b, a) \mapsto B[b, F(a)]
\end{align*}

These constructions are instances of an abstract notion in double
category theory.
\begin{defi}
Let $\mathbb{D}$ be a pseudo double category and let $S : A \to B$ be
a vertical morphism. A \textbf{companion} of $S$ is a horizontal
morphism $\widehat{S} : A \profto B$ together with a pair of 2-cells
\[
\begin{tikzcd}
	A & B \\
	B & B
	\arrow[""{name=0, anchor=center, inner sep=0},
        "{\companion{F}}", from=1-1, to=1-2, "\shortmid"{marking}]
	\arrow["F"', from=1-1, to=2-1]
	\arrow[Rightarrow, no head, from=1-2, to=2-2]
	\arrow[""{name=1, anchor=center, inner sep=0}, Rightarrow, no head, from=2-1, to=2-2]
	\arrow[shorten <=12pt, shorten >=12pt, Rightarrow, from=0, to=1]
      \end{tikzcd}\qquad
      \begin{tikzcd}
	A & A \\
	A & B
	\arrow[""{name=0, anchor=center, inner sep=0}, Rightarrow, no head, from=1-1, to=1-2]
	\arrow[Rightarrow, no head, from=1-1, to=2-1]
	\arrow["{F}", from=1-2, to=2-2]
	\arrow[""{name=1, anchor=center, inner sep=0}, "\companion{F}"', "\shortmid"{marking}, from=2-1, to=2-2]
	\arrow[shorten <=12pt, shorten >=12pt, Rightarrow, from=0, to=1]
\end{tikzcd}
\]
satisfying appropriate axioms (e.g. \cite[Def.~2.6]{ggv24}). A \textbf{conjoint} of $S$ is
defined dually, as a horizontal morphism $\widecheck{F} : B \profto A$
with corresponding data and axioms. 
\end{defi}
Any two companions of $S$, and any two conjoints, are canonically
isomorphic. The constructions we gave for relations and profunctors
provide companions and conjoints, and in this case the accompanying 2-cells are
trivial. 

\begin{lem}
In $\dblRel$, $\dblPRel$, and $\dblDist$, all vertical morphisms have
companions and conjoints.
\end{lem}

A pseudo double category with this property is often directly
presented as a \emph{proarrow equipment} \cite{wood1982abstract}. We do
not use equipments in this paper, because the double category of games
we present in the next section only admits companions and conjoints
for a subclass of vertical morphisms, and so we must discuss these
notions explicitly. 

Note that a functor $F : A \to B$ also determines a pair
 of species $\widehat{F} \in \Esp[A, B]$ and $\widecheck{F} \in
 \Esp[B, A]$, respectively defined by
 \begin{align*}
& \companion{F}((a_1, \dots, a_n), b) = \Sym (B) [(F(a_1), \dots, F(a_n)),
   (b)] \\
&   \conjoint{F}((b_1, \dots, b_n), a) = \Sym (B) [(b_1,
   \dots, b_n), (F(a))].
 \end{align*}
 This construction gives companions and conjoints for $F$ in an
 alternative double category consisting of groupoids, functors, and
 generalized species. We omit the details of this double category,
 which we do not use in the paper. (The notations $\companion{F}$ and
 $\conjoint{F}$ refer ambiguously to the induced distributors or to
 the induced species, but this will be clear from context.)
 
\subsection{Symmetric monoidal structure in double categories.} In
this paper we primarily exploit the double-categorical structure of
our models to simplify the description of symmetric monoidal structure
in the underlying bicategories. This technique is well-established \cite{shulman}. Symmetric monoidal double categories have a much
simpler definition than symmetric monoidal bicategories, because we
can use vertical morphisms to encode the associativity, unit, and
symmetry constraints of the monoidal product. Since vertical morphisms
form a strict category, we can use strict categorical notions,
avoiding the numerous coherence axioms that are needed when defining a
symmetric monoidal bicategory.

\begin{defi} A \textbf{symmetric monoidal double category} is a double
  category $\mathbb{D}$ equipped with a double functor
  $\otimes : \mathbb{D} \times \mathbb{D} \to \mathbb{D}$, an object
  $I$, and vertical double natural transformations for symmetry,
  associativity, and unitality: in particular there are vertical 1-cells
  \begin{align*}
    (A \otimes B) \otimes C \xrightarrow{a_{A, B, C}} A \otimes (B \otimes C)
    & & A \otimes I \xrightarrow{r_A} A & & I \otimes A
\xrightarrow{l_A} A & & A \otimes B \xrightarrow{b_{A, B}} B \otimes A
  \end{align*}
  satisfying a number of coherence axioms (see
  e.g.~\cite[Def.~4.1]{ggv24} for a fully explicit definition). 
\end{defi}

In general, if a double category $\mathbb{D}$ is symmetric
monoidal, then the underlying horizontal bicategory
$\hori(\mathbb{D})$ has no reason to be symmetric monoidal, since the
monoidal data was defined only at the vertical level, and the
horizontal bicategory has no access to it. But, when
the structural isomorphisms have horizontal companions, the
symmetric monoidal structure can be lifted the horizontal level. This is due to Shulman and Wester-Hansen \cite[Thm~1.2]{hansen-shulman}:

\begin{thmC}[\cite{hansen-shulman}]
  \label{thm:shulman}
  Let $\mathbb{D}$ be a symmetric monoidal double category, and
  suppose that in $\mathbb{D}$ all vertical isomorphisms have
  companions. Then the bicategory $\hori(\mathbb{D})$ is symmetric
  monoidal in a canonical way.

  Furthermore, if $\mathbb{E}$ is also a
  symmetric monoidal double category where all vertical isomorphisms
  have companions, then every
  symmetric monoidal pseudo double functor $F : \mathbb{D} \to
  \mathbb{E}$ induces a symmetric  monoidal pseudofunctor  $\hori(F) :
  \hori(\mathbb{D}) \to \hori(\mathbb{E})$. 
\end{thmC}

We can then directly verify the following property:
\begin{lem}
The double categories $\dblRel$, $\dblPRel$, and $\dblDist$ are
symmetric monoidal, and the induced symmetric monoidal structure on
$\Rel$, $\PRel$, and $\Dist$ corresponds to that given in Section~\ref{sec:static}. 
\end{lem}

\subsection{Another application of companions and conjoints in
  $\dblDist$.}
  \label{sec:anoth-appl-comp}

In the rest of this section, we discuss an elementary construction on
distributors which will be useful in Section~\ref{sec:cc_pseudofunctor}.

\begin{defi}
  \label{def:pierres-construction}
  For a distributor $\alpha : A \profto B$, \emph{i.e.} a functor
$\alpha : A^{\op} \times B \to \Set$, and functors $S : A' \to A$ and
$T : B' \to B$, we define distributors $\alpha[S] : A' \profto B$ and
$[T]\alpha : A \profto B'$ by
$
\alpha[S](a, b) = \alpha(Sa, b)$
and $[T]\alpha(a, b) = \alpha(a, Tb).$
\end{defi}

This construction extends in the obvious way to functors between hom-categories:
\[
-[S] : \Dist[A, B] \to \Dist[A', B]\,,
\qquad
\,[T] - : \Dist[A, B] \to \Dist[A, B']\,.
\]

The distributors $\alpha[S]$ and $[T]\beta$ can be presented using
companions and conjoints:
\begin{lem}
  \label{lem:action-with-companions}
  For $\alpha : A \profto B$, $S : A' \to A$ and $T : B'
  \to B$, there are natural isomorphisms
\[
\alpha[S] \cong A' \overset{\companion{S}}{\profto} A
\overset{\alpha}{\profto} B
\qquad\text{and}\qquad
[T]\alpha = A \overset{\alpha}{\profto} B \overset{\conjoint{T}}{\profto} B'.
\]
\end{lem}
\begin{proof}
The first isomorphism is because $\int^{a \in A} \alpha(a, b) \times A[Sa',a] \cong \alpha(Sa', b)$ by the density formula for coends. The
second one is similar. 
\end{proof}

We now introduce a few lemmas expressing compatibility of these operations
with other constructions on distributors. 
\begin{lem}
Consider distributors $\alpha \in \Dist[A, B]$, $\beta \in \Dist[B,
C]$ and functors $S : A'
\to A$, $T : C' \to C$.

Then we have natural isomorphisms, additionally natural in $S$ and $T$:
\[
(\beta \bullet \alpha)[S] \iso \beta \bullet \alpha[S]\,,
\qquad
\qquad
\,[T] (\beta \bullet \alpha) \iso [T]\beta \bullet \alpha\,.
\]
\end{lem}
\begin{proof}
 Immediate from Lemma~\ref{lem:action-with-companions}, and by associativity
 and naturality of composition.
\end{proof}

\begin{lem}\label{lem:comp_ren_comp}\label{lem:comp_ren_id}
Consider $\alpha \in \Dist[A, B]$ and $\beta \in \Dist[B', C]$ distributors,
with an adjunction $\e$ consisting of $L : B \dashv B' : R$ and natural unit
and counit $\eta_b \in B[b, RLb]$ and $\epsilon_{b'} \in B[LRb', b']$. Then, we have a natural isomorphism
\[
\beta [L] \bullet \alpha \iso \beta \bullet [R] \alpha\,
\]
whose action is as follows for $t \bullet s \in (\beta[L] \bullet
\alpha)(a, c)$,
\begin{align*}
  \xi_{\alpha, \beta, \e}(t \bullet s) = t \bullet \alpha(a,\eta_b)(s) \in
  (\beta \bullet [R]\alpha)(a, c),
\end{align*}
and as follows for $t' \bullet s' \in (\beta \bullet [R]\alpha)(a, c)$,
\begin{align*}
\xi^{-1}_{\alpha, \beta, \e}(t' \bullet s') = \beta(\epsilon_{b'}, c)(t')
\bullet s' \in (\beta[L] \bullet \alpha)(a, c)\,.
\end{align*}

This induces a natural isomorphism
\[\chi_{\e} : \id_{B'} [L] \iso [R]\,\id_B \in \Dist[B, B']\,\]
as a special case of the above, using the composition laws for distributors. 
\end{lem}
\begin{proof}
From the adjunction property it is easy to check that $\companion{L} \cong
\conjoint{R}$, and the result follows by associativity of
composition. 
\end{proof}

\section{Concurrent Games and Static Collapse}\label{sec:collapse_prel}

We now construct a dynamic model based on concurrent games and
strategies, without symmetries. We show that it has a static double
categorical collapse in the model of proof-relevant relations
introduced in Section~\ref{subsec:prrel}. 

\subsection{Rudiments of concurrent games} 
\label{subsec:basic-cg}

Game semantics presents computation in terms of a two-player game:
\emph{Player} plays for the program under scrutiny, while
\emph{Opponent} plays for the execution environment. So a program is
interpreted as a \emph{strategy} for Player, and this strategy
is constrained by a notion of \emph{game}, specified
by the type. 
The framework of \emph{concurrent games}
(\cite{DBLP:conf/concur/MelliesM07,DBLP:conf/tlca/FaggianP09,DBLP:conf/lics/RideauW11})
is not merely a game semantics for concurrency, but a deep
reworking of the basic mechanisms of game semantics using causal
``truly concurrent'' structures from concurrency theory
\cite{DBLP:conf/scc/NielsenPW79}.

\subsubsection{Event structures} Concurrent games and
strategies are based on event structures. An event structure
represents the behaviour of a system
as a set of possible computational events equipped with
dependency and incompatibility constraints.

\begin{defi}\label{def:es}
An \textbf{event structure (es)} is $E = (\ev{E}, {\leq_E}, {\conflict_E})$,
where $\ev{E}$ is a (countable) set of \textbf{events}, $\leq_E$ is a
partial order called \textbf{causal dependency} and $\conflict_E$ is an
irreflexive symmetric binary relation on $\ev{E}$ called
\textbf{conflict}, satisfying:
\[
\begin{array}{rl}
\text{\emph{(1)}}& \forall e\in \ev{E},~[e]_E = \{e'\in
\ev{E} \mid
e'\leq_E e\}~\text{is finite,}\\
\text{\emph{(2)}}&
\forall e_1 \conflict_E e_2,~\forall e'_2 \geq_E e_2,~e_1 \conflict_E
e'_2\,.
\end{array}
\]
\end{defi}
Operationally, an event can occur if \emph{all} its
dependencies are met, and \emph{no} conflicting events 
have occurred. A finite set $x \subseteq_f \ev{E}$
down-closed for $\leq_E$ and comprising no conflicting pair
is called a \textbf{configuration} -- we write $\conf{E}$ for the set
of configurations on $E$, naturally ordered by inclusion. If $x \in
\conf{E}$ and $e \in \ev{E}$ is such that $e \not \in x$ but $x \cup
\{e\} \in \conf{E}$, we say that $e$ is \textbf{enabled} by $x$ and
write $x \enb_E e$. For $e_1, e_2 \in \ev{E}$ we write $e_1 \imc_E e_2$
for the \textbf{immediate causal dependency}, \emph{i.e.} $e_1 <_E
e_2$ with no event strictly in between.

There is an accompanying notion of \emph{map}: a \textbf{map of event
structures} from $E$ to $F$ is a function $f : \ev{E} \to \ev{F}$ such
that: \emph{(1)} for all $x \in \conf{E}$, the direct image $f x \in
\conf{F}$; and \emph{(2)} for all $x \in \conf{E}$ and $e, e' \in x$,
if $f e = f e'$ then $e = e'$. There is a category $\ES$ of event
structures and maps.

\subsubsection{Games and strategies} Throughout this paper, we will 
gradually refine our notion of game. For now, a \textbf{plain game} is
simply an event structure $A$ together with a \textbf{polarity}
function $\pol_A : \ev{A} \to \{-, +\}$ which specifies, for each event
$a \in A$, whether it is \textbf{positive} (\emph{i.e.} due to Player /
the program) or \textbf{negative} (\emph{i.e.} due to Opponent / the
environment). Events are often called \textbf{moves}, and annotated
with their polarity (as in $a^-, a^+$).

A strategy is an event structure with a projection map to $A$:
\begin{defi}\label{def:plain_strategy}
Consider $A$ a plain game. A \textbf{strategy} on $A$, written $\sigma
: A$, is an event structure $\sigma$ together with a map $\pr_\sigma : \sigma
\to A$ called the \textbf{display map}, satisfying:
\[
\begin{array}{rl}
\text{\emph{(1)}} &
\text{for all $x\in \conf{\sigma}$ and $\pr_\sigma x \enb_A a^-$,
there is a unique $x \enb_\sigma s$ such that $\pr_\sigma s =
a$.}\\
\text{\emph{(2)}} &
\text{for all $s_1 \imc_\sigma s_2$, if $\pol_A(\pr_\sigma(s_1)) = +$
or $\pol_A(\pr_\sigma(s_2)) = -$, then $\pr_\sigma(s_1) \imc_A
\pr_\sigma(s_2)$.}
\end{array}
\]
\end{defi}
\noindent 
Informally, the two conditions (called \emph{receptivity} and \emph{courtesy})
ensure that the strategy does not constrain the behavior of Opponent
any more than the game does. They are essential
for the compositional structure we describe below, but they do not
play a major role in this paper. 

As a simple example, the usual game $\tbool$ for booleans in
call-by-name is
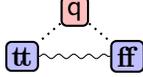
\[
  \begin{tikzpicture}[  baseline=(current bounding box.center)]

  \node[negnode] (qm) at (0,0.6) {$\qu$};
  \node[posnode] (t) at (-0.7,0) {$\ttrue$};
  \node[posnode] (f) at (0.7,0) {$\tfalse$};
  \draw[thick, dotted] (qm) -- (t);
  \draw[thick, dotted] (qm) -- (f);
  \draw[conflict] (t) -- (f);
\end{tikzpicture}
\]
drawn following the order from top to bottom, with the wiggly line
indicating conflict. Player moves are blue, and Opponent moves are red
-- Opponent initiates computation with the first move $\qu$, to
which Player can react with either $\ttrue$ or $\tfalse$. 

Just like $\PRel$, strategies give a ``proof-relevant'' account of
execution, in the sense that moves and configurations of the game can
have multiple witnesses in the strategy. For example, on the left
below, $b$ and $c$ are mapped to the same move $\ttrue$:
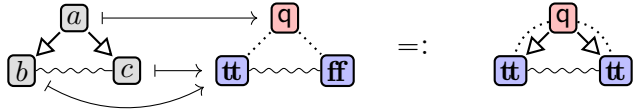
\[
  \begin{tikzpicture}
  \node[neutralnode] (a) at (0,0.7) {$a$};
  \node[neutralnode] (b) at (-0.7,0) {$b$};
  \node[neutralnode] (c) at (0.7,0) {$c$};
  \node[negnode] (qm) at (2.8,0.7) {$\qu$};
  \node[posnode] (t) at (2.1,0) {$\ttrue$};
  \node[posnode] (f) at (3.5, 0) {$\tfalse$};
  
  \draw[strat-causality] (a) -- (b);
  \draw[strat-causality] (a) -- (c);
  \draw[|->, shorten >=5pt, shorten <=5pt] (a) to (qm);
  \draw[conflict] (b) -- (c);
  \draw[|->, shorten >=5pt, shorten <=5pt] (b) to[bend right] (t);
  \draw[|->, shorten >=5pt, shorten <=5pt] (c) to (t);
  \draw[game-causality] (qm) -- (t);
  \draw[game-causality] (qm) -- (f);
  \draw[conflict] (t) -- (f);

  \node (eq) at (4.5, 0.3) {\large$ =\vcentcolon$};
  
  \node[negnode] (qm2) at (6.5, 0.7) {$\qu$};
  \node[posnode] (t2) at (5.8,0) {$\ttrue$};
  \node[posnode] (t2') at (7.2, 0) {$\ttrue$};

  \draw[strat-causality] (qm2) -- (t2);
  \draw[strat-causality] (qm2) -- (t2');
  \draw[game-causality] (qm2) to[bend right] (t2);
  \draw[game-causality] (qm2) to[bend left] (t2');
  \draw[conflict] (t2) -- (t2');
\end{tikzpicture}
\]

Note that we denote immediate causality by
$\imc$ in strategies, while we use dotted lines for games.
This
lets us represent the strategy
in a single diagram, as on the right above. 

\subsubsection{Morphisms between strategies}
\label{subsubsec:mor_strat}
For $\sigma$ and $\tau$ two strategies on $A$, a \textbf{morphism} from $\sigma$ to
$\tau$, written $f : \sigma \Rightarrow \tau$, is a map of event
structures $f : \sigma \to \tau$ preserving the dependency relation $\leq$
(we say it is \textbf{rigid}) 
and such that $\pr_\tau \circ f = \pr_\sigma$.

\subsubsection{+-covered configurations}
 We now describe a useful technical tool: a strategy is completely characterized by a subset of its
configurations, called $+$-covered. 

For a strategy $\sigma$ on a game $A$, a configuration $x \in \conf{\sigma}$ is \textbf{$+$-covered} if all its
maximal events are positive, so every Opponent move has at least one
Player successor. We write $\confp{\sigma}$ for
the partial order of $+$-covered configurations
of $\sigma$. Those are fairly important, because morphisms between
strategies may be entirely described through their action on
$+$-covered configurations:

\begin{lem}\label{lem:pcov_mapification}
Consider $\sigma, \tau : A$ two strategies on a plain game $A$. Assume there is a
function
\[
f : \confp{\sigma} \to \confp{\tau}
\]
compatible with display maps and preserving unions.
Then, there is a unique morphism of strategies $\hat{f} : \sigma
\Rightarrow \tau$ such that for all $x\in \confp{\sigma}$,
$\hat{f}\,x = f\,x$.
\end{lem}

This is \cite[Lemma 6.3.4]{hdr}. We also mention the immediate
consequence:

\begin{lem}\label{lem:iso_confp}
 Consider a plain game $A$, and strategies $\sigma, \tau : A$. 

If $f : \confp{\sigma} \iso \confp{\tau}$ is an order-isomorphism such
that $\pr_\tau \circ f = \pr_\sigma$, then there is a unique
isomorphism of strategies $\hat{f} :
\sigma \iso \tau$ such that for all $x \in
\confp{\sigma}$, $\hat{f}(x) = f(x)$.
\end{lem}

\subsection{A double category of concurrent games and strategies}
We construct a double category whose objects are plain games,
horizontal morphisms are strategies between games, vertical morphisms
are maps of games, and 2-cells are morphisms of strategies. The main
technical point is the composition of strategies.  

\subsubsection{Strategies between games}
\label{subsec:strat-betw-games}

If $A$ is a plain game, its
\textbf{dual} $A^\perp$ has the same components as $A$ except for
the reversed polarity. In particular $\conf{A} = \conf{A^\perp}$. The
\textbf{tensor} $A \tensor B$ of $A$ and $B$ is simply $A$ and $B$ side
by side, with no interaction -- its events are the tagged disjoint
union $\ev{A \tensor B} = \ev{A} + \ev{B} = \{1\} \times \ev{A} \uplus
\{2\} \times \ev{B}$, and other components are inherited. 
We write $x_A \tensor x_B$ for the configuration of $A \tensor B$
that has $x_A \in \conf{A}$ on the left and $x_B \in \conf{B}$ on the
right, informing an order-isomorphism
\begin{eqnarray}
- \tensor - \quad:\quad \conf{A} \times \conf{B} &\iso& \conf{A\tensor B}\,.
\label{eq:isotensor}
\end{eqnarray}

Finally, the \textbf{hom} $A\vdash B$ is $A^\perp \tensor B$. As above, its
configurations are denoted $x_A \vdash x_B$ for $x_A \in \conf{A}$ and $x_B \in
\conf{B}$.
\begin{defi}
A \textbf{strategy from $A$ to $B$} is a strategy on the game $A \vdash B$. If $\sigma : A \vdash B$ and
$x^\sigma \in \conf{\sigma}$, by convention we write
$\pr_\sigma(x^\sigma) = x^\sigma_A \vdash x^\sigma_B \in \conf{A\vdash
B}$. 
\end{defi}
Our first example of a strategy between games is the \textbf{copycat}
strategy $\cc_A$, which is the identity morphism on $A$ in our
bicategory, \emph{i.e.} the horizontal identity.
Concretely, copycat on $A$ has the same events as $A \vdash A$,
but adds immediate causal links between copies of the same
move across components, from the negative copy to the positive.
By Lemma~\ref{lem:iso_confp}, the following characterizes copycat up to isomorphism.
\begin{prop}\label{prop:confp_cc}
If $A$ is a game, there is an order-isomorphism
\[
\cc_{(-)} : \conf{A} \iso \confp{\cc_A}
\]
such that for all $x \in \conf{A}$, $\pr_{\cc_A}(\cc_x) = x \vdash
x$. 
\end{prop}
\begin{proof}
Follows from \cite[Lemma 6.4.4]{hdr}. 
\end{proof}

\subsubsection{Composition} \label{subsubsec:cg_comp}
Consider $\sigma : A \vdash B$ and $\tau :
B \vdash C$. We define their composition $\tau \odot \sigma :
A \vdash C$. This is a dynamic model, and to successfully synchronize,
$\sigma$ and $\tau$ must agree to play the same events \emph{in the same
order}; this is defined in two steps.

We say that configurations $x^\sigma \in \conf{\sigma}$ and $x^\tau
\in \conf{\tau}$ are \textbf{matching} if they reach the same
configuration on $B$, \emph{i.e.} $x^\sigma_B =
x^\tau_B = x_B$. If that is the case, it induces a synchronization (and
we may then ask if that synchronization induces a deadlock).
If all events of
$x^\sigma$ and $x^\tau$ were in $B$, this synchronization would take the form of a
bijection $x^\sigma \bij x^\tau$. But some moves of $x^\sigma$ are in
$A$ and some moves of $x^\tau$ are in $C$, so instead we form the
bijection
\[
\varphi[x^\sigma, x^\tau] : x^\sigma \parallel x^\tau_C 
\stackrel{\pr_\sigma \parallel x^\tau_C}{\bij}
x^\sigma_A \parallel x_B \parallel x^\tau_C 
\stackrel{x^\sigma_A \parallel \pr_\tau^{-1}}{\bij}
x^\sigma_A \parallel x^\tau
\]
where $x\parallel y$ is the tagged disjoint union.
This uses the fact that from the conditions on maps of event
structures,
$\pr_\sigma : x^\sigma \bij x^\sigma_A \vdash x^\sigma_B$ is a
bijection and likewise for $\pr_\tau$.

Now that the synchronization is formed, we import the causal
constraints of $\sigma$ and $\tau$ to (the graph of)
$\varphi[x^\sigma, x^\tau]$, via (with $m, m' \in x^\sigma \parallel
x^\tau_C$ and $n, n' \in x^\sigma_A \parallel x^\tau$):
\[
\begin{array}{rcl}
(m, n) \cleq_\sigma (m', n') &\Leftrightarrow& 
m <_{\sigma \parallel C} m'\\
(m, n) \cleq_\tau (m', n') &\Leftrightarrow& 
n <_{A \parallel \tau} n'
\end{array}
\]
letting us finally say that matching $x^\sigma$ and $x^\tau$ are
\textbf{causally compatible} if ${\cleq} = {\cleq_\sigma \cup
  \cleq_\tau}$ on (the graph of) $\varphi[x^\sigma, x^\tau]$ is
acyclic.
In particular,  $x^\sigma$ and $x^\tau$ in Figure \ref{fig:deadlock}
are
\emph{not} causally compatible, the synchronization induces a
\emph{deadlock}.

\begin{figure}
\[
\begin{tikzpicture}
\matrix (m) [matrix of nodes, column sep=0pt, row sep=5pt] {
  \node (a) {$\tunit$}; & \node (lin) {$\lin$}; & \node (b)
  {$\tunit$}; \\[14pt]
  & & \node[negnode] (c) {$\qu$}; \\
  \node[posnode] (d) {$\qu$}; & & \\[10pt]
  \node[negnode] (e) {$\done$}; & &\\
&   & \node[posnode] (f) {$\done$}; \\
};

\draw[strat-causality] (c) to (d);
\draw[strat-causality] (d) to (e);
\draw[strat-causality] (e) to (f);
\draw[game-causality] (c) to[bend right] (f);
\draw[game-causality] (d) to[bend right] (e);
\draw[game-causality] (c) to[bend right] (d);
    \begin{pgfonlayer}{background}
    \draw[background rectangle,rounded corners=3pt, draw=none]
(a.west) -- (a.north west) -- (b.north east) -- (b.south east) --
(a.south west) -- (a.west);
     \end{pgfonlayer}
\end{tikzpicture}
\quad
\quad
\begin{tikzpicture}
  \node (a) at (0, 4.2) {$ \tunit$};
  \node (lin) at (0.5, 4.2) {$\lin$};
  \node (b) at (1, 4.2) {$\tunit$};
  \node (hom) at (1.5, 4.2) {$\vdash$};
  \node (c) at (2, 4.2) {$\tnat $};
  \node[negnode] (d) at (2, 3.5) {$\qu$};
  \node[posnode] (e) at (1, 3) {$\qu$}; 
  \node[negnode] (f) at (0, 2.5) {$\qu$};
  \node[posnode] (g) at (0, 1.2){$\done$}; 
  \node[negnode] (h) at (1, 1.6) {$\done$}; 
  \node[posnode] (i) at (2, 1){$0$}; 

\draw[strat-causality] (d) -- (e);
\draw[strat-causality] (e) -- (f);
\draw[strat-causality] (f) -- (g);
\draw[strat-causality] (e) -- (h);
\draw[strat-causality] (h) -- (g);
\draw[strat-causality] (h) -- (i);

\draw[game-causality] (d) to[bend left] (i);
\draw[game-causality] (e) to[bend right] (f);
\draw[game-causality] (f) to[bend right] (g);
\draw[game-causality] (e) to[bend right] (h);

    \begin{pgfonlayer}{background}
    \draw[background rectangle,rounded corners=3pt, draw=none]
(a.west) -- (a.north west) -- (c.north east) -- (c.south east) --
(a.south west) -- (a.west);
     \end{pgfonlayer}
\end{tikzpicture}
\]
\caption{An example of matching but causally incompatible
  configurations, in the composition of $\sigma :
  \tunit \lin \tunit$ and $\tau : \tunit \lin \tunit \vdash
  \tnat$. The underlying games are left undefined, but can be recovered
  by removing the arrows $\imc$. The configurations
  are matching on $\tunit \lin \tunit$,  but the arrows $\imc$ impose
  incompatible orders (i.e.~a cycle) between the two occurrences of $\done$. 
}
\label{fig:deadlock}
\end{figure}

The \textbf{composition} of $\sigma$ and $\tau$ is the unique (up to
iso) strategy
whose $+$-covered configurations are essentially causally compatible pairs of $+$-covered
configurations. Write $\CC(\sigma, \tau)$ for the set of causally
compatible pairs $(x^\sigma, x^\tau)
 \in \confp{\sigma} \times \confp{\tau}$, ordered componentwise.   

 \begin{prop}\label{prop:char_comp}
Consider strategies $\sigma : A \vdash B$ and $\tau : B \vdash C$.

There is a strategy $\tau \odot \sigma : A \vdash C$, unique up
to isomorphism, with an order-isomorphism
\[
\begin{array}{rcrcl}
- \odot -
&\!\!:\!\!&
\CC(\sigma, \tau) \iso \confp{\tau \odot \sigma}
\end{array}
\]
s.t. for all $x^\sigma \in \confp{\sigma}$ and $x^\tau \in
\confp{\tau}$ causally compatible,
\[
\pr_{\tau\odot \sigma}(x^\tau \odot x^\sigma) =
x^\sigma_A \vdash x^\tau_C\,.
\]
\end{prop}
\begin{proof}
See \cite[Proposition 6.2.1]{hdr}.
\end{proof}

This description of composition emphasizes 
the conceptual difference between a static model, in which composition
is based on matching pairs as in \eqref{eq:comp_prel},
and a dynamic model, based on causal compatibility and
sensitive to deadlocks. 

\subsubsection{The double category $\dblCG$.} We assemble strategies
and games into a general compositional framework. This is a double
category whose vertical morphisms are maps of event structures, which
compose strictly, and whose horizontal morphisms are strategies
between games, which compose weakly. Altogether, we get:

\begin{thm}\label{th:cg_bicat}
  There is a double category $\dblCG$ with components as follows:
  \begin{itemize}
  \item  objects are plain games;
  \item vertical morphisms $A \to B$ are polarity-preserving rigid maps
of event structures $A \to B$ (also referred to as \emph{maps of
games});
\item horizontal morphisms $A \profto B$ are strategies on $A \vdash B$;
\item $2$-cells of type  \[ \begin{tikzcd}
	A & B \\
	C & D
	\arrow[""{name=0, anchor=center, inner sep=0}, "\sigma", "\shortmid"{marking}, from=1-1, to=1-2]
	\arrow["h"', from=1-1, to=2-1]
	\arrow["k", from=1-2, to=2-2]
	\arrow[""{name=1, anchor=center, inner sep=0}, "\tau"', "\shortmid"{marking}, from=2-1, to=2-2]
	\arrow[shorten <=12pt, shorten >=12pt, Rightarrow, from=0, to=1]
      \end{tikzcd}\]
 are rigid maps $f : \sigma \to \tau$ such that the following diagram
 commutes
\[ \begin{tikzcd}
	\sigma & \tau \\
	{A\vdash B} & {C\vdash D}.
	\arrow["f", from=1-1, to=1-2]
	\arrow["{\pr_\sigma}"', from=1-1, to=2-1]
	\arrow["{\pr_\tau}", from=1-2, to=2-2]
	\arrow["{h\vdash k}"', from=2-1, to=2-2]
\end{tikzcd}\]
\end{itemize}
\noindent 
We write $\CG$ for the horizontal bicategory $\hori(\dblCG)$. It is
clear that 2-cells in the bicategory $\CG$ correspond to the morphisms of
strategies described in Section~\ref{subsubsec:mor_strat}.
\end{thm}
\begin{proof}[Proof sketch]
  We have already introduced most of the necessary components, and here we
   mention two final points. 

First, the copycat construction is functorial, as required in a double
category. More precisely, for a vertical morphism $h : A \to B$, there
is a rigid map of event structures $\cc_h : \cc_A \to \cc_B$, defined to have the
same action on events as $h \vdash h : (A \vdash A) \to (B \vdash
B)$. This defines a 2-cell
\[\begin{tikzcd}
	A & A \\
	B & B
	\arrow[""{name=0, anchor=center, inner sep=0}, "{\cc_A}", "\shortmid"{marking}, from=1-1, to=1-2]
	\arrow["h"', from=1-1, to=2-1]
	\arrow["h", from=1-2, to=2-2]
	\arrow[""{name=1, anchor=center, inner sep=0}, "{\cc_B}"', "\shortmid"{marking}, from=2-1, to=2-2]
	\arrow[shorten <=9pt, shorten >=9pt, Rightarrow, from=0, to=1]
      \end{tikzcd}\]
in a functorial way. 
    
Another remaining challenge is to define the ``horizontal composition''
of $2$-cells
\begin{equation}
  \label{eq:1}
\begin{tikzcd}
	A & B \\
	{A'} & {B'}
	\arrow[""{name=0, anchor=center, inner sep=0}, "\sigma",
        from=1-1, to=1-2, "\shortmid"{marking}]
	\arrow["h"', from=1-1, to=2-1]
	\arrow["k", from=1-2, to=2-2]
	\arrow[""{name=1, anchor=center, inner sep=0}, "{\sigma'}"',
        from=2-1, to=2-2, "\shortmid"{marking}]
	\arrow["f", shorten <=12pt, shorten >=12pt, Rightarrow, from=0, to=1]
      \end{tikzcd}
      \qquad
      \begin{tikzcd}
	B & C \\
	{B'} & {C'}
	\arrow[""{name=0, anchor=center, inner sep=0}, "\tau",
        from=1-1, to=1-2, "\shortmid"{marking}]
	\arrow["k"', from=1-1, to=2-1]
	\arrow["l", from=1-2, to=2-2]
	\arrow[""{name=1, anchor=center, inner sep=0}, "{{\tau'}}"',
        from=2-1, to=2-2, "\shortmid"{marking}]
	\arrow["g", shorten <=12pt, shorten >=12pt, Rightarrow, from=0, to=1]
      \end{tikzcd}
      \qquad \qquad \longmapsto \qquad \qquad
      \begin{tikzcd}
	A && C \\
	{A'} && {C'}
	\arrow[""{name=0, anchor=center, inner sep=0}, "{\tau \odot
          \sigma}", from=1-1, to=1-3, "\shortmid"{marking}]
	\arrow["h"', from=1-1, to=2-1]
	\arrow["l", from=1-3, to=2-3]
	\arrow[""{name=1, anchor=center, inner sep=0}, "{\tau' \odot
          \sigma'}"', from=2-1, to=2-3, "\shortmid"{marking}]
	\arrow["{g \odot f}", shorten <=12pt, shorten >=12pt, Rightarrow, from=0, to=1]
\end{tikzcd}
\end{equation}
and proving the required coherence conditions. This is detailed in
  \cite{DBLP:journals/lmcs/CastellanCRW17}. By Lemma
\ref{lem:pcov_mapification}, it is sufficient to
  define morphisms between strategies on $+$-covered configurations, so
that we may define the horizontal composition simply pointwise with 
\[
(g \odot f)(x^\tau \odot x^\sigma) = g(x^\tau) \odot f(x^\sigma)
\]
for all $x^\sigma \in \confp{\sigma}$ and $x^\tau \in \confp{\tau}$
causally compatible. All the necessary verifications for constructing a
bicategory follow rather easily \cite[Theorem 6.4.11]{hdr}, and it is
only a minor extension to define the full double-categorical structure
(see also
\cite{paquet2020probabilistic}).
\end{proof}

\subsubsection{Symmetric monoidal structure.} We show that $\dblCG$ is
a symmetric monoidal double category. The tensor product of games
$A \otimes B$ extends to strategies: if $\sigma : A \vdash A'$ and
$\tau : B \vdash B'$, then $\sigma \tensor \tau$ is given by
\[ \sigma \tensor \tau \overset{\pr_\sigma \otimes
    \pr_\tau}{\longrightarrow} (A \vdash B) \tensor (A' \vdash B)
  \overset{\cong}{\longrightarrow} (A \otimes A') \vdash (B \otimes
  B')\] and similarly for 2-cells. All coherence data can be defined
using that the category of event structures and maps is symmetric
monoidal under $\tensor$; in particular the empty game gives the
monoidal unit. There is a monoidal interchange law given by the
invertible, globular 2-cell
$(\tau \odot \sigma) \tensor (\tau' \odot \sigma') \to (\tau \otimes
\tau') \odot (\sigma \otimes \sigma')$ that sends $(x^\tau \odot
x^\sigma) \tensor (x^{\tau'} \odot x^{\sigma'})$ to $(x^\tau \otimes
x^{\tau'}) \odot (x^\sigma \otimes x^{\sigma'})$.

\begin{prop}
The double category $\dblCG$ has a symmetric monoidal structure given
by the tensor product $\otimes$ of event structures, which extends to
games, strategies, maps of games, and 2-cells.
\end{prop}
\begin{proof}[Proof sketch]
The required structural vertical morphisms are rigid maps of event
structures representing the symmetry, associativity, and unit
properties of the tensor product. All axioms are verified by elementary reasoning. 
\end{proof}

\subsubsection{Companions in $\dblCG$.} Not all vertical morphisms in
$\dblCG$ admit a horizontal companion, but all vertical isomorphisms
do. The idea is as follows. For an isomorphism $h : A \to
B$ between games, the composite map 
\[
\cc_A \xrightarrow{\display_{\cc_A}} A \vdash A \xrightarrow{A \vdash
  h} A \vdash B
\]
is a strategy because re-indexing along an isomorphism preserves the
strategy axioms. We call this strategy $\companion{h} : A \vdash B$. 

In addition, there is a pair of 2-cells of type 
\[
\begin{tikzcd}
	A & B \\
	B & B
	\arrow[""{name=0, anchor=center, inner sep=0},
        "{\companion{h}}", from=1-1, to=1-2, "\shortmid"{marking}]
	\arrow["h"', from=1-1, to=2-1]
	\arrow[Rightarrow, no head, from=1-2, to=2-2]
	\arrow[""{name=1, anchor=center, inner sep=0}, Rightarrow, no head, from=2-1, to=2-2]
	\arrow[shorten <=12pt, shorten >=12pt, Rightarrow, from=0, to=1]
      \end{tikzcd}\qquad
      \begin{tikzcd}
	A & A \\
	A & B
	\arrow[""{name=0, anchor=center, inner sep=0}, Rightarrow, no head, from=1-1, to=1-2]
	\arrow[Rightarrow, no head, from=1-1, to=2-1]
	\arrow["{h}", from=1-2, to=2-2]
	\arrow[""{name=1, anchor=center, inner sep=0}, "\companion{h}"', "\shortmid"{marking}, from=2-1, to=2-2]
	\arrow[shorten <=12pt, shorten >=12pt, Rightarrow, from=0, to=1]
\end{tikzcd}
  \] 
  given by the diagrams below
  \[
\begin{tikzcd}
	{\cc_A} & {\cc_ B} \\
	{A \vdash A} \\
	{A\vdash B} & {B \vdash B}
	\arrow["{\cc_h}", from=1-1, to=1-2]
	\arrow["{\display_{\cc_A}}"', from=1-1, to=2-1]
	\arrow["{\display_{\cc_B}}", from=1-2, to=3-2]
	\arrow["{A \vdash h}"', from=2-1, to=3-1]
	\arrow["{h\vdash B}"', from=3-1, to=3-2]
      \end{tikzcd}
      \qquad \begin{tikzcd}
	{\cc_A} & {\cc_A} \\
	& {A\vdash A} \\
	{A\vdash A} & {A \vdash B}
	\arrow[Rightarrow, no head, from=1-1, to=1-2]
	\arrow["{\display_{\cc_A}}"', from=1-1, to=3-1]
	\arrow["{\display_{\cc_A}}", from=1-2, to=2-2]
	\arrow["{A \vdash h}", from=2-2, to=3-2]
	\arrow["{A\vdash h}"', from=3-1, to=3-2]
\end{tikzcd}
\]

The two companion axioms are easy to verify.

\begin{prop}
The double category $\dblCG$ has all companions of vertical
isomorphisms. 
\end{prop}

Applying Theorem~\ref{thm:shulman}, we deduce: 
\begin{cor}
The bicategory $\CG$, equipped with the tensor product $\otimes$ and
the associated structural data, is symmetric monoidal.
\end{cor}

\subsection{A static collapse of concurrent games}
\label{subsubsec:plain_oplax}
We are ready to define our first collapse from dynamic to static
semantics. Formally, we describe an oplax double functor
$\coll{-} : \dblCG \to \dblPRel$, where \emph{oplax} means that
composition is not preserved: instead we have globular 2-cells
$\coll{\tau \odot \sigma} \to \coll{\tau} \circ \coll{\sigma}$ which
embed the causally compatible pairs into the matching pairs.

The image of a plain game $A$ is the set $\conf{A}$. To a
strategy $\sigma : A \vdash B$, we associate 
\[
\coll{\sigma}(x_A, x_B) = \{x^\sigma \in \confp{\sigma}\mid
\pr_\sigma(x^\sigma) = x_A \vdash x_B\}
\]
yielding a proof-relevant relation $\coll{\sigma}$ from $\conf{A}$ to
$\conf{B}$. To a 2-cell
\[
\begin{tikzcd}
	A & B \\
	{A'} & {B'}
	\arrow[""{name=0, anchor=center, inner sep=0}, "\sigma",
        from=1-1, to=1-2, "\shortmid"{marking}]
	\arrow["h"', from=1-1, to=2-1]
	\arrow["k", from=1-2, to=2-2]
	\arrow[""{name=1, anchor=center, inner sep=0}, "{\sigma'}"',
        from=2-1, to=2-2, "\shortmid"{marking}]
	\arrow["f", shorten <=12pt, shorten >=12pt, Rightarrow, from=0, to=1]
      \end{tikzcd}
\]
we associate the family of functions mapping $x^\sigma \in
\coll{\sigma}(x_A, x_B)$ to  $f(x^\sigma)
\in \coll{\sigma'}(h(x_A), k(x_B))$, for $x_A \in A$ and $x_B \in B$. 

Proposition \ref{prop:confp_cc} induces an isomorphism of
$\coll{\cc_A}$ with the identity proof-relevant relation. From
Proposition \ref{prop:char_comp} we obtain a function 
$\coll{\tau \odot \sigma} \to \coll{\tau} \circ \coll{\sigma}$. This
is non-invertible in general, because some matching pairs are not
causally compatible (see Figure \ref{fig:deadlock}).

Finally, there is a bijection
$\coll{A \otimes B} \cong \coll{A} \otimes \coll{B}$ for games $A$ and
$B$, which extends to an isomorphism of proof-relevant relations
$\coll{\sigma \otimes \tau} \cong \coll{\sigma} \otimes \coll{\tau}$,
and a similar situation on 2-cells. Overall the double functor (strongly)
preserves the symmetric monoidal structure. 

\begin{thm}\label{th:collPR}
This data determines a monoidal oplax double functor $\coll{-} :
\dblCG \to \dblPRel$. 
\end{thm}
\noindent 
We omit the proof, as the theorem is a special case of the more
elaborate version to come, which includes symmetries. Note that, although this
involves a lot of algebraic data, we emphasize that it would be
incomparably harder to establish a similar result directly at the
bicategorical level. (We discuss a lifting of this theorem to the respective
bicategories in Section~\ref{sec:symm-mono-struct}, when required for our
applications in semantics.)

In summary, we can regard $\CG$ as a dynamic version of $\PRel$, where
the witnesses in $\PRel$ are reached over time and composition is
sensitive to deadlocks.

\section{Accommodating Symmetry}
\label{sec:symmetry}

In this section, we look at a model of concurrent games enriched with
symmetry, known as thin concurrent games
\cite{DBLP:journals/lmcs/CastellanCW19}. We start by 
explaining the basics of event structures with symmetry and
thin concurrent games (Section~\ref{subsec:tcg-defs}). Then we explain how
strategies in thin concurrent games can be viewed as distributors
(Section~\ref{subsubsec:strat_2_dist}).
We then define the double category $\dblTCG$ (Section~\ref{subsec:tcg}) and
construct a monoidal oplax functor
$\dblTCG \to \dblDist$ (Section~\ref{subsec:collapse}). Finally we
discuss the exponential modality (Section~\ref{subsec:problems}). 

\subsection{Symmetry and thin concurrent games}
\label{subsec:tcg-defs}
Recall that we went from $\PRel$ to $\Dist$ by replacing sets with
groupoids. We now go from $\CG$ to $\TCG$ by replacing the set of
configurations  $\conf{A}$ with a groupoid of configurations
$\tilde{A}$ whose morphisms are chosen bijections called
\emph{symmetries}, that behave well with respect to the causal order.

\subsubsection{Event structures with symmetry}
\label{subsubsec:ess}

Our model is based on the following notion of event structure with
symmetry \cite{DBLP:journals/entcs/Winskel07}:
\begin{defi}\label{def:isofam}%
An \textbf{isomorphism family} on es $E$ is a groupoid $\tilde{E}$
having as objects all configurations, and as morphisms certain
bijections between configurations, satisfying: 
\[
\begin{array}{rl}
\text{\emph{restriction:}}&
\text{for all~$\theta : x \bij y \in \tilde{E}$ and $x \supseteq x' \in
\conf{E}$,}\\
&\text{there is $\theta \supseteq \theta' \in \tilde{E}$
such that $\theta' : x' \bij y'$.}\\
\text{\emph{extension:}}&\text{for all $\theta : x \bij y \in
\tilde{E}$, $x
\subseteq x' \in \conf{E}$,}\\
&\text{there is $\theta \subseteq \theta' \in
\tilde{E}$ such that $\theta' : x' \bij y'$.}
\end{array}
\]

We call $(E, \tilde E)$ an \textbf{event structure with symmetry (\emph{ess})}.
\end{defi}
\noindent 
We refer to morphisms in $\tilde{E}$ as \textbf{symmetries}, and write
$\theta : x \sym_E y$ if $\theta : x \bij y$ with $\theta \in
\tilde{E}$. The \textbf{domain} $\dom(\theta)$ of $\theta : x \sym_E y$
is $x$, and likewise its \textbf{codomain} $\cod(\theta)$ is $y$.
A \textbf{map of ess} $E \to F$ is a map of event
structures that preserves symmetry: the bijection
\[
f \theta \qquad \overset{\mathrm{def}}= \qquad 
f x 
\quad \stackrel{f^{-1}}{\bij} \quad
x 
\quad \stackrel{\theta}{\bij} \quad 
y
\quad \stackrel{f}{\bij} \quad
f y,
\]
is in $\tilde{F}$ for
every $\theta : x \sym_E y$ (recall that $f$ restricted to any $y$ is
bijective).  This makes $f : \tilde{E} \to \tilde{F}$ a functor of groupoids. 

We can define a 2-category $\ESS$ of ess, maps of ess, and natural
transformations between the induced functors. For $f, g : E \to F$ such
a natural transformation is necessarily unique
\cite{DBLP:journals/entcs/Winskel07}, and corresponds to the fact that
for every $x \in \conf{E}$ the composite bijection
\[
f\,x 
\quad \stackrel{f^{-1}}{\bij} \quad
x
\quad \stackrel{g}{\bij} \quad
g\,x
\]
via local injectivity of $f$ and $g$, is in $\tilde{F}$.
So this is an equivalence, denoted  $f \sim g$. 

\subsubsection{Thin games} We define games with symmetry. To
match the polarized structure, a game is an ess with two
sub-symmetries, one for each player (see
e.g.~\cite{mellies2003asynchronous,DBLP:journals/lmcs/CastellanCW19,mfps22}). 
\begin{defi}\label{def:tcg}
A \textbf{thin concurrent game (tcg)} is a game $A$ with 
isomorphism families $\tilde{A}, \ptilde{A}, \ntilde{A}$ s.t. $\ptilde{A}, \ntilde{A} \subseteq
\tilde{A}$, symmetries preserve polarity, and 
\[
\begin{array}{rl}
\text{\emph{(1)}} & \text{if $\theta \in \ptilde{A} \cap \ntilde{A}$,
then $\theta = \id_x$ for $x \in \conf{A}$,}\\
\text{\emph{(2)}} & \text{if $\theta \in \ntilde{A}$, 
$\theta \subseteq^- \theta' \in \tilde{A}$, then $\theta' \in
\ntilde{A}$,}\\
\text{\emph{(3)}} & \text{if $\theta \in \ptilde{A}$, 
$\theta \subseteq^+ \theta' \in \tilde{A}$, then $\theta' \in
\ptilde{A}$,}
\end{array}
\]
where $\theta \subseteq^p \theta'$ is $\theta \subseteq
\theta'$ with (pairs of) events of polarity $p$.
\end{defi}

Elements of $\ptilde{A}$ (resp. $\ntilde{A}$) are called
\textbf{positive} (resp. \textbf{negative}); they intuitively
correspond to symmetries carried by positive (resp. negative) moves,
and thus introduced by Player (resp. Opponent). We write $\theta : x
\sym_A^- y$ (resp. $\theta : x \sym_A^+ y$) if $\theta \in \ntilde{A}$
(resp. $\theta \in \ptilde{A}$).

Each symmetry has a unique positive-negative factorization \cite[Lemma
7.1.18]{hdr}:
\begin{lem}\label{lem:sym_factor}
For $A$ a tcg and $\theta : x \sym_A z$, there are unique
$y \in \conf{A}$, $\theta_- : x \sym_A^- y$ and $\theta_+ : y \sym_A^+
z$ s.t. $\theta = \theta_+ \circ \theta_-$.
\end{lem}

We extend with symmetry the basic constructions on games:
the \textbf{dual} $A^\perp$ has the same symmetries as $A$, but
$\ptilde{A^\perp} = \ntilde{A}$ and $\ntilde{A^\perp} = \ptilde{A}$;
the \textbf{tensor} $A_1 \tensor A_2$ has symmetries of the form $\theta_1 \tensor \theta_2 : x_1
  \tensor x_2 \sym_{A_1 \tensor A_2} y_1 \tensor y_2$,   where each
$\theta_i : x_i \sym_{A_i} y_i$, and similarly for positive and
negative symmetries; the \textbf{hom} $A \vdash B$ is $A^\perp \tensor B$.

\subsubsection{Thin strategies} We now define strategies on thin
concurrent games:

\begin{defi}\label{def:thin}
Consider $A$ a tcg. 
\noindent 
A \textbf{strategy} on $A$, written $\sigma :
A$, is an ess $\sigma$ equipped with a morphism of ess $\pr_\sigma :
\sigma \to A$ forming a strategy in the sense of Definition
\ref{def:plain_strategy}, and such that:
\[
\begin{array}{rl}
\text{\emph{(1)}} &
\text{if $\theta \in \tilde{\sigma}, \pr_\sigma \theta
\enb_A (a^-, b^-)$, there are unique $\theta \vdash_\sigma (s, t)$ s.t. $\pr_\sigma s = a$ and
$\pr_\sigma t = b$.}\\
\text{\emph{(2)}} & 
\text{if $\theta : x \sym_\sigma y, \pr_\sigma \theta \in
\ptilde{A}$, then $x = y$ and $\theta = \id_x$.}
\end{array}
\]

As before, a \textbf{strategy from $A$ to $B$} is a strategy on $\sigma : A
\vdash B$. 
\end{defi}
\noindent 
The first condition forces $\sigma$ to acknowledge Opponent symmetries
in $A$; the notation $\theta \vdash_A (a, b)$
means $(a, b) \not \in \theta$ and $\theta \cup \{(a, b)\} \in
\tilde{A}$. The second condition is \textbf{thinness}: it means that any
non-identity symmetry in the strategy must originate from Opponent.

\subsubsection{Comparison with the ``saturated'' approach} 
The ``thin'' approach is only one possible way of adding symmetry to
games. Other models
(e.g.~\cite{DBLP:conf/lics/BaillotDE97,DBLP:conf/csl/CastellanCW14,DBLP:conf/lics/Mellies19})
follow a different approach, where strategies satisfy a 
\emph{saturation} condition.

We explain the difference in the language of concurrent games. 
Consider a strategy  $\sigma : A$ on a tcg, in the sense of
Definition~\ref{def:thin} \emph{without} conditions (1) and (2). 
The saturation condition \cite{DBLP:conf/csl/CastellanCW14} corresponds to a fibration property of the
functor
$  \pr_\sigma : \tilde{\sigma} \to \tilde{A}$: for $x \in \conf{\sigma}$ and $\psi :
\pr_\sigma x \sym_{A} y$, there is a unique $\varphi : x
\sym_\sigma z$ such that $\pr_\sigma \varphi = \psi$:
\begin{equation}
  \label{eq:saturation}
  \begin{tikzcd}
    {\color{gray}\tilde \sigma} &[-2em]& x & z \\
    {\color{gray}\tilde A} && {\pr_\sigma x} & y
    \arrow[maps to, from=1-3, to=2-3]
    \arrow["\varphi", dashed, from=1-3, to=1-4]
    \arrow["{\psi}"', from=2-3, to=2-4]
    \arrow[maps to, from=1-4, to=2-4]
    \arrow["{\color{gray}\pr_\sigma}"', color=gray, from=1-1, to=2-1]
  \end{tikzcd}
\end{equation}

In contrast, thin strategies satisfy a different lifting property:
a unique lifting exists, up to a positive symmetry \cite[Lemma
7.2.6]{hdr}:

\begin{lem}\label{lem:act_strat}
Let $\sigma : A$ be a strategy as in Definition~\ref{def:thin}.
For all $x \in \conf{\sigma}$ and $\psi : \pr_\sigma x \sym_A y$,
there are unique $\varphi : x \sym_\sigma z$ and $\theta^+: \pr_\sigma
z \sym_A^+ y$ such that $\theta^+ \circ \pr_\sigma \varphi= \psi$:
\[\begin{tikzcd}
	{\color{gray}\tilde \sigma} &[-2em]& x && z \\
	{\color{gray}\tilde A} && {\pr_\sigma x} & y & {\pr_\sigma z}
	\arrow[maps to, from=1-3, to=2-3]
	\arrow["\varphi", dashed, from=1-3, to=1-5]
	\arrow["\psi"', from=2-3, to=2-4]
	\arrow["{\color{gray}\pr_\sigma}"', color=gray, from=1-1, to=2-1]
	\arrow["{\theta^+}", dashed, from=2-5, to=2-4]
	\arrow[maps to, from=1-5, to=2-5]
	\arrow["{\pr_\sigma \varphi}", curve={height=-12pt}, from=2-3, to=2-5]
      \end{tikzcd}\]
  \end{lem}
\begin{proof}[Sketch]
\emph{Existence} is proved first for $\psi$ positive, by induction on
$x$ using condition \emph{(1)} of Definition \ref{def:thin} and
properties of isomorphism families; and then generalized to arbitrary
$\psi$. \emph{Uniqueness} follows from condition \emph{(2)} of
Definition \ref{def:thin}.
\end{proof}

Below, we will use this to construct a distributor from a
thin strategy. We note that saturated strategies are closer to
distributors, because the saturation property \eqref{eq:saturation}
directly induces a functorial action of the groupoid $\tilde A$.
However, saturated strategies are more difficult to
understand operationally, and have not achieved the precision
of thin strategies for languages with state, concurrency, or
non-determinism \cite{DBLP:journals/corr/abs-2103-15453,hdr}.

\subsection{From strategies to distributors}
\label{subsubsec:strat_2_dist}
For tcgs $A$ and $B$, we show how to construct 
\[
\coll{\sigma} : \tilde{A}^{\op} \times \tilde{B} \to \Set
\]
a distributor
from a strategy $\sigma : A \vdash B$. The key idea is to use
witnesses ``up to positive symmetry'' and use
the lifting
property in Lemma~\ref{lem:act_strat}.

For $x_A \in \conf{A}$ and $x_B \in \conf{B}$ we define the set of
\textbf{positive witnesses} of $(x_A, x_B)$, written
$\coll{\sigma}(x_A, x_B)$, as the set of all triples $(\theta_A^-,
x^\sigma, \theta_B^+)$ such that $x^\sigma \in \confp{\sigma}$ and
\[
\theta_A^- : x_A {\sym_A^-} x^\sigma_A\,
\qquad\qquad
\theta_B^+ : x^\sigma_B {\sym_B^+} x_B
\]
are positive symmetries on $A^\perp$ and $B$.
The groupoid actions of $A$ and $B$ on this set
are determined by the uniqueness result:
\begin{prop}\label{prop:act_strat}
Consider $(\theta_A^-, x^\sigma, \theta_B^+) \in \coll{\sigma}(x_A,
x_B)$.

For each $\Omega_A : y_A \sym_A x_A$ and $\Omega_B : x_B \sym_B y_B$,
there are unique $\varphi : x^\sigma \sym_\sigma y^\sigma$ and
$\vartheta_A^- : y_A \sym_A^- y^\sigma_A$, $\vartheta_B^+ : y^\sigma_B
\sym_B^+ y_B$ such that the following two diagrams commute:
\[
\xymatrix{
x_A	\ar[r]^{\theta_A^-}
	\ar@{<-}[d]_{\Omega_A}&
x^\sigma_A
	\ar[d]^{\varphi^\sigma_A}\\
y_A	\ar[r]_{\vartheta_A^-}&
y^\sigma_A
}
\qquad\qquad
\xymatrix{
x^\sigma_B
	\ar[r]^{\theta_B^+}
	\ar[d]_{\varphi^\sigma_B}&
x_B	\ar[d]^{\Omega_B}\\
y^\sigma_B
	\ar[r]_{\vartheta_B^+}&
y_B
}
\]
\end{prop}
\begin{proof}
\emph{Existence} by Lem. \ref{lem:act_strat},
\emph{uniqueness} by \emph{(2)} of Definition \ref{def:thin}.
\end{proof}

Thus we may set $\coll{\sigma}(\Omega_A, \Omega_B)(\theta_A^-,
x^\sigma, \theta_B^+)$ as the positive witness $(\vartheta_A^-,
y^\sigma, \vartheta_B^+)$ above. This immediately gives us a
distributor:
\begin{prop}
\label{prop:strategy-to-distributor}
We have $\coll{\sigma} : \tilde{A}^{\op}
\times \tilde{B} \to \Set$.
\end{prop}

\subsection{The double category of thin concurrent games}
\label{subsec:tcg}

We now define the double category $\dblTCG$ of thin concurrent
games. We have already defined the objects (tcgs, Definition~\ref{def:tcg})
and the horizontal morphisms (thin strategies,
Definition~\ref{def:thin}), so it remains to define the vertical
morphisms and $2$-cells, and appropriate compositions and identities. 

\subsubsection{Vertical morphisms: maps between tcgs.}
Recall that in $\dblCG$ the vertical morphisms between games are rigid maps
of event structures preserving the polarity of moves. Here we define 
\textbf{maps of tcgs} $A \to B$ to be rigid maps of event structures with
symmetry, preserving the polarity of moves, and preserving the polarity of symmetry
bijections. 

\subsubsection{2-cells: morphisms of strategies} We now define
generalized morphisms between strategies. The 2-cells of $\dblTCG$ are more liberal than those in $\dblCG$, because
there should be an isomorphism between two strategies which play
symmetric -- rather than equal -- moves.  Recall the
2-dimensional structure in $\ESS$, given by the equivalence relation
$\sim$ on morphisms (Section~\ref{subsubsec:ess}). For two maps $f, g : E \to
A$ into a tcg, we write $f \sim^+ g$ if $f \sim g$ and for every $x \in
\conf{E}$ the symmetry obtained as the composition
\[
f\,x 
\quad \stackrel{f^{-1}}{\bij} \quad
x
\quad \stackrel{g}{\bij} \quad
g\,x\,,
\]
witnessing $f \sim g$ for $x$, is positive.

\begin{defi}
Let $\sigma : A \vdash B$ and $\tau : C \vdash D$ be thin
strategies, and let $h : A \to C$ and $k : B \to D$ be maps of tcgs. A
\textbf{positive morphism of strategies} of type
\[
\begin{tikzcd}
	A & B \\
	C & D
	\arrow[""{name=0, anchor=center, inner sep=0}, "\sigma",
        from=1-1, to=1-2, "\shortmid"{marking}]
	\arrow["h"', from=1-1, to=2-1]
	\arrow["k", from=1-2, to=2-2]
	\arrow[""{name=1, anchor=center, inner sep=0}, "{\tau}"',
        from=2-1, to=2-2, "\shortmid"{marking}]
	\arrow["f", shorten <=12pt, shorten >=12pt, Rightarrow, from=0, to=1]
      \end{tikzcd}
    \]
    is a rigid map of ess $f : \sigma \to \tau$ s.t. the following diagram commutes up to
    positive symmetry:
\[ \begin{tikzcd}
	\sigma & \tau \\
	{A\vdash B} & {C\vdash D}.
	\arrow["f", from=1-1, to=1-2]
	\arrow["{\pr_\sigma}"', from=1-1, to=2-1]
	\arrow["{\pr_\tau}", from=1-2, to=2-2]
	\arrow["{h\vdash k}"', from=2-1, to=2-2]
	\arrow["{\sim^+}"{description, pos=0.7}, draw=none, from=1-1, to=2-2]
\end{tikzcd}\]
    \emph{i.e.}   $\pr_\tau \circ f \sim^+ (h\vdash k) \circ \pr_\sigma$.
  \end{defi}
  As a convention, if $f$ is a 2-cell as above, for
  $x^\sigma \in \conf{\sigma}$ we write
\[
f[x^\sigma] : (h\vdash k)( \pr_\sigma x^\sigma) \sym_A^+ \pr_{\tau}(f\,x^\sigma)
\]
for the positive symmetry witnessing this.

\subsubsection{Composition and identity}  
The composition of thin
strategies $\sigma : A \vdash B$ and $\tau : B \vdash C$ is obtained by
equipping $\tau \odot \sigma$ (Proposition
\ref{prop:char_comp}) with an adequate isomorphism family. If
$\tildep{\sigma}$ is the restriction of $\tilde{\sigma}$ to
$+$-covered configurations, then we can write $\CC(\tildep{\sigma},
\tildep{\tau})$
for the pairs $(\varphi^\sigma, \varphi^\tau)$ of symmetries  
which are matching, i.e. $\varphi^\sigma_B = \varphi^\tau_B$ and whose
domain (and necessarily, codomain) are causally compatible.

\begin{prop} Consider $\sigma : A \vdash B$ and $\tau : B \vdash C$
thin strategies.

There is a unique symmetry on $\tau \odot \sigma$ with a bijection
commuting with $\dom$ and $\cod$
\[
(- \odot -) : \CC(\tildep{\sigma}, \tildep{\tau}) \bij \tildep{\tau
\odot \sigma}
\]
and compatible with
display maps, \emph{i.e.} 
$(\varphi^\tau \odot \varphi^\sigma)_A = \varphi^\sigma_A$ and
$(\varphi^\tau \odot \varphi^\sigma)_C = \varphi^\tau_C$.
\end{prop}
\begin{proof}
This follows from \cite[Proposition 7.3.1]{hdr}.
\end{proof}

Most of the effort in organizing these components into a double
category is due to the difficulty of composing 2-cells
horizontally. Suppose we have a pair of 2-cells $f$ and $g$ as in \eqref{eq:1},
but where all components now have symmetry. Explicitly we have
\[
\begin{tikzcd}
	\sigma & {\sigma'} \\
	{A\vdash B} & {A'\vdash B'}
	\arrow["f", from=1-1, to=1-2]
	\arrow["{{\pr_\sigma}}"', from=1-1, to=2-1]
	\arrow["{{\sim^+}}"{description, pos=0.7}, draw=none, from=1-1, to=2-2]
	\arrow["{{\pr_{\sigma'}}}", from=1-2, to=2-2]
	\arrow["{{h\vdash k}}"', from=2-1, to=2-2]
      \end{tikzcd}\qquad\qquad
      \begin{tikzcd}
	\tau & {\tau'} \\
	{B \vdash C} & {B'\vdash C'}
	\arrow["g", from=1-1, to=1-2]
	\arrow["{{\pr_\tau}}"', from=1-1, to=2-1]
	\arrow["{{\sim^+}}"{description, pos=0.7}, draw=none, from=1-1, to=2-2]
	\arrow["{{\pr_{\tau'}}}", from=1-2, to=2-2]
	\arrow["{{k \vdash l}}"', from=2-1, to=2-2]
\end{tikzcd}
  \]
and we must define an appropriate map $g \odot f : \tau \odot \sigma
\to \tau' \odot \sigma'$. 

Recall (from the discussion following Theorem \ref{th:cg_bicat}) that
in the analogous situation in $\dblCG$,
$g \odot f : \tau \odot \sigma \Rightarrow \tau' \odot \sigma'$ was
characterized by
$(g \odot f)(x^\tau \odot x^\sigma) = g(x^\tau) \odot f(x^\sigma)$. In
$\dblTCG$ this simple description is no longer possible, as we may not
have a matching situation $f(x^\sigma)_{B'} = g(x^\tau)_{B'}$. Instead,
these two configurations are matching up to a symmetry
\[
\theta^{f, g}_{x^\sigma, x^\tau} : f(x^\sigma)_{B'} \sym_{B'} g(x^\tau)_{B'}
\]
obtained as $\theta^{f, g}_{x^\sigma, x^\tau} = g[x^\tau]_{B'} \circ
f[x^\sigma]_{B'}^{-1}$.
Fortunately, interaction of thin strategies supports a notion of
synchronization \emph{up to symmetry} \cite[Proposition 7.4.4]{hdr}:

\begin{prop}\label{prop:sync_sym}
Consider $x^\sigma \in \confp{\sigma}, \theta : x^\sigma_B \sym_B
x^\tau_B, x^\tau \in \confp{\tau}$ \textbf{causally compatible},
\emph{i.e.}
the relation $\cleq$ induced on the graph of the composite bijection
\[
\scalebox{.90}{$
x^\sigma \parallel x^\tau_C 
 \!\stackrel{\pr_\sigma \parallel x^\tau_C}{\bij}\!
x^\sigma_A \parallel x^\sigma_B \parallel x^\tau_C
 \!\stackrel{x^\sigma_A \parallel \theta \parallel x^\tau_C}{\bij}\!
x^\sigma_A \parallel x^\tau_B \parallel x^\tau_C
 \!\stackrel{x^\sigma_A \parallel \pr_\tau^{-1}}{\bij}\!
x^\sigma_A \parallel x^\tau
$}
\]
by $<_{\sigma \parallel C}$ and $<_{A\parallel \tau}$
as in Section~\ref{subsubsec:cg_comp}, is acyclic -- we also say the
composite bijection is \textbf{secured}.

Then, there are unique $y^\tau \odot y^\sigma \in \confp{\tau
\odot \sigma}$ with symmetries
$\varphi^\sigma : x^\sigma \sym_\sigma y^\sigma$ and $\varphi^\tau :
x^\tau \sym_\tau y^\tau$, such that $\varphi^\sigma_A \in \ntilde{A}$
and $\varphi^\tau_C \in \ptilde{C}$, and $\varphi^\tau_B \circ \theta =
\varphi^\sigma_B$.
\end{prop}

In that case, we write $y^\tau \odot y^\sigma = x^\tau \odot_\theta
x^\sigma$. With that notation, there is a unique positive morphism
$g\odot f : \tau \odot \sigma \Rightarrow \tau' \odot \sigma'$ s.t.
$(g \odot f)(x^\tau \odot x^\sigma) = g(x^\tau) \odot_{\theta^{f,
    g}_{x^\sigma, x^\tau}} f(x^\sigma)$. This serves as a
definition of the horizontal composition of $2$-cells, see
\cite{paquet2020probabilistic} or \cite[Theorem 7.4.13]{hdr}. 

Finally, for the identity horizontal morphism in $\dblTCG$, we equip the
copycat strategy $\cc_A : A \vdash A$ (Proposition
\ref{prop:confp_cc}) with the unique symmetry that has an iso
\[
\cc_{(-)} : \tilde{A} \bij \tildep{\cc_A}
\]
commuting with $\dom$ and $\cod$, such that
$\pr_{\cc_A}(\cc_\theta) = \theta \vdash \theta$.

The tensor product $\otimes$ of tcgs extends to a symmetric monoidal
structure. We omit all of the details, as the situation is completely
analogous to that in $\dblCG$, using that the symmetry is always
handled componentwise in a tensor product of games. The following
proposition assembles all of the structure defined in this section.
\begin{prop}
  There is a symmetric monoidal double category $\dblTCG$, and a
  symmetric monoidal embedding $\dblCG \hookrightarrow \dblTCG$. 
\end{prop}
Thus we have extended our basic model $\dblCG$ with symmetries. The
new model supports an exponential modality, which we discuss in
Section~\ref{subsec:problems}.  

\subsection{An oplax functor $\dblTCG \to \dblDist$}
\label{subsec:collapse}

We now give the components of an oplax double functor $\dblTCG \to
\dblDist$. We have already explained the action on objects and
horizontal morphisms,
by mapping every strategy to a distributor
(Section~\ref{subsubsec:strat_2_dist}). We turn to the action on vertical
morphisms and 2-cells.

\subsubsection{From maps of tcgs to functors of groupoids.} If $A$ and
$B$ are tcgs, and $l : A \to B$ is a map preserving symmetries, then
the action of $l$ on configurations and symmetries determines a functor
$\coll{l} : \tilde{A} \to \tilde{B}$. 

\subsubsection{From positive morphisms to natural transformations}
We must show that a 2-cell in $\dblTCG$, say $f$ in the diagram below, 
\[\begin{tikzcd}
	\sigma & \tau \\
	{A \vdash B} & {C \vdash D}
	\arrow["f", from=1-1, to=1-2]
	\arrow["{{\pr_\sigma}}"', from=1-1, to=2-1]
	\arrow["{{\sim^+}}"{description, pos=0.7}, draw=none, from=1-1, to=2-2]
	\arrow["{{\pr_{\tau}}}", from=1-2, to=2-2]
	\arrow["{{h \vdash k}}"', from=2-1, to=2-2]
\end{tikzcd}\]
determines a natural transformation of distributors
$\coll{\sigma}(x_A, x_B)
\Rightarrow \coll{\tau}(h(x_A), k(x_B))$. Its components are the functions
\[
\begin{array}{rcrcl}
\coll{f}_{x_A, x_B} \!\!&\!\!\!\!\!\!:\!\!\!\!\!\!&\!\! \coll{\sigma}(x_A, x_B) &\to& \coll{\tau}(h(x_A), k(x_B))\\
&&(\theta_A^-, x^\sigma, \theta_B^+)
\!\!&\!\!\!\!\!\!\mapsto\!\!\!\!\!\!&\!\! 
(\theta^x_C \circ h(\theta_A^-), f(x^\sigma), k(\theta_B^+) \circ
{\theta^x_D}^{-1})
\end{array}
\]
for $x_A \in \conf{A}$ and $x_B \in \conf{B}$, where
$\theta^x_C : h(x^\sigma_A) \sym_C^- f(x^\sigma)_C$ and
$\theta^x_D : k(x^\sigma_B) \sym_D^+ f(x^\sigma)_D$ are determined by
the commutative diagram above, by definition of $\sim^+$ (see
Section~\ref{subsec:tcg}). 
This is natural, as an application of Proposition \ref{prop:act_strat}.

\subsubsection{The unitor and compositor}
To complete the construction of the double functor $\coll{-}$, we must
define appropriate globular 2-cells. This double functor will be
\emph{oplax normal}, meaning that there are invertible \textbf{unitors} and
non-invertible \textbf{compositors}. We define these in the next two
propositions. 

\begin{prop}
Consider $A$ a tcg. Then, there is a natural iso
\[
\pid^A : \coll{\cc_A} \stackrel{\iso}{\Rightarrow} \tilde{A}[-, -] : 
\tilde{A}^{\op} \times \tilde{A} \to \Set\,.
\]
\end{prop}
\begin{proof}
Consider $(\theta^-, \cc_z, \theta^+) \in \coll{\cc_A}(x, y)$, with
$\theta^- : x \sym_A^- z$ and $\theta^+ : z \sym_A^+ y$. We 
set $\pid^A(\theta^-, \cc_z, \theta^+) = \theta^+ \circ \theta^-$;
naturality and invertibility follow from Lemma
\ref{lem:sym_factor}. 
\end{proof}

Likewise, for two strategies $\sigma : A \vdash B$ and $\tau : B \vdash
C$, we have a compositor as follows. 

\begin{prop}\label{prop:compositor}
For any $\sigma : A \vdash B$ and $\tau : B \vdash C$, there is a
natural transformation: 
\[
\pcomp^{\sigma, \tau} : \coll{\tau \odot \sigma} \Rightarrow
\coll{\tau} \bullet \coll{\sigma} :
\tilde{A}^{\op} \times \tilde{C} \to \Set\,.
\]
\end{prop}
\begin{proof}
Consider $(\theta_A^-, x^\tau \odot x^\sigma, \theta_C^+) \in
\coll{\tau \odot \sigma}(x_A, x_C)$; this is sent by
$\pcomp^{\sigma,\tau}_{x_A, x_C}$ to (the equivalence
class for the equivalence $\sim$ originating for the coend formula
\eqref{eq:coend} of) the pair
\[
((\theta_A^-, x^\sigma, \id_{x_B}), (\id_{x_B}, x^\tau, \theta_C^+))
\in (\coll{\tau} \bullet \coll{\sigma})(x_A, x_C)
\]
for $x^\sigma_B = x^\tau_B = x_B$, exploiting that $x^\sigma$ and
$x^\tau$ are matching.

We must show that the function $\pcomp^{\sigma, \tau}(x_A, x_C) :
\coll{\tau \odot \sigma}(x_A, x_C) \to (\coll{\tau} \bullet
\coll{\sigma})(x_A, x_C)$ is natural in $x_A, x_C$.
Consider $\w^{\tau \odot \sigma} \in \coll{\tau \odot \sigma}(x_A,
x_C)$,
written as
\[
\w^{\tau \odot \sigma} = (\psi_A^-, x^\tau \odot x^\sigma, \psi_C^+)\,,
\]
hence with $\pcomp(\w^{\tau \odot \sigma}) = (\w^\sigma, \w^\tau)$
where $\w^\sigma = (\psi_A^-, x^\sigma, \id_{x^\sigma_B})$ and
$\w^\tau = (\id_{x^\tau_B}, x^\tau, \psi_C^+)$.

Now consider $\theta_A : y_A \sym_A x_A, \theta_C : x_C \sym_C y_C$,
then by definition of $\theta_C \cdot \w^{\tau\odot\sigma} \cdot
\theta_A$,
it must be given as $(\nu_A^-, y^\tau \odot y^\sigma, \nu_C^+)$ as in
the bottom of the following diagram:
\[
\xymatrix@R=15pt@C=15pt{
x_A     \ar[r]^{\psi_A^-}
        \ar[d]_{\theta_A^{-1}}&
x^\sigma_A
        \ar[d]_{\varphi^\sigma_A}&
x^\sigma_B
        \ar@{}|=[r]
        \ar[d]^{\varphi^\sigma_B}&
x_B     \ar@{}[d]|{\rotatebox{90}{$=$}}
        \ar@{}|=[r]&
x^\tau_B
        \ar[d]_{\varphi^\tau_B}&
x^\tau_C
        \ar[d]^{\varphi^\tau_C}
        \ar[r]^{\psi_C^+}&
x_C     \ar[d]^{\theta_C}\\
y_A     \ar[r]_{\nu_A^-}&
y^\sigma_A&
y^\sigma_B
        \ar@{}|=[r]&
y_B     \ar@{}|=[r]&
y^\tau_B&
y^\tau_C
        \ar[r]_{\nu_C^+}&
y_C
}
\]
for $\varphi^\sigma : x^\sigma \sym_\sigma y^\sigma$ and
$\varphi^\tau : x^\tau \sym_\tau y^\tau$. But it also follows
\[
\w^\sigma \cdot \theta_A = (\nu_A^-, y^\sigma, \id_{y^\sigma_B})\,,
\qquad
\theta_C \cdot \w^\tau = (\id_{y^\tau_B}, y^\tau, \nu_C^+)
\]
by definition of these functorial actions, so that
\[
\pcomp(\theta_C \cdot \w^{\tau \odot \sigma} \cdot \theta_A) =
(\w^\sigma \cdot \theta_A, \theta_C \cdot \w^\tau)
\]
as required for the naturality of $\pcomp^{\sigma, \tau}$.
\end{proof}

Next, we must additionally prove that $\pcomp^{\sigma, \tau}$ is
natural in $\sigma$ and $\tau$:

\begin{lem}[Naturality of $\pcomp^{\sigma, \tau}$]
For any pair of 2-cells of the form
\[\begin{tikzcd}
	A & B \\
	{A'} & {B'}
	\arrow[""{name=0, anchor=center, inner sep=0}, "\sigma",
        from=1-1, to=1-2, "\shortmid"{marking}]
	\arrow["h"', from=1-1, to=2-1]
	\arrow["k", from=1-2, to=2-2]
	\arrow[""{name=1, anchor=center, inner sep=0}, "{\sigma'}"',
        from=2-1, to=2-2, "\shortmid"{marking}]
	\arrow["f", shorten <=12pt, shorten >=12pt, Rightarrow, from=0, to=1]
      \end{tikzcd}
      \qquad
      \begin{tikzcd}
	B & C \\
	{B'} & {C'}
	\arrow[""{name=0, anchor=center, inner sep=0}, "\tau",
        from=1-1, to=1-2, "\shortmid"{marking}]
	\arrow["k"', from=1-1, to=2-1]
	\arrow["l", from=1-2, to=2-2]
	\arrow[""{name=1, anchor=center, inner sep=0}, "{{\tau'}}"',
        from=2-1, to=2-2, "\shortmid"{marking}]
	\arrow["g", shorten <=12pt, shorten >=12pt, Rightarrow, from=0, to=1]
      \end{tikzcd}
 \]
the following equation holds:
\[\begin{tikzcd}[column sep=1em]
	{\tilde{A}} && {\tilde{C}} \\
	{\tilde{A}} & {\tilde{B}} & {\tilde{C}} \\
	{\tilde{A'}} & {\tilde{B'}} & {\tilde{C'}}
	\arrow[""{name=0, anchor=center, inner sep=0}, "{\coll{\tau \odot \sigma}}", "\shortmid"{marking}, from=1-1, to=1-3]
	\arrow[Rightarrow, no head, from=1-1, to=2-1]
	\arrow[Rightarrow, no head, from=1-3, to=2-3]
	\arrow[""{name=1, anchor=center, inner sep=0}, "{\coll{\sigma}}", "\shortmid"{marking}, from=2-1, to=2-2]
	\arrow["{\coll{h}}"', from=2-1, to=3-1]
	\arrow[""{name=2, anchor=center, inner sep=0}, "{\coll{\tau}}", "\shortmid"{marking}, from=2-2, to=2-3]
	\arrow["{\coll{k}}", from=2-2, to=3-2]
	\arrow["{\coll{l}}", from=2-3, to=3-3]
	\arrow[""{name=3, anchor=center, inner sep=0}, "{\coll{\sigma'}}"', "\shortmid"{marking}, from=3-1, to=3-2]
	\arrow[""{name=4, anchor=center, inner sep=0}, "{\coll{\tau'}}"', "\shortmid"{marking}, from=3-2, to=3-3]
	\arrow["{\ \pcomp^{\sigma,\tau}}", shorten <=9pt, shorten >=9pt, Rightarrow, from=0, to=2-2]
	\arrow["{\coll{f}}", shorten <=14pt, shorten >=14pt, Rightarrow, from=1, to=3]
	\arrow["{\coll{g}}", shorten <=14pt, shorten >=14pt, Rightarrow, from=2, to=4]
        &\end{tikzcd} \quad
      =\quad
      \begin{tikzcd}[column sep=1em]
	{\tilde{A}} && {\tilde{C}} \\
	{\tilde{A'}} && {\tilde{C'}} \\
	{\tilde{A'}} & {\tilde{B'}} & {\tilde{C'}.}
	\arrow[""{name=0, anchor=center, inner sep=0}, "{\coll{\tau \odot \sigma}}", "\shortmid"{marking}, from=1-1, to=1-3]
	\arrow["{\coll{h}}"', from=1-1, to=2-1]
	\arrow["{\coll{l}}", from=1-3, to=2-3]
	\arrow[""{name=1, anchor=center, inner sep=0}, "{\coll{\tau' \odot \sigma'}}"{description}, from=2-1, to=2-3]
	\arrow[Rightarrow, no head, from=2-1, to=3-1]
	\arrow[Rightarrow, no head, from=2-3, to=3-3]
	\arrow["{\coll{\sigma'}}"', "\shortmid"{marking}, from=3-1, to=3-2]
	\arrow["{\coll{\tau'}}"', "\shortmid"{marking}, from=3-2, to=3-3]
	\arrow["{\coll{g \odot f}}", shorten <=9pt, shorten >=9pt, Rightarrow, from=0, to=1]
	\arrow["{\ \pcomp^{\sigma', \tau'}}", shorten <=6pt, shorten >=6pt, Rightarrow, from=1, to=3-2]
\end{tikzcd}\]
\end{lem}
\begin{proof}
Consider $\w^{\tau \odot \sigma} = (\theta_A^-,
x^\tau \odot x^\sigma, \theta_C^+) \in \coll{\tau \odot \sigma}(x_A,
x_C)$. We know that
\[
\pr_{\sigma'} \circ f \sim^+ (h\vdash k) \pr_\sigma\,,
\]
which unfolds to mean that for each $x^\sigma \in \conf{\sigma}$, the
composite bijection 
\[
  (h \vdash k)(  \pr_\sigma\,x^\sigma) \quad
  \stackrel{h \vdash k}{\bij}
  \quad
  \pr_\sigma\,x^\sigma
\quad \stackrel{\pr_\sigma}{\bij} \quad
x^\sigma
\quad \stackrel{f}{\bij} \quad
f(x^\sigma)
\quad \stackrel{\pr_{\sigma'}}{\bij} \quad
\pr_{\sigma'}(f(x^\sigma))
\]
is a positive symmetry in $A' \vdash B'$. In other words, there is a
negative symmetry $f[x^\sigma]_{A'}$ and a positive symmetry
$f[x^\sigma]_{B'}$ such that the following diagram commutes:
\[
\xymatrix@C=40pt{
x^\sigma
        \ar[r]^{\pr_\sigma}
        \ar[d]_{f}&  
x^\sigma_A \vdash x^\sigma_B         \ar[r]^{h\vdash k}
& h(x^\sigma_A) \vdash k(x^\sigma_B)
        \ar[d]^{f[x^\sigma]_{A'} \vdash f[x^\sigma]_{B'}}\\
f(x^\sigma)
        \ar[rr]_{\pr_{\sigma'}}& &
f(x^\sigma)_{A'} \vdash f(x^\sigma)_{B'}
}
\]
and we have an analogous symmetry induced by the 2-cell $g$. By Proposition
\ref{prop:sync_sym}, there are unique $\varphi^{\sigma'},
\varphi^{\tau'}$ and $\vartheta_{A'}^-, \vartheta_{C'}^+$ s.t.
\[
\xymatrix@R=10pt@C=20pt{
h(x^\sigma_A)
	\ar[r]^{f[x^\sigma]_{A'}}&
f(x^\sigma)_{A'}
        \ar[dd]^{\varphi^{\sigma'}_{A'}}
&f(x^\sigma)_{B'}
        \ar[dd]^{\varphi^{\sigma'}_{B'}}
        \ar[r]^{f[x^\sigma]_{B'}^{-1}}&
k(x_B)     \ar[r]^{g[x^\tau]_{B'}}&
g(x^\tau)_{B'}
        \ar[dd]^{\varphi^{\tau'}_{B'}}&
g(x^\tau)_{C'}
        \ar[dd]^{\varphi^{\tau'}_{C'}}
        \ar[r]^{g[x^\tau]_{C'}^{-1}}&
l(x^\tau_C)\\\\
h(x_A)	\ar[r]_{\vartheta_{A'}^-}
	\ar[uu]^{h(\theta_A^-)}&
y^{\sigma'}_{A'}&
y^{\sigma'}_{B'}
        \ar@{}[r]|{=}&
y_{B'}     \ar@{}[r]|{=}&
y^{\tau'}_{B'}&
y^{\tau'}_{C'}
        \ar[r]_{\vartheta_{C'}^+}&
l(x_C)	\ar[uu]_{l(\theta_C^+)}
}
\]
commutes (the line on the top is secured since $f, g$ are rigid); and
by definition $g \odot f : \tau \odot \sigma \to \tau' \odot \sigma'$
is the unique map such that $(g\odot f)(x^\tau \odot x^\sigma) =
y^{\tau'} \odot y^{\sigma'}$. Thus
\[
\pcomp^{\sigma', \tau'} \circ \coll{g \odot f}(\w^{\tau \odot \sigma})
=
((\vartheta_A^-, y^{\sigma'}, \id), (\id, y^{\tau'}, \vartheta_C^+))\,.
\]

Now, likewise, we have
\begin{eqnarray*}
\coll{f}(\theta_A^-, x^\sigma, \id_{x^\sigma_B}) &=&
(f[x^\sigma]_{A'}\circ h(\theta_A^-), f(x^\sigma), f[x^\sigma]_{B'}^{-1})\\
\coll{g}(\id_{x^\tau_B}, x^\tau, \theta_C^+) &=&
(g[x^\tau]_{B'}, g(x^\tau), l(\theta_C^+) \circ g[x^\tau]_{C'}^{-1})
\end{eqnarray*}
but by the diagram above, writing $\Theta_{B'} = \varphi^{\sigma'}_{B'} \circ
f[x^\sigma]_{B'} = \varphi^{\tau'}_{B'} \circ g[x^\tau]_{B'}$, we have the
equalities
\begin{eqnarray*}
  (\vartheta_{A'}^-, y^{\sigma'}, \id_{y_{B'}}) &=& \Theta_{B'} \cdot
(f[x^\sigma]_{A'}\circ h(\theta_{A'}^-), f(x^\sigma), f[x^\sigma]_{B'}^{-1})\\
(\id_{y_{B'}}, y^{\tau'}, \vartheta_{C'}^+)\cdot \Theta_{B'} &=&
(g[x^\tau]_{B'}, g(x^\tau),  l(\theta_C^+) \circ g[x^\tau]_{C'}^{-1})
\end{eqnarray*}
so that we may now compute
\begin{eqnarray*}
&&((\vartheta_{A'}^-, y^{\sigma'}, \id), (\id, y^{\tau'},
\vartheta_{C'}^+))\\
&=&(\Theta_{B'} \cdot (f[x^\sigma]_{A'}\circ h(\theta_A^-), f(x^\sigma),
f[x^\sigma]_{B'}^{-1}), (\id, y^{\tau'}, 
\vartheta_{C'}^+))\\
&\sim& ((f[x^\sigma]_{A'}\circ h(\theta_A^-), f(x^\sigma),
f[x^\sigma]_{B'}^{-1}), (\id, y^{\tau'},
\vartheta_{C'}^+) \cdot \Theta_{B'})\\
&=& ((f[x^\sigma]_{A'}\circ h(\theta_A^-), f(x^\sigma),
f[x^\sigma]_{B'}^{-1}), (g[x^\tau]_{B'}, g(x^\tau), l(\theta_C^+) \circ g[x^\tau]_{C'}^{-1}))
\end{eqnarray*}
as required to establish the desired commutation.
\end{proof}

Equipped with the above data, the operation $\coll{-}$ is oplax-functorial:

\begin{thm}\label{th:oplax_tcg}
The operation $\coll{-} : \dblTCG \to
  \dblDist$ defines a symmetric monoidal, oplax double functor. 
\end{thm}
\begin{proof}
There are three coherence diagrams to check for functoriality, involving only globular 2-cells. The preservation of
the associator is straightforward. For the preservation of the
unitor, we establish that the diagram of globular 2-cells
\[
\xymatrix{
\coll{\cc_B \odot \sigma}
        \ar[r]^{\rho_\sigma}
        \ar[d]_{\pcomp}&
\coll{\sigma}\\
\coll{\cc_B} \bullet \coll{\sigma}
        \ar[r]_{\pid\bullet \coll{\sigma}}&
\id_B \bullet \coll{\sigma}
        \ar[u]_{\rho_{\coll{\sigma}}}
}
\]
commutes for any $\sigma : A \vdash B$. For this, consider
$
\w = (\theta_A^-, \cc_{x^\sigma_B} \odot x^\sigma, \theta_B^+) \in
\coll{\cc_B \odot \sigma}(x_A, x_B)
$.
We have $\pcomp(\w) = ((\theta_A^-, x^\sigma, \id), (\id,
\cc_{x^\sigma_B}, \theta_B^+))$, sent by $\pid\bullet
\coll{\sigma}$ to $((\theta_A^-, x^\sigma, \id), \theta_B^+)$. Now
\begin{eqnarray*}
((\theta_A^-, x^\sigma, \id), \theta_B^+) &=&
((\theta_A^-, x^\sigma, \id), \id_{x_B} \cdot \theta_B^+)\\
&\sim& (\theta_B^+ \cdot (\theta_A^-, x^\sigma, \id_{x^\sigma_B}),
\id_{x_B})\\
&=& ((\theta_A^-, x^\sigma, \theta_B^+), \id_{x_B})
\end{eqnarray*}
satisfying $\rho((\theta_A^-, x^\sigma, \theta_B^+), \id_{x_B}) =
(\theta_A^-, x^\sigma, \theta_B^+)$ as required. The other coherence
diagram for the unitor is symmetric.

We must additionally show that the double functor has a symmetric
monoidal structure. All the necessary data is defined in the obvious
way: we have a vertical transformation
consisting of vertical isomorphisms 
\[
\coll{A} \otimes \coll{B} \overset{\cong}\longrightarrow \coll{A \otimes B}  
\]
and invertible 2-cells
\[\begin{tikzcd}
	{\coll{A} \otimes \coll{B}} & {\coll{A'} \otimes \coll{B'}} \\
	{\coll{A \otimes B}} & {\coll{A' \otimes B'}}
	\arrow[""{name=0, anchor=center, inner sep=0}, "{\coll{\sigma} \otimes \coll{\tau}}", "\shortmid"{marking}, from=1-1, to=1-2]
	\arrow["\cong"', from=1-1, to=2-1]
	\arrow["\cong", from=1-2, to=2-2]
	\arrow[""{name=1, anchor=center, inner sep=0}, "{\coll{\sigma \otimes \tau}}", "\shortmid"{marking}, from=2-1, to=2-2]
	\arrow[shorten <=14pt, shorten >=14pt, Rightarrow, from=0, to=1]
      \end{tikzcd}\]
    as well as a vertical isomorphism $I \cong \coll{I}$ (we use
    the same $I$ for monoidal units in different categories) and an invertible 2-cell
    \[\begin{tikzcd}
	I & I \\
	{\coll{I}} & {\coll{I}}
	\arrow[""{name=0, anchor=center, inner sep=0}, "{\id_I}", "\shortmid"{marking}, from=1-1, to=1-2]
	\arrow["\cong"', from=1-1, to=2-1]
	\arrow["\cong", from=1-2, to=2-2]
	\arrow[""{name=1, anchor=center, inner sep=0}, "{\coll{\cc_I}}"', "\shortmid"{marking}, from=2-1, to=2-2]
	\arrow[shorten <=14pt, shorten >=14pt, Rightarrow, from=0, to=1]
\end{tikzcd}\]
These must satisfy a small number of coherence axioms (see
\cite[Def.~5.1]{ggv24}), which are all immediately verified.
\end{proof}

At this point, we have defined an oplax double functor preserving the
monoidal structure appropriately. This is not a pseudo double functor:
in general, the compositor 2-cell is not invertible, as witnessed in
the example of Figure~\ref{fig:deadlock}.  We will address this below in
§\ref{sec:collapse_linear}, by specializing the games and strategies
we consider to ensure that no deadlock arises. Before that, we
discuss another orthogonal issue: the collapse above does not preserve
the exponential modality $\oc$. 

\subsection{Difficulties with the exponential modality} 
\label{subsec:problems}
First, we recall how to construct an exponential modality on $\dblTCG$.

\subsubsection{An exponential modality in polarized $\dblTCG$}
For an ess $E$, the ess $\oc E$ is an infinitary symmetric tensor
product, where the elements of the indexing set are called \textbf{copy indices}:
\begin{defi}\label{def:bang_ess}
Consider $E$ an ess. 
Then $\oc E$ has: \emph{events}, $\ev{\oc E} = \mathbb{N}
\times \ev{E}$; \emph{causality and conflict} inherited transparently.
The \emph{isomorphism family} $\tilde{\oc E}$ comprises all bijections
\[
\theta :~\parallel_{i\in \mathbb{N}} x_i 
\bij~
\parallel_{i\in \mathbb{N}} y_i
\]
between configurations
such that there is a bijection $\pi : \mathbb{N} \bij \mathbb{N}$ and
for every $i \in \mathbb{N}$, a symmetry $\theta_i : x_i \sym_A
y_{\pi(i)}$, such that for every $(j, a) \in \ev{\oc E}$, we have
$\theta(j, a) = (\pi(j), \theta_j(a))$.
\end{defi}

To extend this to tcgs, we must treat separately the positive and negative
symmetries. We explained earlier that intuitively, symmetries that only
change the copy indices of negative moves should be negative, and
likewise for positive moves -- but this naive definition does not yield
a tcg in general \cite{DBLP:journals/lmcs/CastellanCW19}. To obtain a
sensible extension of $\oc$ to tcgs, we must
restrict to a polarized setting in which tcgs are  \textbf{negative},
meaning that all minimal events are negative. For a negative tcg $A$, a
symmetry $\theta \in \tilde{\oc A}$ is in the sub-family $\ntilde{\oc
A}$ if each $\theta_i$ in Definition~\ref{def:bang_ess} is negative in
$\tilde{A}$, whereas $\theta$ is in $\ptilde{\oc A}$ if each $\theta_i$
is in $\ptilde{A}$ and \emph{additionally} $\pi$ is the identity
bijection. This extends to a horizontal double comonad on the sub-double-category $\dblTCG^-$
of \emph{negative} tcgs -- this is detailed for instance in
\cite{paquet2020probabilistic}.

\subsubsection{Our functor does not preserve the modality}
\label{subsubsec:problems}
However, in this paper we shall not adopt that exact construction,
because that exponential modality $\oc$ on $\dblTCG^-$ is \emph{not}
preserved by the oplax functor of Theorem \ref{th:oplax_tcg}. We
illustrate with an example:

\begin{exa}
Consider $I$ the empty tcg.

Then, the two groupoids $\tilde{\oc I}$ and $\Sym(\tilde{I})$ are not
equivalent in general. Indeed, $\oc I$ is still empty so that
$\tilde{\oc I}$ is a singleton groupoid. In contrast, $\Sym(\tilde{I})$
includes
\[
\emptyset,~\emptyset\emptyset,~\emptyset\emptyset\emptyset, \dots\,,
\]
\emph{i.e.} countably many non-isomorphic objects.
\end{exa}

Intuitively, the relational model and its relatives such as $\dblDist$ record how many
times we ``do nothing'', whereas $\dblTCG$ only records when we do
something. 
Thus, although one can construct a cartesian closed Kleisli bicategory
from the restriction of $\dblTCG$ to negative games
\cite{paquet2020probabilistic}, the functor $\coll{-}$ will not
preserve cartesian closed structure.

We shall resolve this in §\ref{sec:cc_pseudofunctor} by adopting a
notion of games where not all configurations are considered ``valid''
and correspond to a point of the web in the relational model. Before we
do that, let us address the oplax aspect of our collapse.

\section{Visible strategies and a  pseudofunctor}
\label{sec:collapse_linear}

A fundamental difference between dynamic and static models is the
ability for dynamic models to detect deadlocks, via the causal
nature of strategies. This situation is encapsulated in the ``oplax-ness'' of
the double functor $\coll{-} : \dblTCG \to \dblDist$. In this section,
we show how to restrict
the games and strategies so that deadlocks never occur, which resolves
the mismatch and gives a \emph{pseudo} double functor, that preserves
composition up to iso. 
To perform this restriction, we import from \cite{hdr} the mechanism
of \emph{visibility}. This gives a new double category $\dblVis$. In
this new setting, the symmetric monoidal oplax double functor of
Theorem \ref{th:oplax_tcg} becomes a symmetric monoidal pseudo double
functor $\dblVis \to \dblDist$. The restriction to $\dblVis$ is significant, e.g.~the
game model of mutable state
\cite{DBLP:journals/corr/abs-2103-15453,hdr} is no longer included, but
we retain a model of the $\lambda$-calculus. 

\subsection{The double category $\dblVis$}
\label{subsec:bicategory-vis}

We first introduce the restricted double category $\dblVis$.

Games and strategies in $\dblTCG$ are very general, and mostly
independent of the specific computational paradigms they represent. In
contrast, in $\dblVis$, all games are close to those obtained by the
interpretation of simple types, and strategies are somewhat close to
those needed to model $\lambda$-terms. (We do have a bit more: for
instance, $\dblVis$ supports pure parallel higher-order computation
\cite{lics15}.)

\subsubsection{Arenas}
The objects of our refined model are called \emph{arenas}. Arenas
narrow down the causal structure to an \emph{alternating forest},
required for the definition of visible strategies.

\begin{defi}
\label{def:arena}
  An \textbf{arena} is a tcg $A$
such that
\[
\begin{array}{rl}
\text{\emph{(1)}} &
\text{if $a_1, a_2 \leq_A a_3$ then $a_1 \leq_A a_2$ or $a_2 \leq_A
a_1$,}\\
\text{\emph{(2)}} &
\text{if $a_1 \imc_A a_2$, then $\pol_A(a_1) \neq \pol_A(a_2)$}.
\end{array}
\]

Moreover, $A$ is called a \textbf{$-$-arena} if $A$ is negative as a
tcg.
\end{defi}
\noindent 
The dual, tensor and hom operations on tcgs preserve arenas.

A key consequence of this definition is that any non-minimal $a \in A$
has a unique causal predecessor, called its \textbf{justifier}, and
denoted $\just(a) \in A$ with $\just(a) \imc_A a$.

\subsubsection{Visible strategies}
 Visibility captures a property of purely-functional parallel programs,
in which threads may fork and join but each should be a well-formed
stand-alone sequential execution. 
In an event structure $E$, a thread is 
formalized as a \textbf{grounded causal chain (gcc)}, \emph{i.e.} 
a finite set $\rho \subseteq_f \ev{E}$ on which $\leq_E$ is a total
order, forming a sequence 
\[
\rho_1 \imc_E \dots \imc_E \rho_n
\]
where $\rho_1$ is minimal in $E$. We write $\gcc(E)$ for the set of gccs. A gcc
need not be a configuration (although it will always be if the strategy
interprets a sequential program). A strategy is
\emph{visible} if gccs only reach valid states of the arena:

\begin{defiC}[\cite{lics15,hdr}]
Consider $A$ an arena. Then $\sigma : A$ is \textbf{visible} if it is: 
\[
\begin{array}{rl}
\text{\emph{pointed:}}&
\text{for any $s \in \sigma$, there is a unique $\init(s)
\leq_\sigma s$ minimal in $\sigma$, which is negative,}\\
\text{\emph{valid-gccs:}}&
\text{for all $\rho \in \gcc(\sigma)$, $\pr_{\sigma}(\rho) \in
\conf{A}$.}
\end{array}
\]
\end{defiC}
\noindent 
A key consequence of this definition is that it equips all non-initial
moves of $\sigma$ with a \emph{justifier}, analogously to Hyland-Ong
games: if $s \in \sigma$ is such that $\pr_\sigma\,s$ is non-initial,
then $\just(s)$ is the (unique) $s' \in \sigma$ such that
$\pr_\sigma\,s' \imc_A \pr_\sigma\,s$; if $\pr_\sigma\,s$ is initial
but positive, we set $\just(s) = \init(s)$. Then, \emph{visibility}
entails the following fact: for any gcc in $\sigma$ of the form
\[
\rho = \rho_1 \imc_E \dots \imc_E \rho_n^+\,,
\]
the justifier $\just(\rho_n^+)$ of $\rho_n^+$ must appear in $\rho$ --
this is clearly analogous to the notion of \emph{visibility} from
Hyland-Ong games \cite{DBLP:journals/iandc/HylandO00}, explaining the terminology.
The key property of visible strategies for us here is that their
composition is always deadlock-free \cite[Lemma 10.4.8]{hdr}:

\begin{lem}\label{lem:deadlock_free}
Consider $A, B, C$ $-$-arenas, $\sigma : A \vdash B$ and $\tau : B
\vdash C$ visible strategies, $x^\sigma \in \conf{\sigma}$ and
$x^\tau \in \conf{\tau}$ with a symmetry $\theta : x^\sigma_B \sym_B
x^\tau_B$.\\
\indent
Then, $x^\sigma, \theta, x^\tau$ are causally compatible in the sense
of Proposition \ref{prop:sync_sym}.
\end{lem}
\noindent 
Copycat strategies on $-$-arenas are visible, and visible strategies
are closed under composition \cite{hdr}, so we may consider the
sub-double-category  of $\dblTCG$ over the $-$-arenas and the visible
strategies, containing all vertical morphisms and 2-cells between
them. We call this double category $\dblVis$, and as usual we set
$\Vis = \hori(\dblVis)$, the bicategory of $-$-arenas, visible
strategies, and (globular) positive morphisms.

\subsubsection{A pseudo double functor} \label{subsubsec:a_pseudofunctor}
By restriction of the components from Theorem
\ref{th:oplax_tcg}, we have 
an oplax double functor
\[
\coll{-} : \dblVis \to \dblDist.
\]

This is actually a pseudo double functor, in that the compositor is invertible:

\begin{thm}\label{th:pseudo_vis_dist}
  We have a pseudo double functor $\coll{-} : \dblVis \to \dblDist$. 
\end{thm}
\begin{proof}
We must show that for all $x_A \in \conf{A}$ and $x_C \in \conf{C}$,
$\pcomp^{\sigma, \tau}(x_A, x_C)$ is a bijection. 

For surjectivity, consider
\[
\begin{array}{rcrcl}
\w^\sigma &=& (\theta_A^-, x^\sigma, \theta_B^+) 
        &\in& \coll{\sigma}(x_A, x_B)\\
\w^\tau &=& (\theta_B^-, x^\tau, \theta_C^+) 
        &\in& \coll{\tau}(x_B, x_C)\\
\end{array}
\]
composable witnesses. By Lemma \ref{lem:deadlock_free}, 
$(x^\sigma, \theta_B^- \circ \theta_B^+, x^\tau)$
is causally compatible. By Proposition \ref{prop:sync_sym},
there
are unique $y^\tau \odot y^\sigma \in \confp{\tau \odot \sigma}$ along
with $\varphi^\sigma, \varphi^\tau, \vartheta_A^-, \vartheta_C^+$ such
that:
\[
\begin{tikzcd}[row sep=0.8em, column sep=1.5em]
  &  x^\sigma_A \ar[dd, "\varphi^\sigma_A"] & x^\sigma_B \ar[dd, "\varphi^\sigma_B"] \ar[r, "\theta_B^+"] & x_B \ar[r, "\theta_B^-"] & x^\tau_B \ar[dd, "\varphi^\tau_B"] & x^\tau_C \ar[dd, "\varphi^\tau_C"'] \ar[dr, "\theta_C^+"] \\
  x_A \ar[ur, "\theta_A^-"] \ar[dr, "\vartheta_A^-"'] & & & & & & x_C \\
  & y^\sigma_A & y^\sigma_B \ar[r, Rightarrow, no head]& y_B \ar[r,
  Rightarrow, no head] & y^\tau_B & y^\tau_C \ar[ur,
  "\vartheta_C^+"'] &
\end{tikzcd}
\]
which, writing $\Theta_B = \varphi^\sigma_B \circ {\theta_B^+}^{-1} =
\varphi^\tau_B \circ \theta_B^-$, entails
\[
\begin{array}{rcrcl}
\v^\sigma &=& (\vartheta_A^-, y^\sigma, \id_{y_B}) &=& \Theta_B \cdot (\theta_A^-, x^\sigma, \theta_B^+)\\
\v^\tau &=& (\id_{y_B}, y^\tau, \vartheta_C^+) &=& (\theta_B^-, x^\tau, \theta_C^+) \cdot \Theta_B
\end{array}
\]
so $(\v^\sigma, \v^\tau) = (\Theta_B \cdot \w^\sigma, \v^\tau)
\sim (\w^\sigma, \v^\tau \cdot \Theta_B) = (\w^\sigma, \w^\tau)$.
Now $(\v^\sigma, \v^\tau) = \pcomp^{\sigma, \tau}(\vartheta_A^-,
y^\tau\odot y^\sigma, \vartheta_C^+)$, showing surjectivity.

Now, for injectivity, consider two witnesses 
\[
(\theta_A^-, x^\tau \odot x^\sigma, \theta_C^+) \in \coll{\tau \odot
\sigma}(x_A, x_C)\,,
\qquad
(\vartheta_A^-, y^\tau \odot y^\sigma, \vartheta_C^+) \in
\coll{\tau \odot \sigma}(x_A, x_C)\,,
\]
s.t.
$\pcomp(\theta_A^-, x^\tau \odot x^\sigma, \theta_C^+) \sim 
\pcomp(\vartheta_A^-, y^\tau \odot y^\sigma, \vartheta_C^+)$,
\emph{i.e.}, there are components such that
\begin{eqnarray*}
((\theta_A^-, x^\sigma, \id), (\id, x^\tau, \theta_C^+))
&=& ((\theta_A^-, x^\sigma, \id), (\id, y^\tau, \vartheta_C^+)
\cdot \Theta_B)\\
((\vartheta_A^-, y^\sigma, \id), (\id, y^\tau, \vartheta_C^+))
&=& (\Theta_B \cdot (\theta_A^-, x^\sigma, \id), (\id, y^\tau,
\vartheta_C^+))
\end{eqnarray*}
which by definition of the functorial action, means that
\[
\xymatrix@R=10pt@C=10pt{
&x^\sigma_A
        \ar[dd]^{\varphi^\sigma_A}&
x^\sigma_B
        \ar[dd]^{\varphi^\sigma_B}
        \ar@{}[r]|=&
x_B     \ar@{}[r]|=
        \ar[dd]|{\Theta_B}&
x^\tau_B
        \ar[dd]_{\varphi^\tau_B}&
x^\tau_C
        \ar[dd]_{\varphi^\tau_C}
        \ar[dr]^{\theta_C^+}\\
x_A     \ar[ur]^{\theta_A^-}
        \ar[dr]_{\vartheta_A^-}&&&&&&
x_C\\
&y^\sigma_A&
y^\sigma_B
        \ar@{}[r]|=&
y_B     \ar@{}[r]|=&
y^\tau_B&
y^\tau_C
        \ar[ur]_{\vartheta_C^+}
}
\]
commutes for some $\varphi^\sigma : x^\sigma \sym_\sigma y^\sigma$
and $\psi^\tau : x^\tau \sym_\tau y^\tau$. So
\[
\varphi^\tau \odot \varphi^\sigma : x^\tau \odot x^\sigma \sym_{\tau
\odot \sigma} y^\tau \odot y^\sigma
\]
has a positive display, hence is an identity symmetry by condition
\emph{(2)} of Definition \ref{def:thin} -- thus from the diagram,
$(\theta_A^-, x^\tau \odot x^\sigma, \theta_C^+) = 
(\vartheta_A^-, y^\tau \odot y^\sigma, \vartheta_C^+)$
as needed to conclude.
\end{proof}

\subsection{Symmetric monoidal structure.}
\label{sec:symm-mono-struct}
Visible strategies are closed under the tensor product $\otimes$, and
so the symmetric monoidal structure in $\dblTCG$ restricts to
$\dblVis$. Since all components restrict, we deduce that $\coll{-} :
\dblVis \to \dblDist$ is a symmetric monoidal (pseudo) double
functor.

Applying Theorem~\ref{thm:shulman}, we obtain a symmetric
monoidal pseudofunctor between the induced horizontal bicategories. 

\begin{thm}\label{thm:main_sec6}
The horizontal restriction of the collapse double functor $\coll{-}$
is a symmetric monoidal pseudofunctor of bicategories  $\coll{-} : \Vis \to \TCG$. 
\end{thm}
\noindent 
We emphasize that this theorem makes essential use of the double
categorical aspects, which makes the 2-dimensional symmetric monoidal
structure manageable. Now that this is established, we can focus
directly on the bicategorical structure, as we move
towards applications in the semantics of the $\lambda$-calculus.

\section{A Cartesian closed pseudofunctor}
\label{sec:cc_pseudofunctor}

The paper until this point has focused on the preservation of
symmetric monoidal structure by our collapse functor. In this section
we add further structure to $\Vis$ to ensure the preservation of the
exponential modality, so that our collapse lifts to the Kleisli
bicategory; we then show that the corresponding pseudofunctor is
cartesian closed. This development is mostly 
independent from the preservation of the linear monoidal structure
detailed above; mainly by lack of a mature bicategorical theory of
models of linear logic.

To ensure preservation of the exponential modality, we
must first address the mismatch identified in Section~\ref{subsec:problems}.
We shall do that now, by introducing a mechanism that identifies those
positions in games that correspond to points of the relational model.

\subsection{Payoff and winning} The additional mechanism we import here
is inspired from Melliès \cite{DBLP:conf/lics/Mellies05}; see also the
detailed construction in \cite{hdr}. As mentioned above, it helps in
ignoring intermediate configurations arising in games which are not
\emph{complete}, \emph{i.e.} they do not correspond to a valid state in
relational models. In addition, this mechanism makes the ambient games
model \emph{linear} rather than \emph{affine}; it forces strategies
to explore all available resources, solving another mismatch between
games and relational models.

\subsubsection{The bicategory $\WVis$.} Here we construct a new
bicategory of games and strategies, integrating both mechanisms of
\emph{visibility} and \emph{winning}.

\paragraph{Games with payoff.}
Our new notion of game is that of a \emph{board}:

\begin{defi}\label{def:board}
A \textbf{board} is an arena $A$ along with
$\kappa_A : \conf{A} \to \{-1, 0, +1\}$
a \textbf{payoff function}, such that this data satisfies the following
conditions:
\[
\begin{array}{rl}
\text{\emph{invariant:}} & \text{for all $\theta : x \sym_A y$, we have
$\kappa_A(x) = \kappa_A(y)$,}\\
\text{\emph{race-free:}} & \text{for all $a \mconflict_A\,a'$, we have
$\pol_A(a) = \pol_A(a')$,}
\end{array}
\]
where $a \mconflict_A\,a'$ means an \textbf{immediate
conflict}, \emph{i.e.} $a \conflict_A a'$ and it is not inherited.

A \textbf{$-$-board} is additionally negative as an arena, and must also satisfy:
\[  
\begin{array}{rl}
\text{\emph{initialized:}} & \kappa_A(\emptyset) \geq 0\,.
\end{array}
\]
\noindent 
Finally, a $-$-board $A$ is \textbf{strict} if
$\kappa_A(\emptyset) = 1$ and all its initial moves are in pairwise
conflict. It is \textbf{well-opened} if it is strict
with exactly  one initial move.
\end{defi}

The function $\kappa_A$ assigns a value to each configuration.
Configurations with payoff $0$ are called \textbf{complete}: they
correspond to \emph{terminated} executions, which have reached an
adequate stopping point where all calls have adequately
returned -- we write $\nconf{A}$ for the set of complete configurations
on $A$.
Otherwise, $\kappa_A$ assigns a responsibility for non-completeness. If
$\kappa_A(x) = -1$ then Player is responsible, otherwise it is
Opponent. 

The earlier constructions on arenas specialize into constructions on
boards. If $A$ is a board, then the \textbf{dual} $A^\perp$ has payoff
$\kappa_{A^\perp} = -\kappa_A$. If $A$ and $B$ are boards, then we set
the \textbf{tensor} (\emph{resp.} the \textbf{par}) as having as
underlying arena the tensor, and payoff $\kappa_{A\tensor B}(x_A
\tensor x_B) = \kappa_A(x_A) \tensor \kappa_B(x_B)$ (\emph{resp.}
$\kappa_{A\parr B}(x_A \tensor x_B) = \kappa_A(x_A) \parr
\kappa_B(x_B)$), using the operations on payoff defined in
\begin{figure}
\[
\begin{array}{c|S[table-format=1.0]S[table-format=1.0]S[table-format=1.0]}
\toprule
\tensor & \cellcolor{red!25}{-1} & \cellcolor{gray!25}{0} &
\cellcolor{blue!25}{1}\\
\hline
\cellcolor{red!25}{-1} & \cellcolor{red!25}{-1} &
\cellcolor{red!25}{-1} & \cellcolor{red!25}{-1}\\
\cellcolor{gray!25}{0} & \cellcolor{red!25}{-1} &
\cellcolor{gray!25}{0} & \cellcolor{blue!25}{1}\\
\cellcolor{blue!25}{1} & \cellcolor{red!25}{-1} &
\cellcolor{blue!25}{1} & \cellcolor{blue!25}{1}\\
\bottomrule
\end{array}\qquad
\begin{array}{c|S[table-format=1.0]S[table-format=1.0]S[table-format=1.0]}
\toprule
{\parr} & \cellcolor{red!25}{-1} & \cellcolor{gray!25}{0} &
\cellcolor{blue!25}{1}\\
\hline
\cellcolor{red!25}{-1} & \cellcolor{red!25}{-1} &
\cellcolor{red!25}{-1} & \cellcolor{blue!25}{1}\\
\cellcolor{gray!25}{0} & \cellcolor{red!25}{-1} &
\cellcolor{gray!25}{0} & \cellcolor{blue!25}{1}\\
\cellcolor{blue!25}{1} & \cellcolor{blue!25}{1} &
\cellcolor{blue!25}{1} & \cellcolor{blue!25}{1}\\
\bottomrule
\end{array}
\]
\caption{Payoff tables for operations on arenas, with $A\parr B =
(A^\perp
  \tensor B^\perp)^\perp$.}
\label{fig:def_payoff}
\end{figure}
Figure \ref{fig:def_payoff}. If $A, B$ are $-$-boards, then the
\textbf{hom-board} is defined as $A\vdash B = A^\perp \parr B$,
\emph{i.e.} with $\kappa_{A\vdash B} = \kappa_{A^\perp}(x_A) \parr
\kappa_B(x_B)$.

Note that by definition of payoff, the order-isomorphism of
\eqref{eq:isotensor} refines to bijections:
\begin{eqnarray}
- \tensor - \quad:\quad \nconf{A} \times \nconf{B} &\iso&
  \nconf{A\tensor B}\,,\\ \label{eq:comptensor}
- \parr - \quad:\quad \nconf{A} \times \nconf{B} &\iso& 
  \nconf{A \parr B}\,.
\end{eqnarray}

\paragraph{Winning strategies.} We turn to strategies. The payoff is taken into account in the definition, as follows:

\begin{defi}\label{def:strat_winning}
Consider $A$ a board, and $\sigma : A$ a visible strategy. 

We say that $\sigma$ is \textbf{winning} if for all $x^\sigma \in
\confp{\sigma}$, we have $\kappa_A(\pr_\sigma\,x^\sigma) \geq 0$.
\end{defi}

\noindent 
It is rather easy to show that copycat strategies are winning, and that
winning strategies are stable under composition \cite[Proposition
8.2.15]{hdr}. Thus we obtain a bicategory $\WVis$ with objects the
$-$-boards, morphisms from $A$ to $B$ the winning strategies on $A
\vdash B$, and $2$-cells the positive morphisms.

\subsubsection{Adjusting the pseudofunctor} We now show how to adjust
the pseudofunctor of Section~\ref{subsubsec:a_pseudofunctor} to account
for the new structure we have just introduced. 

Firstly, as before, to each $-$-board $A$ we shall associate a
groupoid. However, it shall not be the groupoid $\tilde{A}$ of
symmetries between \emph{all} configurations as before, instead we set
$\coll{A}$ as the groupoid $\comp{A}$ having as objects the set $\nconf{A}$ of
configurations of \emph{null payoff}, with morphisms from $x_A$ to
$y_A$ still comprising all the symmetries in $A$. 

The same definition as in Section~\ref{subsubsec:strat_2_dist} now
yields a distributor
\[
\coll{\sigma} : \coll{A}^{\op} \times \coll{B} \to \Set
\]
with those restricted groupoids -- we keep the same definition for
$\coll{\sigma}$, which should cause no confusion as this notation shall
remain fixed until the end of the paper. As before, we have:

\begin{thm}\label{th:pseudo_vws_dist}
  We have $\coll{-} : \WVis \to \Dist$ a pseudofunctor.
\end{thm}
\begin{proof}
The new definition of $\comp{A}$ imposes only one new proof obligation,
namely that for $\sigma : A \vdash B$, $\tau : B \vdash C$, $x_A \in
\comp{A}$ and $x_C \in \comp{B}$, then
\[
\pcomp^{\sigma, \tau}_{x_A, x_C} : \coll{\tau \odot \sigma}(x_A, x_C)
\to (\coll{\tau} \bullet \coll{\sigma})(x_A, x_C)
\]
as defined in the proof of Proposition \ref{prop:compositor}, is still
well-defined. Indeed, it sends a witness $(\theta_A^-, x^\tau \odot
x^\sigma, \theta_C^+) \in
\coll{\tau \odot \sigma}(x_A, x_C)$ to (the equivalence
class of) the pair
\[
((\theta_A^-, x^\sigma, \id_{x_B}), (\id_{x_B}, x^\tau, \theta_C^+))
\in (\coll{\tau} \bullet \coll{\sigma})(x_A, x_C)
\]
for $x^\sigma_B = x^\tau_B = x_B$, exploiting that $x^\sigma$ and
$x^\tau$ are matching; but this assumes that $(\theta_A^-, x^\sigma,
\id_{x_B}) \in \coll{\sigma}(x_A, x_B)$ and $(\id_{x_B}, x^\tau,
\theta_C^+) \in \coll{\tau}(x_B, x_C)$, which only makes sense provided
$x_B$ has null payoff (and thus lies in $\coll{B}$).
Thus seeking a contradiction, assume that $\kappa_B(x_B) = 1$. But then
$\kappa_{B^\perp \parr C}(\pr_\tau(x^\tau)) = -1 \parr 0 = -1$, which is
impossible since
$x^\tau \in \confp{\tau}$ and $\tau$ is winning. Symmetrically if
$\kappa_B(x_B) = -1$ then this contradicts that $\sigma$ is winning
since $x^\sigma \in \confp{\sigma}$. Hence, $\kappa_B(x_B) = 0$ as
required.

The other components of the pseudofunctor remain unchanged.
\end{proof}

But we are not just interested in $\WVis$ and $\Dist$: we need
a pseudofunctor relating the Kleisli bicategories for the
corresponding exponential modalities in the two models.

\subsection{The exponential modality on $\WVis$.}
\label{subsec:exp_wvis}

First, we detail the construction of the exponential modality on
$\WVis$ (written $\oc$), along with its algebraic structure.

\subsubsection{The exponential modality for boards} The construction
of $\oc A$ as a countable symmetric tensor of copies of $A$
(Section~\ref{subsec:problems}) can be extended to $-$-boards. Note that any
configuration of $\oc A$ has a representation as $\parallel_{i\in I} x_i$
for $I \subseteq_f \mathbb{N}$, and this representation is unique if we
insist that every $x_i$ is non-empty. Using that, we set (all $x_i$
below are non-empty): 
\[
\begin{array}{rcrclcl}
\kappa_{\oc A} &:& \conf{\oc A} &\to& \{-1, 0, +1\}\\
&& \parallel_{i\in I} x_i &\mapsto& \bigotimes_{i\in I}
\kappa_A(x_i)
\end{array}
\]
that is well-defined because $\tensor$ is associative on $\{-1, 0,
+1\}$. If $A$ is a $-$-board, then this results in a $-$-board; but
$\oc A$ is never strict, even if $A$ is strict.

\subsubsection{Strict boards} To precisely capture the
relationship between $\oc$ and $\Sym$ we use \emph{strict} boards
(Def.~\ref{def:arena}), where $\emptyset$ has payoff 1 and is
not considered complete. Indeed, the situation with the empty
configuration was at the heart of the issue in
Section~\ref{subsubsec:problems}. 
We now state the following key property: for strict boards, the two
constructions $\Sym$ and $\oc$ are equivalent.

\begin{prop}\label{prop:eq_oc_sym}
Consider a strict board $A$. 
There is an adjoint equivalence of categories:
\[L^{\oc}_A : \Sym(\comp{A}) \simeq \comp{\oc A}
: R^{\oc}_A.\]
\end{prop}
\begin{proof}
We first show that we can identify the objects of $\comp{\oc A}$ with
families $(x_i)_{i \in I}$ of objects of $\comp{A}$, where $I$ is a
finite subset of natural numbers. We observed above that any
configuration $x \in \conf{\oc A}$ can be uniquely written as $x =
\parallel_{i\in I} x_i$ where $I \subseteq_f \mathbb{N}$, and $x_i \in
\conf{A}$ is non-empty for all $i \in I$; yielding a family
$(x_i)_{i\in I}$. In addition, as $\kappa_{\oc A}(x) = 0$, by
definition of the tensor on payoff values we must have  $\kappa_A(x_i)
= 0$ for every $i\in I$, \emph{i.e.} $x_i \in \comp{A}$.
This representation of $x \in \comp{\oc A}$ as a family $(x_i)_{i\in
I}$ is clearly injective. Surjectivity boils down to the fact that each
$x_i$ is non-empty, which follows from $\kappa_A(x_i) = 0$ since $A$ is
strict.

Thus, from right to left, $R^{\oc}_A$ sends $(x_i)_{i\in I}$ to the
sequence $x_{i_1} \dots x_{i_n}$ for $I = \{i_1, \dots, i_n\}$ sorted
in increasing order. From left to right, $L^{\oc}_A$ sends $x_1 \dots x_n$ to
$(x_i)_{i \in \{1\dots n\}}$.
\end{proof}
\noindent 
Although this equivalence only holds for \emph{strict} arenas, it shall
suffice for our purposes.

From now on, we use implicitly and without further mention the
representation of complete configurations of $\oc A$ as families of
complete configurations of $A$.

\subsubsection{Relative pseudocomonads} The pseudofunctor $\WVis \to
\Dist$ does not preserve the 
exponential modality as a pseudocomonad on $\WVis$, but as a 
pseudocomonad \emph{relative} to the sub-bicategory of \emph{strict}
arenas. We recall the categorical notions. 

Recall that a monad on category $\C$ relative to a
functor $J: \D \to \C$ is a functor $T : \D \to \C$ with a
restricted monadic structure, which we can use to form a Kleisli
category $\C_T$ with objects those of $\D$. (Often, $\D$ is a sub-bicategory of $\C$ and $J$ is the inclusion functor.)
This  generalizes to
relative pseudomonads \cite{fiore2018relative} and pseudocomonads:
\begin{defi}
\label{def:relative-pseudocomonad}
Consider $J : \C \to D$ a pseudofunctor between bicategories.

\noindent 
A \textbf{relative pseudocomonad} $T$ over $J$ consists of:

\noindent \emph{(1)} an object $T X \in \D$, for every $X \in \C$,

\noindent \emph{(2)} a family of functors $(-)^*_{X, Y} : \D[TX, JY]
\to \D[TX, TY]$,

\noindent \emph{(3)} a family of morphisms $i_X \in \D[TX, JX]$,

\noindent \emph{(4)} a natural family of invertible $2$-cells, 
for $f \in \D[TX, JY]$ and $g \in \D[TY, JZ]$:
\[
\mu_{f, g} : (g \circ f^*)^* \stackrel{\iso}{\Rightarrow} g^* \circ f^*
\]

\noindent \emph{(5)} a natural family of invertible $2$-cells, for $f
\in \D[TX, JY]$:
\[
\eta_f : f \stackrel{\iso}{\Rightarrow} i_Y \circ f^*
\]

\noindent \emph{(6)} a family of invertible $2$-cells
$\theta_X : i_X^* \stackrel{\iso}{\Rightarrow} \id_{TX}$,
where $X, Y, Z$ range over objects of $\C$, subject to the
coherence conditions in Figure \ref{fig:coherence_relpscom}.
\end{defi}

\begin{figure}
\[
\xymatrix{
&(h \circ (g \circ f^{*})^{*})^{*}
        \ar[dr]^{\mu_{g \circ f^{*},h}}
        \ar[dl]_{(h \circ \mu_{f, g})^{*}}\\
(h \circ (g^{*} \circ f^{*}))^{*}
        \ar[d]_{\iso}&&
(h^{*} \circ (g \circ f^{*})^{*})  
        \ar[d]^{h^{*} \circ \mu_{f, g}}\\
((h \circ g^{*}) \circ f^{*})^{*}
        \ar[d]_{\mu_{f, h\circ g^{*}}}&&
h^{*} \circ (g^{*} \circ f^{*})
        \ar[d]^{\iso}\\
(h \circ g^{*})^{*} \circ f^{*}
        \ar[rr]_{\mu_{h, g} \circ f^{*}}&&
(h^{*} \circ g^{*}) \circ f^{*}
}
\]
\[
\xymatrix@C=40pt{
f^{*}
        \ar[r]^{\eta_f^*}
        \ar@/_/[drr]_{\iso}&
(i_Y \circ f^{*})^{*}
        \ar[r]^{\mu_{f, i_Y}}&
i_Y^{*} \circ f^{*}
        \ar[d]^{\theta_Y \circ f^{*}}\\
&&\id_{TY} \circ f^{*}
}
\]
\caption{Coherence conditions for relative pseudocomonads}
\label{fig:coherence_relpscom}
\end{figure}
\noindent 
Those conditions are just what is needed to form a Kleisli
bicategory written $\D_T$, with \emph{objects} those of $\C$,
\emph{morphisms} and \emph{$2$-cells} from $X$ to $Y$ the category
$\D[TX, JY]$.  We can compose $f \in \D[TX, JY]$ and $g \in
\D[TY, JZ]$ as $g \circ_T f = g \circ f^*$, and
the \emph{identity} on $X$ is $i_X$.

\subsubsection{The exponential relative pseudocomonad} We now form a
concrete relative pseudocomonad on $\WVis$. Here, $\C$ is
the sub-bicategory $\WVis_s$ of \emph{strict} arenas, and $J : \WVis_s
\hookrightarrow \WVis$ the embedding. Note that even if $A$ is strict, 
$\oc A$ is not strict, and so $\oc$ is not an endo(pseudo)functor;
instead we have $\oc : \WVis_s \to \WVis$.

We now outline the components in Definition~\ref{def:relative-pseudocomonad}. For the component \emph{(2)}, we must introduce additional
notions. Fix an
injection $\tuple{-, -} : \mathbb{N}^2 \to \mathbb{N}$. If $I \subseteq_f \mathbb{N}$ and $J_i \subseteq_f
\mathbb{N}$ for all $i \in I$, write $\bigsqcup_{i\in I} J_i \subseteq_f
\mathbb{N}$ for the set of all $\tuple{i,j}$ for $i \in I$ and $j \in
J_i$. Then, we may define:

\begin{defi}
Consider $\sigma \in \WVis[\oc A, B]$. 
\noindent 
Its \textbf{promotion} $\sigma^{\oc}$ has ess $\oc \sigma$, and display
map the unique map of ess such that 
\begin{eqnarray}
\pr_{\sigma^{\oc}}((x^{\sigma,i})_{i\in I})
= 
(x^{\sigma,i}_{A,j})_{\tuple{i,j}\,\in\,\bigsqcup_{i\in I} J_i} \vdash
(x^{\sigma,i}_B)_{i\in I}
\label{eq:dec_display_prom}
\end{eqnarray}
where $\pr_\sigma(x^{\sigma,i}) = (x^{\sigma,i}_{A, j})_{j\in J_i}
\vdash x^{\sigma,i}_B$ for all $i \in I$.
\end{defi}

For \emph{(3)}, the \textbf{dereliction} $\der_A \in \WVis[\oc A, A]$ on
strict $A$ has ess $\cc_A$, and display map $\pr_{\der_A}(\cc_x)
= (x)_{\{0\}} \vdash x$. For \emph{(4)}, we shall use a positive
isomorphism
\[
\join_{\sigma, \tau} : (\tau \odot \sigma^{\oc})^{\oc}
\stackrel{\iso}{\Rightarrow}
\tau^{\oc} \odot \sigma^{\oc}
\]
sending $(x^\tau_i \odot (x^\sigma_{i,j})_{j\in J_i})_{i\in I}$ to
$(x^\tau_i)_{i \in I} \odot (x^\sigma_{i,j})_{\tuple{i,j} \in
\bigsqcup_{i\in I} J_i}$. For \emph{(5)}, given $B$ strict and $\sigma \in
\WVis[\oc A, B]$ we have a positive iso  $\runit_\sigma : \sigma \iso
\der_B \odot \sigma^{\oc}$ sending $x^\sigma \in \confp{\sigma}$ to
$\cc_{x^\sigma_B} \odot (x^\sigma)_{\{0\}} \in \confp{\der_B \odot
\sigma^{\oc}}$. Finally, for \emph{(6)} we have
$\lunit_A : \der_A^{\oc} \iso \cc_{\oc A}$ sending $(\cc_{x_i})_{i\in
I}$ to $\cc_{(x_i)_{i\in I}}$. For all those, we use the fact that
positive isos are entirely determined by their action on $+$-covered
configurations, see \emph{e.g.} \cite[Lemma 7.2.11]{hdr}.

Altogether, this gives us: 

\begin{thm}
The components described above define a pseudocomonad $\oc$ relative to
the embedding of $\WVis_s$ into $\WVis$.
\end{thm}
\noindent 
The naturality and coherence laws follow from lengthy but direct
calculations, which we omit. 
In particular, there is a Kleisli bicategory $\WVis_{\oc}$ whose
objects are strict boards -- in the next section we shall see that this
is a cartesian closed bicategory.  

\subsection{The exponential modality on $\Dist$.} For the sake of
relating $\WVis$ and $\Dist$, we also need to properly introduce the
exponential modality structure of $\Sym(-)$ on $\Dist$, which we shall
also present as a relative pseudocomonad to ease the correspondence. 

\subsubsection{The groupoid $\Sym(A)$.}
For $A$ a groupoid, the objects of $\Sym(A)$ are sequences of objects of
$A$, written $\seq{a_1, \dots, a_n}$, or just $\seq{a_i}_{i\in n}$ --
implicitly treating $n$ as the set $\{1, \dots, n\}$.
There can be a morphism from $\seq{a_i}_{i\in n}$ to $\seq{b_j}_{j\in
p}$ only when $n = p$, in which case it is a permutation
$\pi$ on $n$, along with a family comprising $f_i \in A(a_{\pi(i)},
b_i)$ for each $i\in n$, written
\[
\seq{f_i}^\pi_{i\in n} \in \Sym(A)(\seq{a_i}_{i\in n},
\seq{b_j}_{j \in n})\,,
\]
which we write simply as $\seq{f_i}_{i\in n}$ when $\pi$ is the
identity. We insist that we have $f_i \in A(a_{\pi(i)}, b_i)$ rather
than $f_i \in A(a_i, b_{\pi(i)})$: the family $(f_i)_{i\in n}$ is
thought of as indexed by the codomain, not the domain. This distinction
will make calculations easier later on. 

We shall also write $\seq{a_{i,j}}_{i\in n, j \in p_i} \in \Sym(A)$
for the concatenated sequence $\vec{a_{1, j}} \dots \vec{a_{n,j}}$.

\subsubsection{Promotion and dereliction in $\Dist$.} Next, we present the
relative pseudocomonad structure on $\Dist$. We follow Definition
\ref{def:relative-pseudocomonad}, even though this is not really
relative; it shall be a pseudocomonad ``relative'' to the inclusion of
$\Dist$ in $\Dist$. For this structure, we shall keep the same
notations as in Definition \ref{def:relative-pseudocomonad}, to avoid
notational collisions with the relative pseudocomonad in games.

We start with promotion. Recall that if $\alpha \in \Dist(\Sym(A), B)$,
the coend formula
\[
\alpha^{\dagger}(\vec{a}, \seq{b_1,\dots,b_n}) = \int^{\vec{a_1},\dots,
\vec{a_n}} \left(\prod_{i=1}^n \alpha(\vec{a_i}, b_i) \right)\times
\Sym(A)[\vec{a_1} \dots \vec{a_n}, \vec{a}]
\]
yields a distributor $\alpha^{\dagger} \in \Dist[\Sym(A), \Sym(B)]$,
called the \textbf{promotion} of $\alpha$. Concretely, this means that
witnesses in $\alpha^{\dagger}(\vec{a}, \seq{b_1,\dots,b_n})$ consist
in the choice of three components
\begin{eqnarray}
(\vec{a_1}, \dots, \vec{a_n} \in \Sym(A),
\quad
\seq{s_i}_{i\in n} \in \prod_{i=1}^n \alpha(\vec{a_i}, b_i)\,,
\quad
f \in \Sym(A)[\vec{a_1} \dots \vec{a_n}, \vec{a}])\,,
\label{eq:witdef}
\end{eqnarray}
subject to the following equivalence relation, for each $(f_i \in
\Sym(A)[\vec{a_i}, \vec{a'_i}])_{1\leq i \leq n}$:
\begin{eqnarray}
(\vec{a_1}, \dots, \vec{a_n}\,,~
\seq{s_i}_{i\in n}\,,~
f \circ (f_1\dots f_n))
&\sim&
(\vec{a'_1}, \dots, \vec{a'_n}\,,~
\seq{\alpha(f_i,b_i)(s_i)}_{i\in n}\,,~
f)
\label{eq:prom_coend}
\end{eqnarray}
for $s_i \in \alpha(\vec{a_i}, b_i)$, $f \in \Sym(A)[\vec{a'_1} \dots
\vec{a'_n}, \vec{a}]$. 

Here, we make some notational simplifications. 
Firstly, the data of the $\vec{a_i}$s is redundant, provided other
components are typed. Secondly, writing $\vec{a_i} = \seq{a_{i, 1},
\dots, a_{i, p_i}}$ and $\vec{a} = \seq{a'_1, \dots, a'_p}$, then a
morphism $f \in \Sym(A)[\vec{a_1} \dots \vec{a_n}, \vec{a}]$ is some
$\seq{f_i}^{\pi}_{i\in p}$ for $\pi : p \bij \sum_{i=1}^n p_i$; but up
to \eqref{eq:prom_coend} we can -- and we will -- always assume that
the $f_i$s are identities, and only $\pi$ remains. So altogether, a
witness in $\alpha^{\dagger}(\vec{a}, \seq{b_1, \dots, b_n})$ as in
\eqref{eq:witdef} is specified by 
\[
\prom{s_i}_{i\in n}^{\pi} \in \alpha^{\dagger}(\vec{a}, \seq{b_1, \dots,
b_n})
\]
an expression
where $s_i \in \alpha(\vec{a_i}, b_i)$ and $\pi : p \bij \sum_{i=1}^n
p_i$; from now on we fix this notation\footnote{Note the double
brackets, which we adopt to ease the distinction with morphisms of
$\Sym(A)$.} -- note in passing that we may
omit the permutation for the identity, \emph{e.g.} $\prom{s_i}_{i\in n}
= \prom{s_i}_{i\in n}^{\id}$. 

For dereliction, we simply have $\dder_A(\seq{a}, a') = A[a, a']$.

\subsubsection{Additional components.} We carry on with the additional
component of the relative pseudocomonad structure. For every $\alpha
\in \Dist[\Sym A, B]$ and $\beta \in \Dist[\Sym B, C]$, we need 
\[
\mu_{\alpha, \beta} : (\beta \bullet \alpha^{\dagger})^{\dagger} \iso
\beta^\dagger \bullet \alpha^{\dagger}
\]
a natural iso which, given $s_{i,j} \in \alpha(\vec{a_{i,j}}, b_{i,j})$ and
$t_i \in \beta(\vec{b_i}, c_i)$ for $1\leq i \leq n$, $1 \leq j \leq
p_i$, $\vec{a_{i,j}} = \seq{a_{i, j, 1}, \dots, a_{i, j, k_{i,j}}}$,
$\pi_i : \sum_{j \in p_i} k_{ij} \bij \sum_{j \in p_i} k_{ij}$ and $\pi
: \sum_{i,j} k_{ij} \bij \sum_{i,j} k_{ij}$, is set to
\[
(\mu_{\alpha,\beta})_{\vec{a}, \vec{c}}(\prom{t_i \bullet \prom{s_{i,j}}_{j\in p_i}^{\pi_i}}_{i\in
n}^{\pi}) = 
\prom{t_i}_{i\in n} \bullet \prom{s_{i,j}}_{i\in n, j \in p_i}^{\pi \circ
\sum \pi_i}\,.
\]

For cancellation of dereliction we need, for any $\alpha
\in \Dist[\Sym A, B]$, natural isos
\[
\eta_\alpha : \alpha \iso \dder_B \bullet \alpha^{\dagger}\,,
\qquad
\qquad
\theta_A : \dder_A^\dagger \iso \id_{\Sym A}
\]
set by $(\eta_\alpha)_{\vec{a}, b}(s) = \id_b
\bullet \prom{s}$ for $s \in \alpha(\vec{a}, b)$; and
$\theta_A(\prom{f_i}_{i\in n}^{\pi}) =
\seq{f_i}_{i\in n}^{\pi}$ for $(f_i \in A[a_i, a'_i])_{i\in n}$.

\begin{thm}
This specifies a (relative) pseudocomonad $\Sym$ on $\Dist$.
\end{thm}
\noindent 
As usual, we refer to the Kleisli bicategory $\Dist_{\Sym}$ as $\Esp$.

\subsection{Lifting $\coll{-}$ to the Kleisli bicategories}
We show how to lift $\coll{-}$ to a pseudofunctor $\collexp{-} :
\WVis_{\oc} \to \Esp$. 
(We give a direct proof, although 
one could write down a notion of pseudofunctor between
relative pseudocomonads that lifts to the
Kleisli bicategories \cite{street1972formal}.) 

Before we delve into the proof, we introduce some additional notation.

\subsubsection{Additional conventions and notations} If $X$ is a set,
write $\Fam(X)$ for the set of families $(x_i)_{i\in I}$ indexed by $I
\subseteq_f \mathbb{N}$. This also applies to categories: if $A$ is a
category, then $\Fam(A)$ has morphisms from $(x_i)_{\in I}$ to
$(y_j)_{j\in J}$ given by a permutation $\pi : I \bij J$, and a family
$(f_i : x_i \to y_{\pi(i)})_{i\in I}$ of morphisms in $A$ -- so that
$\comp{\oc A} \iso \Fam(\comp{A})$.
If $I$ is a finite subset of natural numbers, let $\card{I}$ be its
cardinal and
$\kappa_I : I \bij \card{I}$ be the unique monotone bijection. 
If $(x_i)_{i\in I}$ is a family, for simplicity we write
$(x_i)_{\card{i} \in \card{I}}$ for the reindexing
$(x_{\kappa_I^{-1}(j)})_{j\in \card{I}}$. Similarly, if $I
\subseteq_f \mathbb{N}$ and $(J_i)_{i\in I}$ is a family with $J_i
\subseteq_f \mathbb{N}$ for all $i\in I$, we write
\[
\kappa_{I, (J_i)} : \bigsqcup_{i\in I} J_i \bij \sum_{i\in I}
\card{J_i}
\]
for the bijection corresponding to arranging (encodings of) pairs
$\seq{i, j}$ in lexicographic order. 

We also introduce a notation for morphisms in $\Fam(\C)$,
in line with the earlier notation introduced for $\Sym(A)$:
we shall sometimes write
$(f_j)_{j\in J}^\pi$
for the element of $\Fam(\C)[(x_i)_{i\in I}, (y_j)_{j\in J}]$ normally
consisting of
\[
(\pi^{-1} : I \simeq J, (f_{\pi^{-1}(i)} \in \C[x_i, y_{\pi^{-1}(i)}])_{i\in
I})\,.
\]

Notice that the family is indexed by its target index set instead of
the source.

\subsubsection{Preservators for exponentials} To lift $\coll{-}$ to
the Kleisli bicategories, we must introduce explicit natural
isomorphisms acting on witnesses for dereliction and promotion. We
make repeated use of the actions of functors on distributors introduced
in Definition~\ref{def:pierres-construction}. 

\begin{lem}\label{lem:pder}
For any strict board $A$, there is a natural isomorphism
\[
\pder_A : \coll{\der_A}[L^{\oc}_A] \iso \dder_{\comp{A}} \in
\Dist[\Sym(\comp{A}), \comp{A}]
\]
\end{lem}
\begin{proof}
From the definition, $\coll{\der_A}[L^{\oc}_A](\seq{x_i}_{i\in n}, y)$ is
non-empty iff $n = 1$, in which case
\begin{eqnarray*}
\coll{\der_A}[L^{\oc}_A](\seq{x}, y) &=& 
\coll{\der_A}((x)_{i\in \{1\}}, y)\\
&=& \{[(\theta^-)^{\{0\} \bij \{1\}}_{\{0\}}, \cc_z \in \confp{\der_A}, \theta^+]
\mid \theta^- : x \sym_A^- z,~\theta^+ : z \sym_A^+ y\}
\end{eqnarray*}
where $\pr_{\der_A}(\cc_z) = (z)_{\{0\}} \vdash z$. This is sent by $\pder_A$
to
\[
\theta^+ \circ \theta^- \in \der_{\comp{A}}[\seq{x}, y]
\]
which is a bijection by Lemma \ref{lem:sym_factor}. Additional
verifications are routine.
\end{proof}

Likewise, there is another natural isomorphism for preservation of promotion:

\begin{lem}\label{lem:pprom}
For any $\sigma \in \WVis[\oc A, B]$, there is a natural isomorphism
\[
\pprom_\sigma : \coll{\sigma^{\oc}} [L^{\oc}_A] \iso [R^{\oc}_B]
(\coll{\sigma} [L^{\oc}_A])^{\dagger}\,,
\]
between distributors in $\Dist[\Sym(\comp{A}), \comp{\oc B}]$.
Furthermore, it is natural in $\sigma$.
\end{lem}
\begin{proof}
Recall that $+$-covered configurations of $\sigma^{\oc}$ correspond to
families $(x^{\sigma,i})_{i\in I}$; and
\[
\pr((x^{\sigma,i})_{i\in I}) 
\quad = \quad
(x^{\sigma,i}_{A, j})_{\pair{i, j}\,\in\,\bigsqcup_{i\in I} J_i} \vdash
(x^{\sigma,i}_{B})_{i\in I}
\]
provided $\pr(x^{\sigma,i}) = (x^{\sigma,i}_{A,j})_{j\in
J_i} \vdash x^{\sigma,i}_B$ for all $i \in I$. 

Now, consider $\seq{x_{A, k}}_{k\in p} \in \Sym(\comp{A})$, and
$(x_{B,l})_{l \in L} \in \comp{\oc B}$. We have:
\begin{eqnarray*}
\coll{\sigma^{\oc}} [L^{\oc}_A](\seq{x_{A, k}}_{k\in p}, (x_{B,l})_{l
\in L}) &=& 
\coll{\sigma^{\oc}}((x_{A, k})_{k\in p},(x_{B,l})_{l \in L})
\end{eqnarray*}
which, by definition, is composed of triples
\[
[(\theta^-_{i,j})^\pi_{\pair{i,j}\,\in\,\bigsqcup_{i\in I}
J_i},
\qquad
(x^{\sigma, i})_{i\in I} \in \confp{\sigma^{\oc}},
\qquad
(\theta^+_l)_{l\in L}^{\varpi}
] 
\quad
\in
\quad 
\coll{\sigma^{\oc}}((x_{A, k})_{k\in p},(x_{B,l})_{l \in L})
\]
where $\pi : \bigsqcup_{i\in I} J_i \simeq p$ is a bijection,
$\theta^-_{i,j} : x_{A, \pi\tuple{i,j}} \sym_A^- x^{\sigma, i}_{A, j}$,
and $(\theta^+_l)_L^\varpi : (x^{\sigma,i}_B)_{i\in I} \sym_{\oc B}^+
(x_{B, l})_{l\in L}$. This can be simplified: the positivity of 
$(\theta^+_l)_L^\varpi$ entails $L = I$ and $\varpi = \id_I$. We are
left with
\[
[(\theta^-_{i,j})^\pi_{\pair{i,j}\,\in\,\bigsqcup_{i\in I}
J_i},
\qquad
(x^{\sigma, i})_{i\in I} \in \confp{\sigma^{\oc}},
\qquad
(\theta^+_i)_{i\in I}
]
\quad
\in
\quad
\coll{\sigma^{\oc}}((x_{A, k})_{k\in p},(x_{B,i})_{i \in I})\,.
\]

We may finally define $\pprom_\sigma$ to send the above witness to:
\[
\prom{[(\theta^-_{i,j})^{\kappa_{J_i}}_{j\in J_i}, x^{\sigma, i},
\theta^+_i]}^{\overline{\pi}}_{\card{i} \in \card{I}}
\in (\coll{\sigma} [L^{\oc}_A])^{\dagger}(\seq{A,k}_{k\in p},
\seq{x_{B,i}}_{\card{i} \in \card{I}})
\]
where $\overline{\pi}$ is the unique bijection such that the
composition $\bigsqcup_{i\in I} J_i \stackrel{\kappa_{I, (J_i)}}{\bij} \Sigma_{i \in
I} \card{J_i} \stackrel{\overline{\pi}}{\bij} p$ is
$\pi$.

It is routine, if lengthy, to verify that this is a bijection, along
with naturality in $\sigma$.
\end{proof}

By analogy with the \emph{unitors} and \emph{associators} forming a
pseudofunctor, we refer to the natural isomorphisms above as the
\textbf{preservators}.

\subsubsection{Coherence of the preservators.}
We show a series of lemmas expressing coherence between the
preservators and the components of the relative pseudocomonads.

First, we have compatibility with cancellation of dereliction,
expressed in two lemmas:

\begin{lem}\label{lem:pres_id_l}
For any strict board $A$, the following diagram commutes
\[
\xymatrix{
\coll{\cc_{\oc A}}[L^{\oc}_A]
        \ar[r]
        \ar[d]&
\id_{\comp{\oc A}}[L^{\oc}_A]
        \ar[d]\\
\coll{\der_A^{\oc}}[L^{\oc}_A]
        \ar[d]&
[R^{\oc}_A]\,\id_{\Sym(\comp{A})}
        \ar[d]\\
[R^{\oc}_A] (\coll{\der_A}[L^{\oc}_A])^{\dagger}
        \ar[r]&
[R^{\oc}_A]\,\dder_{\comp{A}}^{\dagger}
}
\]
with all arrows the obvious structural maps or obtained
by Lemmas \ref{lem:comp_ren_comp},
\ref{lem:pder} and \ref{lem:pprom}.
\end{lem}
\begin{proof}
This diagram is an equation of natural isomorphisms between
distributors
\[
\Sym(\comp{A})^{\op} \times \comp{\oc A} \to \Set\,,
\]
so taking $\seq{z_i}_{\card{i}\in \card{I}} \in \Sym(\comp{A})$
and $(y_{i})_{i\in I} \in \comp{\oc A}$, it boils down to an equality
between functions, checked pointwise. An element of $\coll{\cc_{\oc
A}}[L^{\oc}_A](\seq{z_i}_{\card{i}\in \card{I}}, (y_i)_{i\in
I})$ comprises 
\[ 
\,[
(\theta_i^-)^\pi_{i\in I}\,,\qquad
\cc_{(x_i)_{i\in I}} \in \confp{\cc_{\oc A}}\,,\qquad
(\theta^+_i)_{i\in I}]
\quad
\in
\quad
\coll{\cc_{\oc
A}}[L^{\oc}_A](\seq{z_i}_{\card{i}\in \card{I}}, (y_i)_{i\in I})\,.
\]
where $\pi : I \simeq \card{I}$ is any bijection, $\theta_i^- : z_i
\sym_A^- x_i$ and $\theta_i^+ : x_i \sym_A^+ y_i$.

Using the description of the action of $\pder$ and $\pprom$ on
witnesses outlined earlier, we calculate its image alongside the two
paths around the diagram. In both cases, we get
\[
\seq{\theta^+_i \circ \theta^-_i}^{\overline{\pi}}_{\card{i} \in
\card{I}} 
\in \dder_{\comp{A}}^{\dagger}(\seq{z_i}_{\card{i} \in
\card{I}}, \seq{y_i}_{\card{i}\in \card{I}})
=
([R^{\oc}_A]\,\dder_{\comp{A}}^{\dagger})(\seq{z_i}_{\card{i}\in \card{I}}, (y_i)_{i\in
I})\,,
\]
where $\overline{\pi} : \card{I} \bij \card{I}$ is the only
bijection such that $\overline{\pi} \circ \kappa_I = \pi$. 
\end{proof}

This will show that the two models deal with Kleisli composition with
dereliction on the left in the same way. Similarly, though not quite
symmetrically, the next lemma deals with the Kleisli composition with
dereliction on the right hand side:

\begin{lem}\label{lem:pres_id_r}
Consider $\sigma \in \WVis[\oc A, B]$. Then the diagram commutes
\[
\xymatrix@C=5pt@R=8pt{
\coll{\der_B \odot \sigma^{\oc}}[L^{\oc}_A]
        \ar[r]
        \ar[d]&
\coll{\sigma}[L^{\oc}_A]
        \ar[r]&
\dder_{\comp{B}} \bullet (\coll{\sigma}[L^{\oc}_A])^{\dagger}\\
(\coll{\der_B} \bullet \coll{\sigma^{\oc}})[L^{\oc}_A]
        \ar[d]&&
\coll{\der_B}[L^{\oc}_B] \bullet (\coll{\sigma} [L^{\oc}_A])^{\dagger}
        \ar[u]\\
\coll{\der_B} \bullet (\coll{\sigma^{\oc}}[L^{\oc}_A])
        \ar[rr]
&&
\coll{\der_B} \bullet [R^{\oc}_B] (\coll{\sigma} [L^{\oc}_A])^{\dagger}
        \ar[u]
}
\]
with all arrows the obvious structural maps or obtained
by Lemmas \ref{lem:comp_ren_comp}, \ref{lem:comp_ren_id},
\ref{lem:pder} and \ref{lem:pprom}.
\end{lem}
\begin{proof}
This diagram states an equality between natural isomorphisms
between distributors
\[
\Sym(\comp{A})^{\op} \times \comp{B} \to \Set\,,
\]
which again can be checked pointwise. Taking $\seq{x_{A,i}}_{\card{i}
\in \card{I}} \in \Sym(\comp{A})$ and $x_B \in \comp{B}$, an
element of $(\coll{\der_B \odot
\sigma^{\oc}}[L^{\oc}_A])(\seq{x_{A,i}}_{\card{i}
\in \card{I}}, x_B)$ is formed of three components
\[
\,[
(\theta^-_i)^\pi_{\pair{0, i}\,\in\,\bigsqcup_{\{0\}} I}\,,
\quad
\cc_{x^\sigma_B} \odot (x^\sigma)_{\{0\}}
\in \confp{\der_B \odot \sigma^{\oc}}\,,\quad
\theta_+
]
\]
where $\pi : \bigsqcup_{\{0\}} I \bij \card{I}$, $\theta^-_i :
x_{A, i} \sym_A^- x^\sigma_{A, i}$, and $\theta^+ : x^\sigma_B \sym_B^+
x_B$.

Following the upper path in the diagram, this is sent to:
\begin{eqnarray}
\id_{x_B} \bullet \prom{[(\theta^-_i)^{\pi'}_{i\in I},~x^\sigma,~\theta^+]}
&\in& 
\dder_{\comp{B}} \bullet (\coll{\sigma}[L^{\oc}_A])^{\dagger}(\seq{x_{A,
i}}_{\card{i} \in \card{I}}, x_B)
\label{eq:first_expr}
\end{eqnarray}
where $\pi' : I \bij \card{I}$ is the unique bijection such that
$\bigsqcup_{\{0\}} I \stackrel{\mathsf{snd}}{\bij} I
\stackrel{\pi'}{\bij} \card{I}$ is $\pi$.

On the other hand, following the lower path, we get:
\begin{eqnarray}
\theta^+ \bullet
\prom{[(\theta^-_i)^{\kappa_I}_{i\in I},~x^\sigma,~\id_{x^\sigma_B}]}^{\overline{\pi}} \in 
\der_{\comp{B}} \bullet (\coll{\sigma}[L^{\oc}_A])^{\dagger}(\seq{x_{A,
i}}_{\card{i} \in \card{I}}, x_B)\,,
\label{eq:second_expr}
\end{eqnarray}
for $\overline{\pi} : \card{I} \bij \card{I}$ the only
bijection such that $\bigsqcup_{\{0\}} I \stackrel{\kappa_{\{0\},
(I)}}{\bij} \card{I} \stackrel{\overline{\pi}}{\bij} \card{I}$
is $\pi$. But recall that $\kappa_{\{0\}, (I)} : \sqcup_{\{0\}} I \bij
\card{I}$ is the bijection arranging pairs in lexicographic order,
and $\kappa_I : I \bij \card{I}$ is the unique monotone bijection.
From that and the definition of $\overline{\pi}$ and $\pi'$, it follows
that $\overline{\pi} \circ \kappa_I = \pi'$, which entails that
\eqref{eq:second_expr} is equivalent to \eqref{eq:first_expr}.
\end{proof}

Finally, we have compatibility with the multiplication $\join/\mu$, expressed via:

\begin{lem}\label{lem:pres_join}
If $\sigma \in \WVis[\oc A, B]$ and $\tau \in \WVis[\oc B, C]$, then the
following diagram commutes
\[
\xymatrix@C=-25pt@R=10pt{
&\coll{\tau^{\oc} \odot \sigma^{\oc}}[L^{\oc}_A]
        \ar[dl]
        \ar[dr]\\
\coll{(\tau \odot \sigma^{\oc})^{\oc}}[L^{\oc}_A]
        \ar[d]&&
(\coll{\tau^{\oc}} \bullet \coll{\sigma^{\oc}})[L^{\oc}_A]
        \ar[d]\\
[R^{\oc}_C](\coll{\tau\odot \sigma^{\oc}}[L^{\oc}_A])^{\dagger}
        \ar[d]&&
\coll{\tau^{\oc}}\bullet(\coll{\sigma^{\oc}}[L^{\oc}_A])
        \ar[d]\\
[R^{\oc}_C]((\coll{\tau}\bullet
\coll{\sigma^{\oc}})[L^{\oc}_A])^{\dagger}
        \ar[d]&&
\coll{\tau^{\oc}}\bullet [R^{\oc}_B](\coll{\sigma}[L^{\oc}_A])^{\dagger}
        \ar[d]\\
[R^{\oc}_C](\coll{\tau}\bullet \coll{\sigma^{\oc}}[L^{\oc}_A])^{\dagger}
        \ar[d]
&&
\coll{\tau^{\oc}}[L^{\oc}_B] \bullet
(\coll{\sigma}[L^{\oc}_A])^{\dagger}
        \ar[d]\\
[R^{\oc}_C](\coll{\tau}\bullet
[R^{\oc}_B](\coll{\sigma}[L^{\oc}_A])^{\dagger})^{\dagger}
        \ar[d]&&
[R^{\oc}_C](\coll{\tau}[L^{\oc}_B])^{\dagger} \bullet
(\coll{\sigma}[L^{\oc}_A])^{\dagger}
        \ar[d]\\
[R^{\oc}_C](\coll{\tau}[L^{\oc}_B] \bullet
(\coll{\sigma}[L^{\oc}_A])^{\dagger})^{\dagger}
        \ar[rr]&&
[R^{\oc}_C]((\coll{\tau}[L^{\oc}_B])^{\dagger} \bullet
(\coll{\sigma}[L^{\oc}_A])^{\dagger})
}
\]
with all arrows the obvious structural maps or obtained
by Lemmas \ref{lem:comp_ren_comp}, \ref{lem:comp_ren_id},
\ref{lem:pder} and \ref{lem:pprom}.
\end{lem}
\begin{proof}
This diagram states an equality between natural isomorphisms
between distributors
\[
\Sym(\comp{A})^{\op} \times \comp{\oc C} \to \Set
\]
which again can be checked pointwise. Taking $\seq{x_{A,l}}_{l \in
n} \in \Sym(\comp{A})$ and $(x_{C, i})_{i\in I} \in
\comp{\oc C}$, an element of $\coll{\tau^{\oc} \odot
\sigma^{\oc}}[L^{\oc}_A](\seq{x_{A,l}}_{l \in n}, (x_{C, i})_{i\in I})$
comprises
\[
\,[(\theta^-_{i,j,k})_{\seq{\seq{i,j},k}}^\pi\,,
\quad
(x^{\tau,i})_{i\in I} \odot (x^{\sigma, i, j})_{\seq{i,j}\,\in
\,\bigsqcup_{i\in I} J_i}\,,
\quad
(\theta^+_i)_{i\in I}
]
\]
where $\pr_\tau\,x^{\tau,i} = (x^{\tau,i}_{B,j})_{j\in J_i} \vdash
x^{\tau,i}_C$, $\pr_\sigma\,x^{\sigma,i,j} = (x^{\sigma, i, j}_{A,
k})_{k\in K_{i,j}}$, so that 
\[
\pr_{\sigma^{\oc}}((x^{\sigma, i, j})_{\seq{i,j}}) =
(x^{\sigma,i,j}_{A,k})_{\seq{\seq{i,j}, k}} \vdash
(x^{\sigma,i,j}_B)_{\seq{i,j}}\,,
\qquad
\pr_{\tau^{\oc}}((x^{\tau,i})_{i\in I}) =
(x^{\tau,i}_{B,j})_{\seq{i,j}} \vdash (x^{\tau, i}_C)_{i\in I}
\]
with $x^{\sigma,i,j}_B = x^{\tau, i}_{B, j}$, $\pi :
\bigsqcup_{\seq{i,j}} K_{i,j} \bij n$, $\theta^-_{i,j,k} : x_{A, \pi
\seq{\seq{i,j},k}} \sym_A^- x^{\sigma, i, j}_{A, k}$ and $\theta^+_i :
x^{\tau, i}_C \sym_C^+ x_{C, i}$.

For the left hand side path, a careful calculation shows that this is sent to
\[
\prom{[(x^{\tau, i}_{B,j})_{j\in J_i}^{\kappa_{J_i}}, x^{\tau, i},
\theta^+_i]}_{\card{i} \in \card{I}} \bullet
\prom{[(\theta^-_{ijk})_{k\in K_{ij}}^{\kappa_{K_{ij}}}, x^{\sigma,i,j},
x^{\sigma,i,j}_B]}_{\card{i} \in \card{I}, \card{j} \in
\card{J_i}}^{\pi'}
\]
where $\pi' : n \bij n$ is the only bijection such that $\pi' \circ
\kappa_{I,(J),(K)} = \pi$, where $\kappa_{I,(J), (K)} :
\bigsqcup_{\seq{i,j}} K_{i,j} \bij \sum_{i\in I} \sum_{j\in J_i}
\card{K_{ij}} = n$ arranges the triples $(i,j,k)$ in lexicographic
order.

Likewise, for the right hand side path, it is sent to
\[
\prom{[(x^{\tau, i}_{B,j})_{j\in J_i}^{\kappa_{J_i}}, x^{\tau,i},
\theta^+_i]}_{\card{i} \in \card{I}}^{\omega} \bullet
\prom{[(\theta^-_{ijk})_{k\in K_{ij}}^{\kappa_{K_{ij}}}, x^{\sigma,i,j},
x^{\sigma,i,j}_B]}^{\omega'}_{\card{\seq{i,j}} \in
\card{\sqcup_{i\in I} J_i}}
\]
where $\omega : \sum_{i\in I} \card{J_i} \bij \sum_{i\in I}
\card{J_i}$ is obtained as $\kappa_{\sqcup_i J_i} \circ \kappa_{I,
(J)}^{-1}$ -- that is, it links the lexicographic ordering of pairs
$(i,j)$, and their ordering following their representations
$\seq{i,j} \in \mathbb{N}$. Likewise, $\omega' : n \bij n$ is the
unique bijection such that $\omega' \circ \kappa_{\sqcup_i J_i, (K)} =
\pi$ -- that is, it is $\pi$ where $i,j,k$ are ordered in the
lexicographic order of the pairs $\seq{i,j}, k$. 

It is then routine to
show that these two witnesses are equivalent, as required.
\end{proof}

\subsubsection{Preservation of Kleisli composition} Finally, we are
equipped to show how $\coll{-}$ lifts to the Kleisli bicategories.
Recall that, for $\sigma \in \WVis[\oc A, B]$, we have
\[
\collexp{\sigma} = \coll{\sigma}[R^\oc_A] \in \Dist[\Sym(\comp{A}),
\comp{B}]\,.
\]

Thus, we may set:
\[
\pid^{\oc}_A = \pder_A : \collexp{\der_A} \iso \dder_{\comp{A}}
\]
for the natural isomorphism witnessing preservation of identity.
Likewise, for $\sigma \in \WVis[\oc A, B]$ and $\tau \in \WVis[\oc B,
C]$,  we set the natural iso witnessing preservation of composition as
\[
\begin{array}{rclcl}
\pcomp^{\oc}_{\sigma, \tau} &:& \collexp{\tau \odot \sigma^{\oc}} &=&
\coll{\tau \odot \sigma^{\oc}}[L^{\oc}_A]\\
&&&\stackrel{\pcomp_{\sigma, \tau}}{\iso}&
(\coll{\tau} \bullet \coll{\sigma^{\oc}})[L^{\oc}_A]\\
&&&\iso& \coll{\tau} \bullet \coll{\sigma^{\oc}}[L^{\oc}_A]\\
&&&\stackrel{\pprom_\sigma}{\iso} &
\coll{\tau} \bullet [R^{\oc}_B](\coll{\sigma}[L^{\oc}_A])^{\dagger}\\
&&&\iso& \coll{\tau}[L^{\oc}_B] \bullet
(\coll{\sigma}[L^{\oc}_A])^{\dagger}\\
&&&=& \collexp{\tau} \bullet \collexp{\sigma}^{\dagger}
\end{array}
\]

These definitions provide the necessary components for:

\begin{thm}
This provides the data for a pseudofunctor
\[
\collexp{-} : \WVis_{\oc} \to \Esp\,.
\]
\end{thm}
\begin{proof}
Naturality of $\pcomp^{\oc}_{\sigma, \tau}$ in $\sigma$ and $\tau$ is
direct by composition of natural isos. The coherence
diagrams follow by diagram chasing, relying on Lemmas
\ref{lem:pres_id_l}, \ref{lem:pres_id_r} and \ref{lem:pres_join}.
\end{proof}

\subsection{Cartesian closed structure}
\label{subsec:cc_vis}

We show the bicategory $\WVis_\oc$ is cartesian closed, using typical
constructions in concurrent games, in sufficient detail to keep the
paper self-contained.  
For a precise definition of the structure of cartesian closed bicategories we refer to \cite{philip:thesis}.

\subsubsection{Cartesian products in $\WVis_{\oc}$} 
The empty board $\top$, with $\kappa_\top(\emptyset) = 1$, is strict,
and it is direct that $\top$ is a terminal object in $\WVis_{\oc}$,
since any negative strategy $A \vdash \top$ must be empty.  
Now if $A, B$ are strict, the \textbf{product board}
$A \with B$ is defined as $A \tensor B$, except that all events of $A$
are in conflict with events of $B$. This means that
configurations of $A \with B$ are either empty, or of the form $\{1\} \times x$
for $x \in \conf{A}$ (written $\inj_1(x)$) or $\{2\} \times x$ for $x
\in \conf{B}$ (written $\inj_2(x)$). For the payoff, we set
$\kappa_{A\with B}(\emptyset) = 1$ and $\kappa_{A_1 \with
A_2}(\inj_i(x)) = \kappa_{A_i}(x)$, making $A \with B$ a strict arena.
By strictness, we have an isomorphism
\[
L^{\with}_{A, B} : \comp{A\with B} \iso \comp{A} + \comp{B} : R^{\with}_{A, B}
\]
which reflects the definition of binary products in $\Esp$.
The \textbf{first projection} 
$\pi_A \in \WVis[\oc (A \with B), A]$
has ess $\cc_A$, with display map the map of ess characterized by
\[
\pr(\cc_x) = (\inj_1(x))_{\{0\}} \vdash x\,,
\]
with the \textbf{second projection} defined symmetrically. 
If $\sigma \in \WVis[\oc \Gamma, A]$ and $\tau \in \WVis[\oc \Gamma, B]$,
their \textbf{pairing} has ess $\sigma \with \tau$ and display map the
unique such that
\[
\pr(\inj_1(x^\sigma)) = x^\sigma_{\oc \Gamma} \vdash \inj_1(x^\sigma_A)
\in \conf{\oc \Gamma \vdash A \with B}\,,
\]
and likewise for $\pr(\inj_2(x^\tau))$, yielding $\pair{\sigma, \tau}
\in \WVis[\oc \Gamma, A\with B]$. This extends to a functor
$\pair{-, -} : \WVis_{\oc}[\Gamma, A] \times \WVis_{\oc}[\Gamma, B] \to
\WVis_{\oc}[\Gamma, A \with B]$ in a straightforward way.

\begin{prop}
For any strict arenas $\Gamma, A$ and $B$, there is an adjoint
equivalence
\[
\begin{tikzcd}[row sep=0pt, column sep=10pt]
  \WVis_{\oc}[\Gamma,A \with B]
  \ar[rr,bend left=15,"{(\pi_A \odot (-)^\oc, \pi_B \odot (-)^\oc )}"]
  &
\quad  \simeq
  &
  \WVis_{\oc}[\Gamma,A]
  \times
  \WVis_{\oc}[\Gamma,B]
  \ar[ll,bend left=15,"\tuple{-, -}"]
\end{tikzcd}
\]
providing the data to turn $A \with B$ into a cartesian product in the
bicategorical sense.  
\end{prop}
\begin{proof}
To witness the equivalence, we need three natural isomorphisms 
\[
\eta_\sigma : \sigma \iso \tuple{\pi_1 \odot \sigma^{\oc}, \pi_2 \odot
\sigma^{\oc}}\,,
\qquad
\epsilon^i_{\sigma_1, \sigma_2} :
\pi_i \odot \pair{\sigma_1, \sigma_2}^{\oc} \iso \sigma_i
\]
for $\sigma \in \WVis_{\oc}[\Gamma, A_1\with A_2]$, $\sigma_i \in
\WVis_{\oc}[\Gamma, A_i]$ and $i \in \{1, 2\}$.

We focus first on the former. Notice that for all $x^{\sigma} \in
\confp{\sigma}$, if $x^\sigma$ is non-empty then as $\sigma$ is
negative and the $A_i$ are strict, the minimal event of $x^\sigma$ must
occur in $A_1$ or $A_2$. We set
\[
\begin{array}{rcrclcl}
\eta_\sigma &:& \sigma &\iso& \tuple{\pi_1 \odot \sigma^{\oc}, \pi_2
\odot \sigma^{\oc}}\\
&&\emptyset &\mapsto& \emptyset\\
&& x^\sigma &\mapsto& \inj_i(\cc_{x^\sigma_{A_i}} \odot
(x^\sigma)_{\{0\}})&\quad&
\text{(if $x^\sigma \neq \emptyset$, $\pr(x^\sigma) = x^\sigma_\Gamma
\vdash \inj_i(x^\sigma_{A_i})$)}\,,
\end{array}
\]
yielding a positive isomorphism, additionally natural in $\sigma$. For
the co-unit, we set
\[
\begin{array}{rcrclcl}
\epsilon^i_{\sigma_1, \sigma_2} &:& \pi_i \odot \tuple{\sigma_1,
\sigma_2}^{\oc} &\iso& \sigma_i\\
&& \cc_{x^{\sigma_i}_{A_i}} \odot (\inj_i(x^{\sigma_i}))_{\{0\}}
&\mapsto&
x^{\sigma_i}
\end{array}
\]
again a positive isomorphism natural in $\sigma_1$ and $\sigma_2$. For
the adjunction, we check that
\[
\xymatrix{
\pi_i \odot \sigma^{\oc}
	\ar[r]^{\hspace{-40pt}\pi_i \odot \eta_\sigma^{\oc}}
	\ar[dr]_{\id}&
\pi_i \odot \tuple{\pi_1 \odot \sigma^{\oc}, \pi_2 \odot \sigma^{\oc}}
	\ar[d]^{\epsilon^i_{\pi_1 \odot \sigma^{\oc}, \pi_2 \odot
\sigma^{\oc}}}\\
&\pi_i \odot \sigma^{\oc}
}
\qquad
\xymatrix{
\tuple{\sigma_1, \sigma_2}
	\ar[r]^{\hspace{-50pt}\eta_{\tuple{\sigma_1, \sigma_2}}}
	\ar[dr]_{\id}&
\tuple{\pi_1 \odot \tuple{\sigma_1, \sigma_2}^{\oc}, \pi_2 \odot
\tuple{\sigma_1, \sigma_2}^{\oc}}
	\ar[d]^{\tuple{\epsilon^1_{\sigma_1, \sigma_2},
\epsilon^2_{\sigma_1, \sigma_2}}}\\
&\tuple{\sigma_1, \sigma_2}
}
\]
commute for all $i=1, 2$, $\sigma \in \WVis_{\oc}[\Gamma, A_1\with
A_2]$, $\sigma_i \in \WVis_{\oc}[\Gamma, A_i]$, which is direct.
\end{proof}

\subsubsection{Closed structure in $\WVis_{\oc}$} 
Recall that the objects of $\WVis_\oc$ are strict boards, and observe
that any strict board $B$ is isomorphic to $\with_{i\in I} B_i$ where
the $B_i$ are \textbf{pointed}, meaning that they have exactly one
minimal event. We first define a \textbf{linear arrow} for a $-$-board
$A$ and a pointed strict board $B$, by setting $A \lin B$ to be $A
\vdash B$ with a stricter dependency order, so that all events in $A$
causally depend on the unique minimal move in $B$. 

This is generalized for any strict $B$ as 
\[
A \lin \with_{i\in I} B_i = \with_{i\in I} (A \lin B_i)\,
\]
whose configurations have a convenient description:

\begin{lem}\label{lem:conf_iso_lin}
For any $-$-board $A$, and strict board $B$, we have
\[
(- \lin -) : \nconf{A} \times \nconf{B} \iso \nconf{A \lin B}\,.
\]
\end{lem}
\begin{proof}
Straightforward, using crucially the fact that since $B$ is strict, a
configuration of null payoff in $B$ must be non-empty, so that the $A$
component is always reachable.
\end{proof}

Now for $A, B$ strict, we define the \textbf{arrow} $A \tto B$ as 
$\oc A \lin B$. This is equipped with 
\[
\evm_{A, B} \in \WVis[\oc ((\oc A \lin B) \with A), B]
\]
an \textbf{evaluation} strategy
consisting of the ess $\cc_{\oc A \lin B}$, and where the display map
is
\[
\pr(\cc_{(x^i_A)_{i\in I} \lin x_B}) = (\inj_1((x^i_A)_{i\in I} \lin
x_B))_{\{\tuple{0,0}\}} \uplus
(\inj_2(x^i_A))_{\tuple{1, i} \in \Sigma_{\{1\}} I} \vdash x_B
\]
if $x_B \neq \emptyset$, and empty otherwise,
with $\uplus$ the union of families with disjoint index sets.
Likewise, the \textbf{currying} of $\sigma \in \WVis[\oc (\Gamma \with
A), B]$ is a strategy $\Lambda(\sigma)$ with ess $\sigma$ and display
\[
\pr_{\Lambda(\sigma)}(x^\sigma) = (x_\Gamma^i)_{i\in I} \vdash (x_A^j)_{j\in J} \lin x_B
\]
for $x^\sigma \neq \emptyset$, where
$\pr_\sigma(x^\sigma) = (\inj_1(x_\Gamma^i))_{i\in I} \uplus
(\inj_2(x_A^j))_{j\in J} \vdash x_B$. 

Altogether, this gives a functor $\Lambda : \WVis_{\oc}[\Gamma \with A,
B] \to \WVis_{\oc}[\Gamma, A \tto B]$.

\begin{prop}\label{prop:arrow_wis}
For any strict $\Gamma, A$ and $B$, there is an adjoint equivalence
\[
\begin{tikzcd}[row sep=0pt, column sep=10pt]
  \WVis_{\oc}[\Gamma, A\tto B]
  \ar[rr,bend left=20,"{\evm_{A, B} \odot \pair{- \odot
\pi_\Gamma^{\oc}, \pi_A}^{\oc}}"]
  &
  \simeq
  &
  \WVis_{\oc}[\Gamma \with A, B]
  \ar[ll,bend left=20,"\Lambda(-)"]
\end{tikzcd}
\]
providing the data to turn $A \Rightarrow B$ into an exponential object in the bicategorical sense.
\end{prop}
\begin{proof}
To witness the equivalence, we need two natural isomorphisms
\[
\eta_\sigma : \sigma \iso \Lambda(\evm_{A, B} \odot \tuple{\sigma \odot
\pi_\Gamma^{\oc}, \pi_A}^{\oc})\,,
\qquad
\beta_\tau : \evm_{A, B} \odot \tuple{\Lambda(\tau) \odot
\pi_\Gamma^{\oc}, \pi_A}^{\oc} \iso \tau
\]
for $\sigma \in \WVis_{\oc}[\Gamma, A\tto B]$ and $\tau \in
\WVis_{\oc}[\Gamma \with A, B]$. For the former, we set
\[
\eta_\sigma(x^\sigma) = \cc_{(x_{A, i})_i \lin x_B} \odot
(
(\inj_1(x^\sigma \odot (\cc_{x_{\Gamma, j}})_{j \in J}))_{\pair{0,0}}
\uplus (\inj_2(\cc_{x_{A, i}}))_{\pair{1,i}}
)
\]
where $\pr_\sigma(x^\sigma) = (x_{\Gamma, j})_{j\in J} \vdash (x_{A,
i})_{i\in I} \lin x_B$. A careful check shows this is
well-defined and
\[
\pr_{\Lambda(\evm_{A, B} \odot \tuple{\sigma \odot
\pi_\Gamma^{\oc}, \pi_A}^{\oc})}(\eta_\sigma(x^\sigma)) = 
(x_{\Gamma, j})_{\pair{\pair{0,0},\pair{j,0}}} 
\vdash (x_{A,i})_{\pair{\pair{1,i},0}} \lin x_B
\]
which is clearly positively symmetric to $\pr_\sigma(x^\sigma)$ as
required. Likewise, the expression
\[
y = \cc_{(x_{A, i})_i \lin x_B} \odot (
(\inj_1(x^\tau \odot (\cc_{x_{\Gamma, j}})_{j \in J}))_{\pair{0,0}}
\uplus (\inj_2(\cc_{x_{A, i}}))_{\pair{1,i}}
\]
captures exactly all $+$-covered configurations of $\evm_{A, B} \odot
\tuple{\Lambda(\tau) \odot \pi_\Gamma^{\oc}, \pi_A}^{\oc}$, displaying
to 
\[
(\inj_1(x_{\Gamma, j}))_{\pair{\pair{0,0},\pair{j,0}}}
\uplus (\inj_2(x_{A,i}))_{\pair{\pair{1,i},0}} \vdash x_B
\]
for $\pr_\tau(x^\tau) = (\inj_1(x_{\Gamma,j}))_{j\in J} \uplus
(\inj_2(x_{A, i}))_{i\in I} \vdash x_B$; we set $\beta_\tau(y) =
x^\tau$. Unfolding these definitions, a careful calculation confirms
that the triangle identities hold.
\end{proof}

\subsubsection{Cartesian closed structure in $\Esp$.} Here, we briefly
review the cartesian closed structure of $\Esp$. Firstly, we have
$\top$ the empty groupoid, with an adjoint equivalence
\[
\Esp[\Gamma, \top] \simeq 1
\]
for $1$ the category with one object (and one morphism). For binary
products, given $A$ and $B$
groupoids, we simply set their \textbf{with} as the disjoint sum $A
\with B = A + B$; writing as in games $\inj_1(a) \in A + B$ for $a \in
A$ and $\inj_2(b) \in A + B$ for $b \in B$. 
For projections, we set distributors
\[
\begin{array}{rcrcl}
\pi_1 &:& \Sym(A + B)^{\op} \times A &\to& \Set\\
&& (\tuple{\inj_1(a)}, a') &\mapsto& A[a, a']
\end{array}
\qquad
\begin{array}{rcrcl}
\pi_2 &:& \Sym(A + B)^{\op} \times B &\to& \Set\\
&& (\tuple{\inj_2(b)}, b') &\mapsto& B[b, b']
\end{array}
\]
and for the pairing of $\alpha \in \Dist[\Sym(\Gamma), A]$ and $\beta
\in \Dist[\Sym(\Gamma), B]$, we set
\[
\begin{array}{rcrcl}
\pair{\alpha, \beta} &:& \Sym(\Gamma)^{\op} \times (A + B) &\to&
\Set\\
&& (\vec{\gamma}, \inj_1(a)) &\mapsto& \alpha(\vec{\gamma}, a)\\
&& (\vec{\gamma}, \inj_2(b)) &\mapsto& \beta(\vec{\gamma}, b)
\end{array} 
\]
with obvious functorial action. We skip the details of the adjoint
equivalence 
\[
\Esp[\Gamma, A \with B] \simeq \Esp[\Gamma, A] \times \Esp[\Gamma,
B]\,,
\]
which are straightforward and not required for the remainder of this
paper.

Given groupoids $A$ and $B$, the \textbf{arrow} is $A \tto B =
\Sym(A)^{\op} \times B$. For \textbf{evaluation}, 
\[
\begin{array}{rcrcl}
\evm_{A, B}' &:& (\Sym(A \tto B) \times \Sym(A))^{\op}\times B &\to&
\Set\\
&& ((\seq{(\vec{a}, b)}, \vec{a'}), b')&\mapsto&
\Sym(A)[\vec{a'}, \vec{a}] \times B[b, b']
\end{array}
\]
and empty otherwise, provides the core of the evaluation mechanism; but
we need a distributor in $\Esp[\Sym(A)^{\op} \times B + \Sym(A), B]$.
For that, we set $\evm_{A, B} = \evm'_{A, B}[L_{A\tto B, A}^{\seely}]$,
\emph{i.e.}
\[
\evm_{A, B}(\vec{s}, a) = \evm'_{A, B}(L_{A\tto B,
A}^{\seely}(\vec{s}), b)
\]
where $L_{A, B}^{\seely} : \Sym(A+B) \simeq \Sym(A) \times \Sym(B) :
R_{A, B}^{\seely}$ is the Seely equivalence.

Finally, the \textbf{currying} of $\alpha \in \Esp[\Gamma + A, B]$ is
defined by
\[
\begin{array}{rcrcl}
\Lambda(\alpha) &:& \Sym(\Gamma)^{\op} \times (\Sym(A)^{\op} \times B)
&\to& \Set\\
&& (\vec{\gamma}, (\vec{a}, b)) & \mapsto& \alpha(R^{\seely}_{\Gamma,
A}(\vec{\gamma}, \vec{a}), b)
\end{array}
\]
informing an adjoint equivalence $\Esp[\Gamma, A\tto B] \simeq
\Esp[\Gamma + A, B]$ as in Proposition \ref{prop:arrow_wis}.

\subsection{A cartesian closed pseudofunctor}
\label{subsec:cc_pseudofuncto}
 We first show that the pseudofunctor $\collexp{-} : \WVis_\oc \to \Esp$
preserves the cartesian structure. The key observation is that:

\begin{lem}\label{lem:pair_iso_L}
Consider $A, B$ two mixed boards. Then, the distributor
\[
\pair{\collexp{\pi_A}, \collexp{\pi_B}} \in \Esp[\comp{A \with B},
\comp{A} + \comp{B}]
\]
is naturally isomorphic to $\widehat{L^{\with}_{A, B}}$.
\end{lem}

This is a straightforward variation on the proof that $\coll{-}$
preserves the identity. From this, we obtain that the pseudofunctor
$\collexp{-}$ preserves products: 

\begin{prop}
The pseudofunctor $\collexp{-}$ is a \emph{fp-pseudofunctor} in the
sense of \cite{philip-marcelo:tt}.
\end{prop}
\begin{proof}
The terminal object is preserved in a strict sense, since $\comp{\top}$
is empty. For preservation of binary products, we must provide the
missing components for an adjoint equivalence 
\[
\xymatrix{
\comp{A\with B}
	\ar@/^1pc/[rr]^{\pair{\collexp{\pi_A},
\collexp{\pi_B}}}&\simeq&
\comp{A} + \comp{B}
	\ar@/^1pc/[ll]^{\q^{\with}_{A, B}}
}
\]
in $\Esp$. But the iso $L^{\with}_{A, B} : \comp{A \with B} \iso
\comp{A} + \comp{B} : R^{\with}_{A, B}$ lifts to an adjoint equivalence
\[
\widehat{L^{\with}_{A, B}} : \comp{A \with B} \simeq \comp{A} +
\comp{B} : \widehat{R^{\with}_{A, B}}
\]
in $\Esp$; so providing $\q^{\with}_{A, B} = \widehat{R^{\with}_{A,
B}}$, we conclude via Lemma \ref{lem:pair_iso_L}.
\end{proof}

Finally, it remains to prove that the cartesian closed structure is
preserved as well. Observe that if $A$ and $B$ are $-$-boards with $B$
strict, then we have an adjoint equivalence
\[
L^{\tto}_{A, B} : \comp{A\tto B} \simeq \Sym(\comp{A})^{\op} \times \comp{B}
: R^{\tto}_{A, B}
\]
using first Lemma \ref{lem:conf_iso_lin} as since $B$ is strict, its
complete configurations are non-empty; and observing that this
decomposition also holds for symmetries; followed by Proposition
\ref{prop:eq_oc_sym}.

Now, as for binary products, our key observation is the following:

\begin{lem}\label{lem:keytto}
Consider $A$ and $B$ two $-$-boards, with $B$ strict. Then, the
distributor
\[
\Lambda(\collexp{\evm_{A, B}} \bullet_{\oc} \q^\with_{A\tto B, A})
\in
\Esp[\comp{A \tto B}, \Sym(\comp{A})^{\op} \times \comp{B}] 
\]
is naturally isomorphic to $\widehat{L^{\tto}_{A,B}}$.
\end{lem}
\begin{proof}
Unfolding the definitions on both sides, this boils down to a natural
isomorphism
\[
\begin{array}{l}
\coll{\evm_{A, B}}((\inj_1((x_{A, i})_{i\in I} \lin x_B))_{0} \uplus 
(\inj_2(y_{A, j}))_{j\in n}, y_B)\\
\qquad\qquad \iso
\Sym(\comp{A})[\seq{y_{A, j}}_{j \in n}, \seq{x_{A, i}}_{\card{i} \in
\card{I}}] \times \comp{B}[x_B, y_B]
\end{array}
\]
which again is a variation on the preservation of the identity by
$\coll{-}$.
\end{proof}

We may now use this to conclude:

\begin{prop}
The pseudofunctor $\collexp{-}$ is a \emph{cc-pseudofunctor} in the
sense of \cite{philip-marcelo:tt}.
\end{prop}
\begin{proof}
We must provide the missing components for an adjoint equivalence 
\[
\xymatrix{
\comp{A\tto B}
	\ar@/^1pc/[rr]^{\Lambda(\collexp{\evm_{A, B}} \bullet_{\oc}
\q^\with_{A\tto B, A})}&\simeq&
\comp{A} + \comp{B}
	\ar@/^1pc/[ll]^{\q^{\tto}_{A, B}}
}
\]
in $\Esp$. But the equivalence $L^{\tto}_{A, B} : \comp{A \tto B}
\simeq \Sym(\comp{A})^{\op} \times \comp{B} : R^{\tto}_{A, B}$ lifts to 
\[
\widehat{L^{\tto}_{A, B}} : \comp{A \tto B} \simeq 
\Sym(\comp{A})^{\op} \times \comp{B} : \widehat{R^{\tto}_{A, B}}
\]
in $\Esp$; so providing $\q^{\tto}_{A, B} = \widehat{R^{\tto}_{A, B}}$,
we conclude via Lemma \ref{lem:keytto}. 
\end{proof}

Altogether this completes the proof of our main theorem:

\begin{thm}\label{th:main}
We have a cartesian closed pseudofunctor
$\collexp{-} : \WVis_{\oc} \to \Esp$.
\end{thm}

\subsubsection*{Note.} We have constructed a cartesian closed
pseudofunctor directly, without leveraging the preservation of
symmetric monoidal structure established in earlier sections.
Although the cartesian closed structure is related to the monoidal
structure via Seely isomorphisms, the theory of bicategorical models
of linear logic, and a fortiori functors between models, is still
under-developed, and a direct approach seemed more achievable. This
is an important aspect of this work which deserves further investigation. 

\section{Some consequences for the $\lambda$-calculus}\label{sec:lambda}

Finally, for the last technical section of this paper, we illustrate
this pseudofunctor by relating a dynamic and a static model of the pure
(untyped) $\lambda$-calculus. 
\subsection{Two models of the pure $\lambda$-calculus}
\subsubsection{A reflexive object in $\WVis_{\oc}$}
Our bicategory of games contains a \textbf{universal arena} $\U$ with
an isomorphism of arenas
\[
\unf_\U :\ \  \U\ \  \iso \ \  \U \tto \U \ \ :  \fld_\U
\]
making $\U$ an extensional reflexive object
\cite{DBLP:books/daglib/0067558}. Concretely, $\U$ is constructed
corecursively as the solution of the following `domain-like equation':
$\U = \tensor_{\mathbb{N}} \oc \U \lin o$, where $o$ is the arena with
one negative move $*$, $\kappa_o(\emptyset) = 1$ and $\kappa_o(\{*\}) = 0$. 

Following the interpretation of the $\lambda$-calculus in a
reflexive object, we have, for every closed $\lambda$-term $M$, a
strategy $\intr{M} : \U$. This strategy has a clear interpretation:
it is a representation of the \emph{Nakajima tree} of $M$ (see
e.g.~\cite{DBLP:journals/tcs/KerNO02,DBLP:conf/csl/ClairambaultP18}).

\subsubsection{The Pure $\lambda$-Calculus in $\Esp$} Likewise,
one can construct models of the $\lambda$-calculus in $\Esp$ 
\cite{fiore2008cartesian, ol:intdist, ol:proof}.  A \emph{categorified graph model} \cite{ol:proof} consists of a small category $\D $ together with an embedding
\[          \iota :  \Sym(\D)^{\op} \times \D \hookrightarrow \D\,.                   \]

We adopt a type-theoretic understanding of these models, setting $
\tyl \multimap \ty :=  \iota (\tyl, \ty) .$ The main intuition is that
elements of $ \D$ are \emph{types}, where the arrow type is built exploiting \emph{sequences of
types}. These sequences correspond to a notion of $k$-ary
\emph{intersection types} \cite{ol:intdist}, by setting
$ ( \ty_1 \cap \cdots \cap \ty_k) := \seqdots{\ty}{1}{k}$.
The \emph{free} categorified graph model can be presented syntactically
as the category generated by the graph in Figure \ref{fig:multigraph}.
Composition and identities are defined by induction in the natural way,
exploiting the definition of composition  of morphisms of sequences
induced by $\Sym{(-)} $. We shall now \emph{complete} $\D$ in an
appropriate way, producing an extensional model that we will then
relate to the universal arena $\U$. 

\begin{figure}
\begin{gather*}
\text{Types:} \\[1em]
       \D \ni \ty ::= \ast \mid \seqdots{\ty}{1}{k} \multimap \ty                                   
\end{gather*}
\begin{gather*}
\text{Proof-relevant subtyping:}\\[1em]
{  \begin{prooftree}  
\infer0{  id_\ast : \ast \to \ast     }
\end{prooftree} }
\qquad
{\begin{prooftree}
\hypo{ \vec{f}^\sigma : \vec{b} \to \tyl                           } \hypo{     f : \ty \to b     }
\infer2{     \vec{f}^\sigma \multimap f : (\tyl \multimap \ty) \to (\vec{b} \multimap b )          }
\end{prooftree}}  \\[1em]
{\begin{prooftree}
\hypo{    \sigma \in S_k                            } \hypo{  f_1 : \ty_{\sigma 1} \to b_1     } \hypo{\cdots}  \hypo{   f_k : \ty_{\sigma k} \to b_k    }
\infer4{  \seqdots{f}{1}{k}^\sigma : \seqdots{\ty}{1}{k} \to \seqdots{b}{1}{k}         }
\end{prooftree} }
\end{gather*} \hrulefill
\caption{Graph of intersection types $ \mathsf{G} .  $} \label{fig:multigraph}
\end{figure}

\subsection{Constructing an extensional model $ \D^\ast $}

 The extensional model $ \D^\ast $ is defined as an appropriate
completion of $D$, where we add chosen isomorphisms between types in
a coherent way. This construction can be defined in abstract
categorical terms as an appropriate pseudocolimit (a
\emph{coisoinserter}), but in the present paper we will just give its
explicit syntactic presentation. The following construction is a
particular case of \emph{completion of partial algebras} for
categorified graph models, in the sense of \cite{ol:proof}. 

We start by adding the following two arrows to the graph of intersection types $ \mathsf{G} :$ 
\[       \begin{prooftree}
\infer0{  \mathsf{e} : \ast \to \seq{} \multimap \ast  }
\end{prooftree}
\qquad
 \begin{prooftree}
\infer0{  \mathsf{e}^{-1} : \seq{} \multimap \ast \to \ast }
\end{prooftree}
           \]
resulting in a graph $\mathsf{G}^\ast$. We now want to obtain a category
from $\mathsf{G}^\ast$. Identities are the same as $\D$. Composition is inherited from the
composition of $\D$, by adding the following cases:
\[    \mathsf{e} \circ \mathsf{id}_{\ast} = \mathsf{e} \qquad \mathsf{id}_{  \seq{} \multimap \ast } \circ \mathsf{e}^{-1} = \mathsf{e}^{-1}                             \qquad \mathsf{e} \circ \mathsf{e}^{-1} = \mathsf{id}_{\seq{}\multimap \ast } \qquad \mathsf{e}^{-1} \circ \mathsf{e} = \mathsf{id}_\ast             . \]

The category $\D^\ast$ is a categorified graph model:
\[  
\iota^\ast : \Sym(\D^\ast)^{op}\times \D^\ast \hookrightarrow \D^\ast \qquad (\tyl, \ty) \mapsto \tyl \multimap \ty                                            \]

\begin{thmC}[\cite{ol:proof}] \label{th:dastff}
The functor $ \iota^\ast$ is fully faithful and essentially surjective on objects.
\end{thmC}

\paragraph{Morphism actions and congruence on type derivations} We now explain how the interpretation of $\lambda $-terms in $\Esp $ can be presented exploiting intersection types. 

The intuition is that the semantics of a term will compute all its
terminating executions in face of  a given environment. These
executions can be encoded by the means of an intersection type system.
Given $a \in \D$,
$\intr{M}_\Esp^{\D}(a)$ will appear (up to a canonical
isomorphism) as the set of derivations
of the judgement $\vdash M : a$ in the type system. The construction of
this system reflects the corresponding operations on
distributors: in particular derivations carry explicit symmetries,
$\intr{M}_\Esp^{\D}(a)$ is really the set of derivations \emph{quotiented} by an
equivalence relation letting symmetries flow though the concrete
derivation. 

Let $ \D$ be a categorified graph model. A \emph{type declaration} over $ \D$ consists of a pair $ x : \tyl $ where $x $ is a variable of the $ \lambda$-calculus and $ \tyl \in \Sym(\D) . $
A  \emph{type context} consists of a sequence of type declarations $
x_1 : \tyl_1 , \dots, x_n : \tyl_n . $ Type contexts are denoted with
greek letters $ \gamma, \delta \dots . $ Given two type contexts $
\gamma = x_1 : \tyl_1 , \dots, x_n : \tyl_n  $ and $ \delta = x_1 :
\vec{b}_1, \dots, x_n : \vec{b}_n $ we define their \emph{pointwise
concatenation} as $ \gamma \otimes \delta = x_1 : \tyl_1 ::
\vec{b}_1,\dots, x_n : \tyl_n :: \vec{b}_n $ where $  \tyl_i ::
\vec{b}_i $ denotes list concatenation. A \emph{morphism} of type
contexts $ \eta : \gamma \to \delta $ consist of a family of morphisms
$ \vec{f}_i^{\sigma_i} : \tyl_i \to \vec{b}_i$ for $  i \in [n] .$
Morphisms of type contexts are denoted with greek letters $ \eta,
\theta \dots$ The definition of the type system is given in Figure
\ref{fig:IntD}.

Given a type derivation $ \pi $ of conclusion $\gamma \vdash M : \ty $
and  type morphisms $ \theta : \delta\to \gamma $ and $f : \ty \to b $
we define contravariant and covariant actions on derivations:
\[     \begin{prooftree}
\hypo{  \pi \{ \theta \}       }\ellipsis{}
{    \delta \vdash M : \ty  }
\end{prooftree}            \qquad  \qquad       \begin{prooftree}
\hypo{  [f]\pi        }\ellipsis{}
{    \gamma \vdash M : b  }
\end{prooftree}                             \]

 The explicit definitions are given, respectively, in Figures
\ref{fig:la} and \ref{fig:ra}. By a straightforward induction, one can
prove that both actions are compositional and unital in the appropriate
way, that is $  (\pi \{ \theta \}) \{ \theta' \} = \pi \{\theta \circ
\theta' \}, \pi \{ id\} = \pi $ and $ [g]([f]\pi) = [g \circ f] \pi,
[id] \pi = \pi$.
 
 By the means of these morphism actions, we define an \emph{equivalence
relation} on type derivations, as the smallest equivalence relation
generated by the rules of Figure \ref{fig:congr} and that is compatible
with the structure of type derivations.
This equivalence corresponds to the one induced by an appropriate \emph{coend formula}:
namely, the coend formula that derives from the categorified graph
model interpretations of application and $ \lambda$-abstraction in $\Esp$.

\begin{figure*}[t] 
\begin{gather*}
{\begin{prooftree} 
\hypo{f : \ty' \to \ty }
	\infer1{ x_1 :  \seq{}, \dots,  x_i  : \seq{\ty'}, \dots , x_n :  \seq{} \vdash \tyf{x_i} : \ty }
\end{prooftree}} \\[1em]
{\begin{prooftree}
\hypo{    \gamma, x : \tyl \vdash \tyf{M} : \ty             }\hypo{  f : (\tyl \multimap \ty) \to b       }
\infer2{   \gamma \vdash \tyf{ \lambda x . M } : b                  }
\end{prooftree} }
\\[1em]
{\begin{prooftree}
\hypo{\gamma_0 \vdash \tyf{M} : \seqdots{\ty}{1}{k}  \multimap \ty  } \hypo{(\gamma_i \vdash \tyf{N} : \ty_i)_{i =1}^k} \hypo{ \eta : \delta \to \bigotimes_{i=1}^k \gamma_i}
\infer3{        \delta \vdash \tyf{MN} : \ty                     } 
\end{prooftree} }
\end{gather*} \hrulefill\caption{Intersection type system over a graph model $ \D .  $}
\label{fig:IntD}
\end{figure*}

\begin{figure*}[t]\scalebox{0.85}{\parbox{1.05\linewidth}{
\begin{align*} &
\left(\raisebox{-8.5pt}{{\begin{prooftree}
\hypo{ f : \ty' \to \ty }
	\infer1[]{ \seq{}, \dots, \seq{\ty'}, \dots, \seq{} \vdash \ty }
\end{prooftree}}} \right)\{ g : b \to \ty' \} \quad &=& \quad \raisebox{-8.5pt}{{\begin{prooftree}
\hypo{ f \circ g }
	\infer1[]{ \seq{}, \dots, \seq{b}, \dots, \seq{} \vdash \ty }
\end{prooftree} }}
\\[0.5cm]
 &\left(\raisebox{-22pt}{{\begin{prooftree}
\hypo{\pi}\ellipsis{}{ \delta, \tyl \vdash \ty } \hypo{ f : (\tyl \multimap \ty) \to b }
	\infer2[]{ \delta \vdash  b }
\end{prooftree}}}\right)\{ \eta \} \quad &=& \quad{ \raisebox{-12pt}{\begin{prooftree} 
\hypo{\pi\{ \eta, id_{\tyl} \}}\ellipsis{}{ \delta', \tyl \vdash \ty 
 } \hypo{ f : (\tyl \multimap \ty) \to b }
	\infer2[]{ \delta' \vdash \tyl \multimap \ty }
\end{prooftree}}}
\\[0.5cm]
& \left(\raisebox{-22pt}{{\begin{prooftree}
\hypo{\pi_1}\ellipsis{}{\gamma_{0} \vdash \tyl \multimap \ty }\hypo{\pi_i }
\ellipsis{}{\gamma_{i} \vdash \ty_i}
\delims{\left(}{\right)_{i =1}^{k}}
\hypo{ \theta : \delta \to \bigotimes_{j=0}^{k} \gamma_j}
	\infer3[]{ \delta \vdash \ty}\end{prooftree}}}\right)\{ \eta \}
		 &=& \quad
\raisebox{-22pt}{{	\begin{prooftree}
\hypo{\pi_1}\ellipsis{}{\gamma_{0} \vdash \tyl \multimap \ty }\hypo{\pi_i }
\ellipsis{}{\gamma_{i} \vdash \ty_i}
\delims{\left(}{\right)_{i =1}^{k}}\hypo{ \theta \circ \eta }
	\infer3[]{ \delta' \vdash \ty}
\end{prooftree} }}
\end{align*}
where $ \tyl = \seqdots{\ty}{1}{k} $ and $ \eta : \delta' \to \delta . $}}
\caption{Right action on derivations.} \hrulefill
\label{fig:la}
\end{figure*}

\begin{figure*}[t] 	\scalebox{0.85}{\parbox{1.05\linewidth}{
\begin{align*}
&[g : a \to b]\left(
\raisebox{-8.5pt}{$\begin{prooftree}
\hypo{ f : \ty' \to \ty }
	\infer1[]{ \seq{}, \dots, \seq{\ty'}, \dots, \seq{} \vdash \ty }
\end{prooftree}$} \right) 
&=& \quad
\raisebox{-8.5pt}{$
\begin{prooftree}
\hypo{g \circ f : \ty' \to b}
	\infer1[]{ \seq{}, \dots, \seq{\ty'}, \dots, \seq{} \vdash b}
\end{prooftree} $}
\\[3ex]
&{}[g : a' \to b]\left(\raisebox{-22pt}{$\begin{prooftree}
\hypo{\pi}\ellipsis{}{ \delta, \tyl \vdash \ty } \hypo{ f : (\tyl \multimap \ty) \to \ty' }
	\infer2[]{ \delta \vdash  \ty }
\end{prooftree}$}\right) \quad 
&=& \quad \raisebox{-12pt}{\begin{prooftree} 
\hypo{\pi}\ellipsis{}{ \delta, \tyl \vdash \ty 
 } \hypo{ g \circ f : (\tyl \multimap \ty) \to b }
	\infer2[]{ \delta \vdash b }
\end{prooftree}} 
\\[6ex]
&{}[g : a \to b]\left(\raisebox{-22pt}{\begin{prooftree}
\hypo{\pi_0}\ellipsis{}{\gamma_{0} \vdash \tyl \multimap \ty }\hypo{\pi_i }
\ellipsis{}{\gamma_{i} \vdash \ty_i}
\delims{\left(}{\right)_{i =1}^{k}}
\hypo{ \morpCone : \delta \to \bigotimes_{j=0}^{k} \Gamma_j }
	\infer3[]{ \delta \vdash \ty}\end{prooftree}}\right)
 &=&  \
	\raisebox{-22pt}{\begin{prooftree}
\hypo{[ id_{\tyl} \multimap g] \pi_0}\ellipsis{}{ \gamma_{0} \vdash \tyl \multimap b }\hypo{\pi_i }
\ellipsis{}{\gamma_{i} \vdash \ty_i}
\delims{\left(}{\right)_{i =1}^{k}}
\hypo{ \morpCone }
	\infer3[]{ \Delta \vdash b }
\end{prooftree}}
\end{align*}
where $ \tyl = \seqdots{\ty}{1}{k} . $}}
\caption{Left action on derivations.}  \hrulefill
\label{fig:ra}
\end{figure*}

\begin{figure*}
\scalebox{0.84}{\parbox{1.05\linewidth}{
\[
	\begin{array}{rcl}
	{ \begin{prooftree}
\hypo{\pi_0}\ellipsis{}{ \gamma_{0} \vdash \vec{b} \multimap \ty }\hypo{[f_{i}]\pi_{\sigma(i)}}
\ellipsis{}{\gamma_{\sigma(i)} \vdash b_i}
\delims{\left(}{\right)_{i =1}^{k}}
\hypo{ ( id_{\gamma_0}\otimes \sigma^{\star}) \circ \morpCone }
	\infer3[]{ \delta \vdash \ty }
\end{prooftree}}
	&\raisebox{8pt}{$\sim$}&
{	\begin{prooftree}
\hypo{[ \vec{f}\thinspace^{\sigma} \multimap id_{\ty}] \pi_0 }\ellipsis{}{ \gamma_{0} \vdash \tyl \multimap \ty} \hypo{\pi_i }\ellipsis{}{\gamma_{i} \vdash \ty_i}\delims{\left(}{\right)_{i =1}^{k}}\hypo{ \morpCone }
	\infer3[]{ \delta \vdash \ty }
	\end{prooftree}}\\
	\\
{	\begin{prooftree}
	\hypo{\pi_0 \{ \theta_0 \} } \ellipsis{}{ \gamma_{0} \vdash \tyl \multimap \ty}\hypo{\pi_i \{\theta_i \} }\ellipsis{}{\gamma_{i} \vdash \ty_i}\delims{\left(}{\right)_{i =1}^{k}}\hypo{ \morpCone : \delta \to \bigotimes_{j =0}^{k} \gamma_j }\infer3[]{ \delta \vdash \ty }
	\end{prooftree} }
	&\raisebox{8pt}{$\sim$} &
{	\begin{prooftree}
	\hypo{\pi_0}\ellipsis{}{ \gamma'_{0} \vdash \tyl \multimap \ty}\hypo{\pi_i }\ellipsis{}{\gamma'_{i} \vdash \ty_i}\delims{\left(}{\right)_{i =1}^{k}}
	\hypo{ (\bigotimes_{j=0}^{k} \theta_j) \circ \morpCone}\infer3[]{ \delta \vdash \ty}
	\end{prooftree}}\\
	\\
{	\begin{prooftree}
\hypo{ [g] \pi \{ id_{\delta}, \vec{g}\thinspace^{\sigma} \} }\ellipsis{}{ \delta, \tyl' \vdash \ty' } \hypo{ f : (\tyl \multimap \ty) \to b }
	\infer2[]{ \delta \vdash  b }
\end{prooftree}}
	&\raisebox{8pt}{$\sim$} &
{\begin{prooftree}
\hypo{\pi}\ellipsis{}{ \delta, \tyl \vdash \ty } \hypo{  f \circ (\vec{g}\thinspace^{\sigma} \multimap g  ): (\tyl' \multimap \ty') \to b }
	\infer2[]{ \delta \vdash  b }
\end{prooftree}}
	\end{array}
\]

where $ \seq{f_1, \dots, f_{k}}^{\sigma} : \tyl = \seqdots{\ty}{1}{k} \to \vec{b} = \seqdots{b}{1}{k}$, $\vec{g}\thinspace^{\sigma} : \tyl' \to \tyl$, $g : \ty \to \ty'$ and $ \theta_i : \gamma_i \to \gamma'_{i}.$  $\sigma^\star  $ denotes the canonical morphism $  \sigma^\star  : \bigotimes_{i}^{k} \Gamma_i \to \bigotimes_{i = 1}^k \Gamma_{\sigma (i)} $ for a permutation $ \sigma$. }}
\caption{Congruence on derivations.} \hrulefill
\label{fig:congr}
\end{figure*}

\paragraph{Intersection type distributors} We are now ready to introduce the syntactic presentation of the categorified graph model semantics, namely \emph{intersection type distributors}.

Given a $ \lambda$-term $ M$ and a categorified graph model $ \D $, with $ \vec{x} \supseteq \mathsf{fv}(M) $, we define 
\[     \mathsf{ITD}^{\D}(M)_{\vec{x}} : (\Sym(\D) \overbrace{\otimes
\dots \otimes}^{ \mid\vec{x}\mid \text{ times} } \Sym(\D))^{op} \times \D \to \Set \]
the \emph{intersection type distributor} of $ M $,
by the following family of sets: 
\[        
\mathsf{ITD}^{\D}(M)_{x_1,\dots,x_n}(\tyl_1,\dots,\tyl_n; \ty) = 
\left\{ 
\raisebox{-10pt}{
\begin{prooftree}
\hypo{\pi}\ellipsis{}{  x_1 : \tyl_1 , \dots, x_n : \tyl_n \vdash M : \ty  } 
\end{prooftree}
}
\right\} / \sim\,.
\]

Stating the connection between the intersection type distributor of a
term and its species of structure semantics requires the
\emph{Seely equivalence} for contexts:
\[
L_{\vec{x}}^{\seely} : \Sym(\D \overbrace{+ \ldots +}^{\text{$|\vec{x}|$
times}} \D) \simeq \Sym(\D) \overbrace{\times \ldots
\times}^{\text{$|\vec{x}|$ times}} \Sym(\D) : R_{\vec{x}}^{\seely}\,.
\]

\begin{thmC}[\cite{ol:proof}]\label{thm:ol}
Let $M$ be a $\lambda$-term, $\vec{x} \supseteq \mathsf{fv}(M)$  and $\D$ a categorified graph model. 

Then, we have an isomorphism, natural in $\delta$ and $\ty$:
\[         \mathsf{ITD}(M)_{\vec{x}}(\delta, \ty) 
\quad\cong\quad                
\intr{M}_\Esp^{\D} (  R_{\vec{x}}^{\seely}(\delta)     , \ty)\,. \]
\end{thmC}

\subsection{Relating the Interpretations} We first relate the reflexive
objects.

\begin{figure}
\begin{gather*}
\text{Objects:} \\[1em]
\comp{ \U }  \ni v ::=  (\bar{v}_1, \dots, \bar{v}_k)  \multimap \ast \qquad ( \bar{v}_1 \ \text{non-empty.} ) 
\\ \bar{v} ::= (v_i)_{i \in I \subseteq_f \mathbb{N}}
\\[2em]
\text{  Morphisms: } \\[1em]
{\begin{prooftree}
\hypo{   \theta_i : \bar{v}_i \cong \bar{v}'_i             }\infer1{ (\theta_1, \dots, \theta_n)  \multimap \ast : (\bar{v}_1,\dots, \bar{v}_n) \multimap \ast \to  (\bar{v}'_1,\dots, \bar{v}'_n) \multimap \ast     }
\end{prooftree}} \\[1em]
{\begin{prooftree}
\hypo{      \pi : I \cong I'         }\hypo{\theta_i : v_i \to v_{\pi i}}
\infer2{  (\theta_i)^{\pi}_{i \in I} : (v_i)_{i \in I} \to (v'_j)_{j \in I'}                       }
\end{prooftree}}
\end{gather*}
\caption{Groupoid $ \comp{\U} $ associated to the reflexive object $ \U$.} \label{fig:grouppres} \hrulefill
\end{figure}

\begin{thm}\label{th:equiv}
We have an adjoint equivalence of groupoids:
$L^\U : \comp{\U} \simeq \D^\ast : R^\U$.
\end{thm}
\begin{proof}
First, $\comp{\U}$ is \emph{isomorphic} to that
presented via the grammar in Figure \ref{fig:grouppres}. Indeed, 
\[
\comp{\U} = \comp{\otimes^\N \oc \U \lin o}
\iso \comp{\otimes^\N \oc \U} \times \comp{o}
\]
where $\comp{o} = \{\emptyset\}$. Likewise, (the obvious infinitary
generalization of) \eqref{eq:comptensor} informs a bijection between
$\comp{\otimes^\N \oc \U}$ and the sequences $\prod_{n \in \omega}
\comp{\oc \U}$ which are empty almost everywhere -- those may be
represented uniquely as sequence $(x_1, \dots, x_n)$ where $x_i \in
\comp{\oc \U}$ and $x_1$ is non-empty, removing the
heading infinite stream of empty configurations. Finally, $\comp{\oc
\U}$ is in bijection with families $(x_i)_{i\in I}$ for $I \subseteq_f
\N$, as outlined in the proof of Proposition \ref{prop:eq_oc_sym}. From
these observations, the claimed bijection is direct and extends to symmetries.

Hence, we treat Figure \ref{fig:grouppres} as a syntactic presentation
of $\comp{\U}$. Armed with this, we may now define the components of
the desired equivalence simply by induction, with:
\[
\begin{array}{rclcl}
L^{\U}((\bar{v}_1, \dots, \bar{v}_k)  \multimap \ast) &=& 
L^\U(\bar{v}_1) \lin \dots \lin L^\U(\bar{v}_k) \lin \ast\\
L^\U((v_i)_{i\in I}) &=& \seq{L^\U(v_{i_1}), \ldots, L^\U(v_{i_n})}
&\quad&\text{\hspace{-16pt}(where $I = \{i_1 < \ldots < i_n\}$)}\\\\
R^{\U}(\seq{} \lin a) &=& R^\U(a)\\
R^{\U}(\vec{a}_1 \lin \ldots \lin \vec{a}_n \lin \ast) &=&
(R^\U(\vec{a}_1), \ldots, R^\U(\vec{a}_n)) \lin \ast
&&\text{\hspace{-16pt}($\vec{a}_1$ non-empty)}\\
R^\U(\seq{a_1, \ldots, a_n}) &=& (R^\U(a_i))_{1\leq i \leq n}
\end{array}
\]
which easily extends to an equivalence as claimed.
\end{proof}

Additionally, this equivalence is compatible with unfoldings in the
sense that  
\begin{figure}
  \[
\begin{tikzcd}[column sep=0pt, row sep=15pt]
  \comp{\U}
    \arrow[rr, "L^\U"]
    \arrow[d, "\comp{\unf_\U}"]
  & & \D^\ast
    \arrow[d, "\unf_{\D^\ast}"] \\
  \comp{\U \tto \U}
    \arrow[dr, "L^\tto_{\U, \U}"', pos=0.4]
  & & \Sym(\D^\ast)^{\op} \times \D^\ast \\
  & \Sym(\comp{\U})^{\op} \times \comp{\U}
    \arrow[ur, "\Sym(L^\U)^{\op}\times L^\U"', pos=0.6]
\end{tikzcd}
\]
\caption{Compatibility with unfoldings}
\label{fig:comp_unf}
\end{figure}
the diagram of Figure \ref{fig:comp_unf} commutes, and compatible with
folding in the same way. 

Using our Theorem \ref{th:main}, it follows that:

\begin{thm}\label{thu:main}
For any closed term $M$, we have a natural iso
\[
\intr{M}^{\D^\ast}_{\Esp} \iso \collexp{\intr{M}^\U_{\WVis_{\oc}}} \circ R^U
\]
\end{thm}

This shows that, for $a \in \D^\ast$, the set
$\intr{M}^{\D^\ast}_{\Esp}(a)$ described in Theorem \ref{thm:ol} as a set of derivations up
to congruence may be equivalently described, up to canonical
isomorphism, as 
\[
\{
(x, \theta^+) \quad \mid \quad
x \in \confp{\intr{M}^\U_{\WVis_{\oc}}}, \quad \theta^+ : \pr(x) \sym_\U^+
R^\U(a)\}\,.
\]

In other words, the interpretation of
pure $\lambda$-terms as species computes the set of $+$-covered
configurations \emph{equipped with} a positive
symmetry. Interestingly, this set is 
constructed \emph{without quotient}, and thus
provides canonical representatives for the equivalence classes of
derivations. We now include a detailed example
illustrating this by examining the interpretation of a simple
$\lambda$-term in both models.

\begin{exa}
Take the $ \lambda$-term $M =   xx.  $ We shall compute its species
interpretation on appropriate points and then compare it to its
concurrent game semantics. 

The species interpretation of $xx $ consists of a distributor 
\[         \intr{xx}^{\D^\ast}_{\Esp}(-;-) : \Sym(\D^\ast)^{op} \times \D^\ast  \to \Set                   \]

We evaluate the functor on $   \intr{xx}^{\D^\ast}_{\Esp}( \seq{
\seq{\ast, \ast} \multimap \ast, \ast , \ast } ; \ast )$. To compute an
element of this set, we shall exploit the type theoretic presentation
of the semantics. Following Figure \ref{fig:IntD}, an element of $
\intr{xx}^{\D^\ast}_{\Esp}( \seq{  \seq{\ast, \ast} \multimap \ast,
\ast , \ast } ; \ast ) $ can be defined by the following derivation
scheme:  

\[\begin{prooftree}
\hypo{  f : c_0 \to    (\seq{b_1,b_2} \multimap \ast)         }
\infer1{   x : \seq{c_0} \vdash x : \seq{b_1,b_2} \multimap \ast        } \hypo{g_1 : c_1 \to b_1}\infer1{ x : \seq{c_1} \vdash x : b_1  } \hypo{g_2 : c_2 \to b_2} \infer1{  x : \seq{c_2} \vdash x : b_2}\hypo{  \eta}
\infer4{  x : \seq{\seq{\ast,\ast} \multimap \ast,  \ast, \ast}   \vdash xx : \ast           } 
\end{prooftree} \]
written $\pi^{f,g_1,g_2}_\eta$,
where $ \eta =  (f_0, f_1, f_2 ): (\seq{\ast,\ast} \multimap \ast,
\ast, \ast) \to (c_0, c_1, c_2)  $. as a consequence of Theorem
\ref{th:dastff}, we have $ c_0 = \seq{\ty_1,\ty_2} \multimap \ty $
and $ f = \seq{f_1,f_2}^\tau \multimap g $ with $ f_1 : b_{\tau i} \to
\ty_i  , g : \ty \to \ast .$ 

Then, by the congruence on derivations
induced by the coend formula, we have that 
\[ \pi^{f,g_1,g_2}_\eta  \sim \]
\[               \begin{prooftree} 
\hypo{      id_{\seq{b_1,b_2} \multimap \ast}                   }
\infer1{  x : \seq{ \seq{b_1,b_2} \multimap \ast } \vdash x : \seq{b_1,b_2} \multimap \ast  } \hypo{ id_{b_1}  }\infer1{  x : \seq{b_1} \vdash x : b_1 } \hypo{ id_{b_2}}\infer1{   x : \seq{b_2} \vdash x : b_2    } \hypo{   (f,g_1,g_2) \circ \eta  }
\infer4{    x : \seq{\seq{\ast,\ast} \multimap \ast,  \ast, \ast}   \vdash xx : \ast       }
\end{prooftree}                        \]

Hence, we can restrict to the derivation scheme:
\begin{equation}
       \begin{prooftree} 
\hypo{      id_{\seq{b_1,b_2} \multimap \ast}                   }
\infer1{  x : \seq{ \seq{b_1,b_2} \multimap \ast } \vdash x : \seq{b_1,b_2} \multimap \ast  } \hypo{ id_{b_1}  }\infer1{  x : \seq{b_1} \vdash x : b_1 } \hypo{ id_{b_2}}\infer1{   x : \seq{b_2} \vdash x : b_2    } \hypo{   \theta  }
\infer4{    x : \seq{\seq{\ast,\ast} \multimap \ast,  \ast, \ast}   \vdash xx : \ast       }
\end{prooftree}     \label{eq:sc1}        
\end{equation}
where $\theta = \seq{h_0, h_1, h_2}^\sigma $, where the permutation $
\sigma$ is either the identity or the swap between the two occurrences
of $\ast $. Actually, one can prove that there are just \emph{two}
different equivalence classes of derivations in  $
\intr{xx}^{\D^\ast}_{\Esp}( \seq{  \seq{\ast, \ast} \multimap \ast,
\ast , \ast } ; \ast ) $. These correspond precisely to the choice of
either the identity or the swap permutation in the following restricted
schema: 
\begin{equation}
      \begin{prooftree} 
\hypo{      id_{\seq{\ast,\ast} \multimap \ast}                   }
\infer1{  x : \seq{ \seq{\ast,\ast} \multimap \ast } \vdash x : \seq{\ast,\ast} \multimap \ast  } \hypo{ id_{\ast}  }\infer1{  x : \seq{\ast} \vdash x : \ast } \hypo{ id_{\ast}}\infer1{   x : \seq{\ast} \vdash x : \ast    } \hypo{   \theta  }
\infer4{    x : \seq{\seq{\ast,\ast} \multimap \ast,  \ast, \ast}   \vdash xx : \ast       }
\end{prooftree}        \label{eq:sc2}    
\end{equation} 
where now $ \theta = \seq{h_0, h_1, h_2}^\sigma$ with $ h_i $ being the
appropriate identity morphism. This happens because for any isomorphism
$ b_1 \cong \ast, b_2 \cong \ast$, we can prove that $ \ref{eq:sc1} $ and $ \ref{eq:sc2}$ are equivalent by the rules of congruence
(Figure \ref{fig:congr}).

Now, let us consider the game semantics side of things. As $x\,x$ has
just one free variable, $\intr{x\,x}_{\WVis_\oc}$ is a strategy on $\oc
\U \vdash \U$. For our example, the useful unfolding of that is $\oc
(\tensor^\N \oc \U \lin o) \vdash \U$ -- and following the equivalence
of Theorem \ref{th:equiv}, the point $(\seq{\seq{\ast, \ast} \multimap
\ast, \ast , \ast }; \ast)$ corresponds to the configuration in the
diagram below:
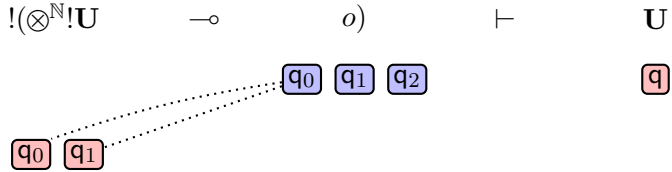
\[
  \begin{tikzpicture}
  \node (t1) at (0, 0.8) {$\oc (\tensor^\N \oc \U$};
  \node (t2) at (2, 0.8) {$\lin$};
  \node (t3) at (4, 0.8) {$o)$};
  \node (t4) at (6, 0.8) {$\vdash$};
  \node (t5) at (8, 0.8) {$\U$};

  \node[negnode] (qm) at (8,0) {$\q$};
  \node[posnode] (q0) at (3.3,0) {$\q_0$};
  \node[posnode] (q1) at (4,0) {$\q_1$};
  \node[posnode] (q2) at (4.7,0) {$\q_2$};

  \node[negnode] (a1) at (-0.3,-1) {$\q_0$};
  \node[negnode] (a2) at (0.4,-1) {$\q_1$};

  \draw[game-causality] (q0) to[bend right=5] (a1.north east);
  \draw[game-causality] (q0) -- (a2);

  \end{tikzpicture}
\]
where the numbers in index are the copy indices, corresponding to the
distinct copies generated by the exponential modality -- $x$ is called
three times, and one of these calls its argument twice; the copy
indices simply follow the rank in the sequence representation in
$\D^\ast$.
Now, the configurations of $\intr{x\,x}_{\WVis_\oc}$ whose display is
symmetric to the above turn out to be those of the form
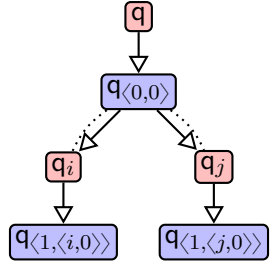
\[
\begin{tikzpicture}
  \node[negnode] (qm) at (2,4) {$\q$};
  \node[posnode] (q00) at (2,3) {$\q_{\pair{0,0}}$};
  \node[negnode] (qi) at (1,2) {$\q_{i}$};
  \node[negnode] (qj) at (3,2) {$\q_{j}$};
  \node[posnode] (q1i) at (1,1) {$\q_{\pair{1, \pair{i,0}}}$};
  \node[posnode] (q1j) at (3,1) {$\q_{\pair{1, \pair{j,0}}}$};

  \draw[strat-causality] (qm) -- (q00);
  \draw[strat-causality] (q00) -- (qi);
  \draw[strat-causality] (q00) -- (qj);
  \draw[strat-causality] (qi) -- (q1i);
  \draw[strat-causality] (qj) -- (q1j);

  \draw[game-causality, bend right=15] (q00) to (qi);
  \draw[game-causality, bend left=15] (q00) to (qj);
\end{tikzpicture}
\]
for arbitrary $i, j \in \N$. If we insist that the display is
\emph{positively} symmetric to the configuration above, then this fixes
$i=0$ and $j=1$, so that there is only \emph{one} such configuration of
$\intr{x\,x}_{\WVis_\oc}$. Any finally, there are exactly two positive
symmetries
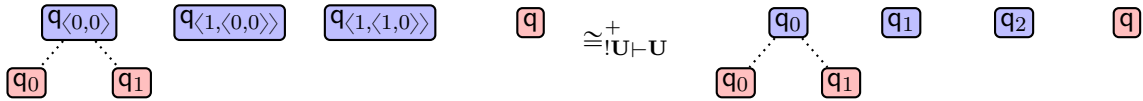
\[
  \raisebox{-20pt}{
    \begin{tikzpicture}
  \node[posnode] (q00) at (4,1) {$\q_{\pair{0,0}}$};
  \node[posnode] (q100) at (6,1) {$\q_{\pair{1, \pair{0, 0}}}$};
  \node[posnode] (q110) at (8,1) {$\q_{\pair{1, \pair{1, 0}}}$};
  \node[negnode] (q) at (10,1) {$\q$};
  \node[negnode] (q0) at (3.3,0.2) {$\q_0$};
  \node[negnode] (q1) at (4.7,0.2) {$\q_1$};

  \draw[game-causality] (q00) -- (q0);
  \draw[game-causality] (q00) -- (q1);
\end{tikzpicture}}
\quad
\sym^+_{\oc \U \vdash \U}
\quad
\raisebox{-20pt}{
\begin{tikzpicture}
  \node[posnode] (q0) at (4,1) {$\q_{0}$};
  \node[posnode] (q1) at (5.5,1) {$\q_{1}$};
  \node[posnode] (q2) at (7,1) {$\q_{2}$};
  \node[negnode] (qm) at (8.5,1) {$\q$};
  \node[negnode] (qn0) at (3.3,0.2) {$\q_0$};
  \node[negnode] (qn1) at (4.7,0.2) {$\q_1$};

  \draw[game-causality] (q0) -- (qn0);
  \draw[game-causality] (q0) -- (qn1);
\end{tikzpicture}}
\]
as the only degree of freedom is whether or not to swap the two
standalone positive $\q$s. Hence, we recover the same number of
witnesses as in generalized species of structure; though as the reader
can observe, the reasoning and the objects of study are very different.
\end{exa}

Finally, in what follows, we show how our cartesian closed
pseudofunctor allows us to transfer results from game semantics to
generalized species. It is known that the game semantics of the
$\lambda$-calculus captures the maximal sensible $\lambda$-theory
$\mathcal{H}^*$, because strategies coincide with the corresponding
normal forms, Nakajima trees -- this was first shown in traditional
game semantics by Ker \emph{et al.} \cite{DBLP:journals/tcs/KerNO02},
and adapted to concurrent games (additionally with probabilistic
choice)  -- in \cite{DBLP:conf/csl/ClairambaultP18}.  Using Theorem
\ref{thu:main}, we can deduce a result on the $\lambda $-theory induced
by $\D^\ast .$ Given a model $\D $ of $ \lambda $-calculus in an
arbitrary bicategory, the \emph{theory} induced by $ \D $ is the $
\lambda $-theory induced by the relation on $\lambda $-terms:
\[     \{ (M,N) \mid M,N \in \Lambda \text{ s.t. } \intr{M}^D \cong \intr{N}^D . \} .        \]

  The correspondence established in this paper
allows us to derive the following new result:
\begin{cor}
The theory of $ \D^\ast $ is $\mathcal{H}^*$.
\end{cor}
\begin{proof}
Consider two $\lambda$-terms $M$ and $N$. 
If $M \equiv_{\mathcal{H}^*} N$, then $\intr{M}^\U_{\WVis_{\oc}} \iso
\intr{N}^\U_{\WVis_{\oc}}$ and by Theorem \ref{thu:main}, 
$\intr{M}^{\D^\ast}_{\Esp} \iso \intr{N}^{\D^\ast}_{\Esp}$.

Reciprocally, isomorphisms in $\Esp$ form a sensible $\lambda$-theory
(see \cite{ol:proof}, Corollary 6.14). As $H^*$ is the greatest
sensible $\lambda$-theory, it must include isomorphisms in $\Esp$.
\end{proof}

This is a simple application, but recent work suggests that there is much to explore in
the bicategorical semantics of the $\lambda$-calculus \cite{ol:proof}.

\section{Conclusion}
\label{sec:conclusion}

In this paper, we have mapped out the links between thin concurrent games
and generalized species of structures, two bicategorical models of
linear logic and programming languages. By giving a proof-relevant
and bicategorical extension of  the
relationship between dynamic and static models, we have established the new state of
the art in this line of work.

This bridges previously disconnected semantic realms. In the past,
such bridges have proved fruitful for transporting results between
dynamic and static semantics
(\cite{DBLP:conf/lics/CastellanCPW18,DBLP:journals/pacmpl/ClairambaultV20,DBLP:conf/csl/ClairambaultP18}).
This opens up many perspectives: bicategorical models are a
very active field, and several recent developments may be re-examined in 
light of this connection
(\cite{ol:intdist,DBLP:conf/lics/TsukadaAO18,DBLP:journals/pacmpl/ClairambaultV20}).

Moreover, this work exposes fundamental phenomena
regarding symmetries in quantitative semantics. Symmetries lie at the heart of both thin
concurrent games and generalized species, but they are treated
completely differently: in $\Esp$, witnesses referring to multiple
copies of a resource are closed under the action of
all symmetries (``saturated''), whereas $\TCG$ relies on a mechanism
for choreographing a choice of copy indices, providing an address for
individual resources (``thin''). This distinction is also at the
forefront of a recent line of work on \emph{thin spans of groupoids}
\cite{DBLP:conf/lics/ClairambaultF23,DBLP:conf/lics/ClairambaultF24},
which shares with game semantics this handling of symmetry, without the
the combinatorial details of event structures.

\section*{Acknowledgment}
This work was supported by the ANR project DyVerSe
(ANR-19-CE48-0010-01); by the Labex MiLyon (ANR-10-LABX-0070) of
Universit\'e de Lyon, within the program ``Investissements d'Avenir''
(ANR-11-IDEX-0007), operated by the French National Research Agency
(ANR); by the US Air Force Office for Scientific Research under award
number FA9550-21-1-0007; by a Royal Society University Research
Fellowship; by a Paris Region Fellowship co-funded by the
European Union (Marie Skłodowska-Curie grant agreement 945298); and by
the PEPR integrated project EPiQ
ANR-22-PETQ-0007 part of Plan France 2030

\bibliographystyle{alphaurl}
\bibliography{main}

\begin{thebibliography}{BDER97b}

\bibitem[AHM98]{DBLP:conf/lics/AbramskyHM98}
Samson Abramsky, Kohei Honda, and Guy McCusker.
\newblock A fully abstract game semantics for general references.
\newblock In {\em {LICS}}, pages 334--344. {IEEE} Computer Society, 1998.

\bibitem[AJM00]{DBLP:journals/iandc/AbramskyJM00}
Samson Abramsky, Radha Jagadeesan, and Pasquale Malacaria.
\newblock Full abstraction for {PCF}.
\newblock {\em Inf. Comput.}, 163(2):409--470, 2000.

\bibitem[Bar85]{DBLP:books/daglib/0067558}
Hendrik~Pieter Barendregt.
\newblock {\em The lambda calculus - its syntax and semantics}, volume 103 of
  {\em Studies in logic and the foundations of mathematics}.
\newblock North-Holland, 1985.

\bibitem[BDER97a]{DBLP:conf/lics/BaillotDE97}
Patrick Baillot, Vincent Danos, Thomas Ehrhard, and Laurent Regnier.
\newblock Believe it or not, {AJM}'s games model is a model of classical linear
  logic.
\newblock In {\em {LICS}}, pages 68--75. {IEEE} Computer Society, 1997.

\bibitem[BDER97b]{DBLP:conf/csl/BaillotDER97}
Patrick Baillot, Vincent Danos, Thomas Ehrhard, and Laurent Regnier.
\newblock Timeless games.
\newblock In Mogens Nielsen and Wolfgang Thomas, editors, {\em Computer Science
  Logic, 11th International Workshop, {CSL} '97, Annual Conference of the
  EACSL, Aarhus, Denmark, August 23-29, 1997, Selected Papers}, volume 1414 of
  {\em Lecture Notes in Computer Science}, pages 56--77. Springer, 1997.
\newblock \href {https://doi.org/10.1007/BFb0028007}
  {\path{doi:10.1007/BFb0028007}}.

\bibitem[BEM07]{er:point}
Antonio Bucciarelli, Thomas Ehrhard, and Giulio Manzonetto.
\newblock Not enough points is enough.
\newblock In Jacques Duparc and Thomas~A. Henzinger, editors, {\em Computer
  Science Logic, 21st International Workshop, {CSL} 2007, 16th Annual
  Conference of the EACSL, Lausanne, Switzerland, September 11-15, 2007,
  Proceedings}, volume 4646 of {\em Lecture Notes in Computer Science}, pages
  298--312. Springer, 2007.
\newblock \href {https://doi.org/10.1007/978-3-540-74915-8\_24}
  {\path{doi:10.1007/978-3-540-74915-8\_24}}.

\bibitem[B{\'e}n73]{benabou1973distributeurs}
Jean B{\'e}nabou.
\newblock Les distributeurs.
\newblock {\em Universit{\'e} Catholique de Louvain, Institut de
  Math{\'e}matique Pure et Appliqu{\'e}e, rapport}, 33, 1973.

\bibitem[BM20]{DBLP:journals/pacmpl/BarbarossaM20}
Davide Barbarossa and Giulio Manzonetto.
\newblock Taylor subsumes {Scott}, {Berry}, {Kahn} and {Plotkin}.
\newblock {\em Proc. {ACM} Program. Lang.}, 4({POPL}):1:1--1:23, 2020.

\bibitem[Bou09]{DBLP:conf/tlca/Boudes09}
Pierre Boudes.
\newblock Thick subtrees, games and experiments.
\newblock In Pierre{-}Louis Curien, editor, {\em Typed Lambda Calculi and
  Applications, 9th International Conference, {TLCA} 2009, Brasilia, Brazil,
  July 1-3, 2009. Proceedings}, volume 5608 of {\em Lecture Notes in Computer
  Science}, pages 65--79. Springer, 2009.
\newblock \href {https://doi.org/10.1007/978-3-642-02273-9\_7}
  {\path{doi:10.1007/978-3-642-02273-9\_7}}.

\bibitem[CC24]{DBLP:journals/corr/abs-2103-15453}
Simon Castellan and Pierre Clairambault.
\newblock Disentangling parallelism and interference in game semantics.
\newblock {\em Log. Methods Comput. Sci.}, 20(3), 2024.

\bibitem[CCPW18]{DBLP:conf/lics/CastellanCPW18}
Simon Castellan, Pierre Clairambault, Hugo Paquet, and Glynn Winskel.
\newblock The concurrent game semantics of probabilistic {PCF}.
\newblock In {\em {LICS}}, pages 215--224. {ACM}, 2018.

\bibitem[CCRW17]{DBLP:journals/lmcs/CastellanCRW17}
Simon Castellan, Pierre Clairambault, Silvain Rideau, and Glynn Winskel.
\newblock Games and strategies as event structures.
\newblock {\em Log. Methods Comput. Sci.}, 13(3), 2017.

\bibitem[CCW14]{DBLP:conf/csl/CastellanCW14}
Simon Castellan, Pierre Clairambault, and Glynn Winskel.
\newblock Symmetry in concurrent games.
\newblock In {\em {CSL-LICS}}, pages 28:1--28:10. {ACM}, 2014.

\bibitem[CCW15]{lics15}
Simon Castellan, Pierre Clairambault, and Glynn Winskel.
\newblock The parallel intensionally fully abstract games model of {PCF}.
\newblock In {\em 2015 30th Annual ACM/IEEE Symposium on Logic in Computer
  Science}, pages 232--243. IEEE, 2015.

\bibitem[CCW19]{DBLP:journals/lmcs/CastellanCW19}
Simon Castellan, Pierre Clairambault, and Glynn Winskel.
\newblock Thin games with symmetry and concurrent {Hyland-Ong} games.
\newblock {\em Log. Methods Comput. Sci.}, 15(1), 2019.

\bibitem[CdV20]{DBLP:journals/pacmpl/ClairambaultV20}
Pierre Clairambault and Marc de~Visme.
\newblock Full abstraction for the quantum lambda-calculus.
\newblock {\em Proc. {ACM} Program. Lang.}, 4({POPL}):63:1--63:28, 2020.
\newblock \href {https://doi.org/10.1145/3371131} {\path{doi:10.1145/3371131}}.

\bibitem[CF23]{DBLP:conf/lics/ClairambaultF23}
Pierre Clairambault and Simon Forest.
\newblock The cartesian closed bicategory of thin spans of groupoids.
\newblock In {\em 38th Annual {ACM/IEEE} Symposium on Logic in Computer
  Science, {LICS} 2023, Boston, MA, USA, June 26-29, 2023}, pages 1--13.
  {IEEE}, 2023.
\newblock \href {https://doi.org/10.1109/LICS56636.2023.10175754}
  {\path{doi:10.1109/LICS56636.2023.10175754}}.

\bibitem[CF24]{DBLP:conf/lics/ClairambaultF24}
Pierre Clairambault and Simon Forest.
\newblock An analysis of symmetry in quantitative semantics.
\newblock In Pawel Sobocinski, Ugo~Dal Lago, and Javier Esparza, editors, {\em
  Proceedings of the 39th Annual {ACM/IEEE} Symposium on Logic in Computer
  Science, {LICS} 2024, Tallinn, Estonia, July 8-11, 2024}, pages 26:1--26:13.
  {ACM}, 2024.
\newblock \href {https://doi.org/10.1145/3661814.3662092}
  {\path{doi:10.1145/3661814.3662092}}.

\bibitem[Cla24]{hdr}
Pierre Clairambault.
\newblock Causal investigations in interactive semantics, 2024.
\newblock Habilitation à Diriger les Recherches.

\bibitem[CM10]{DBLP:journals/entcs/CalderonM10}
Ana~C. Calderon and Guy McCusker.
\newblock Understanding game semantics through coherence spaces.
\newblock In {\em {MFPS}}, volume 265 of {\em Electronic Notes in Theoretical
  Computer Science}, pages 231--244. Elsevier, 2010.

\bibitem[COP23]{DBLP:conf/lics/ClairambaultOP23}
Pierre Clairambault, Federico Olimpieri, and Hugo Paquet.
\newblock From thin concurrent games to generalized species of structures.
\newblock In {\em 38th Annual {ACM/IEEE} Symposium on Logic in Computer
  Science, {LICS} 2023, Boston, MA, USA, June 26-29, 2023}, pages 1--14.
  {IEEE}, 2023.
\newblock \href {https://doi.org/10.1109/LICS56636.2023.10175681}
  {\path{doi:10.1109/LICS56636.2023.10175681}}.

\bibitem[CP18]{DBLP:conf/csl/ClairambaultP18}
Pierre Clairambault and Hugo Paquet.
\newblock Fully abstract models of the probabilistic lambda-calculus.
\newblock In Dan~R. Ghica and Achim Jung, editors, {\em 27th {EACSL} Annual
  Conference on Computer Science Logic, {CSL} 2018, September 4-7, 2018,
  Birmingham, {UK}}, volume 119 of {\em LIPIcs}, pages 16:1--16:17. Schloss
  Dagstuhl - Leibniz-Zentrum f{\"{u}}r Informatik, 2018.
\newblock \href {https://doi.org/10.4230/LIPIcs.CSL.2018.16}
  {\path{doi:10.4230/LIPIcs.CSL.2018.16}}.

\bibitem[CS10]{cruttwell2010unified}
G~S~H Cruttwell and Michael~A Shulman.
\newblock A unified framework for generalized multicategories.
\newblock {\em Theory and Applications of Categories}, 24(21):580--655, 2010.

\bibitem[dC18]{carv:ex}
Daniel de~Carvalho.
\newblock Execution time of {\(\lambda\)}-terms via denotational semantics and
  intersection types.
\newblock {\em Math. Struct. Comput. Sci.}, 28(7):1169--1203, 2018.
\newblock \href {https://doi.org/10.1017/S0960129516000396}
  {\path{doi:10.1017/S0960129516000396}}.

\bibitem[Ehr63]{ehresmann1963categories}
Charles Ehresmann.
\newblock Cat{\'e}gories structur{\'e}es.
\newblock In {\em Annales scientifiques de l'{\'E}cole Normale Sup{\'e}rieure},
  volume~80, pages 349--426, 1963.

\bibitem[ER03]{DBLP:journals/tcs/EhrhardR03}
Thomas Ehrhard and Laurent Regnier.
\newblock The differential lambda-calculus.
\newblock {\em Theor. Comput. Sci.}, 309(1-3):1--41, 2003.

\bibitem[FGHW08]{fiore2008cartesian}
Marcelo Fiore, Nicola Gambino, Martin Hyland, and Glynn Winskel.
\newblock The cartesian closed bicategory of generalised species of structures.
\newblock {\em Journal of the London Mathematical Society}, 77(1):203--220,
  2008.

\bibitem[FGHW18]{fiore2018relative}
Marcelo Fiore, Nicola Gambino, Martin Hyland, and Glynn Winskel.
\newblock Relative pseudomonads, {Kleisli} bicategories, and substitution
  monoidal structures.
\newblock {\em Selecta Mathematica}, 24(3):2791--2830, 2018.

\bibitem[FP09]{DBLP:conf/tlca/FaggianP09}
Claudia Faggian and Mauro Piccolo.
\newblock Partial orders, event structures and linear strategies.
\newblock In {\em {TLCA}}, volume 5608 of {\em Lecture Notes in Computer
  Science}, pages 95--111. Springer, 2009.

\bibitem[FS19]{philip-marcelo:tt}
Marcelo Fiore and Philip Saville.
\newblock A type theory for cartesian closed bicategories (extended abstract).
\newblock In {\em 34th Annual {ACM/IEEE} Symposium on Logic in Computer
  Science, {LICS} 2019, Vancouver, BC, Canada, June 24-27, 2019}, pages 1--13.
  {IEEE}, 2019.
\newblock \href {https://doi.org/10.1109/LICS.2019.8785708}
  {\path{doi:10.1109/LICS.2019.8785708}}.

\bibitem[GGV24]{ggv24}
Nicola Gambino, Richard Garner, and Christina Vasilakopoulou.
\newblock Monoidal kleisli bicategories and the arithmetic product of coloured
  symmetric sequences.
\newblock {\em Doc. Math}, 29:627--702, 2024.

\bibitem[Gir87]{DBLP:journals/tcs/Girard87}
Jean{-}Yves Girard.
\newblock Linear logic.
\newblock {\em Theor. Comput. Sci.}, 50:1--102, 1987.
\newblock \href {https://doi.org/10.1016/0304-3975(87)90045-4}
  {\path{doi:10.1016/0304-3975(87)90045-4}}.

\bibitem[Gir88]{normal-functors}
Jean-Yves Girard.
\newblock Normal functors, power series and $\lambda$-calculus.
\newblock {\em Annals of pure and applied logic}, 37(2):129--177, 1988.

\bibitem[GJ17]{gambino-joyal}
Nicola Gambino and Andr{\'e} Joyal.
\newblock {\em On operads, bimodules and analytic functors}, volume 249.
\newblock American Mathematical Society, 2017.

\bibitem[GM08]{DBLP:journals/apal/GhicaM08}
Dan~R. Ghica and Andrzej~S. Murawski.
\newblock Angelic semantics of fine-grained concurrency.
\newblock {\em Ann. Pure Appl. Log.}, 151(2-3):89--114, 2008.

\bibitem[HO00]{DBLP:journals/iandc/HylandO00}
J.~M.~E. Hyland and C.{-}H.~Luke Ong.
\newblock On full abstraction for {PCF:} {I}, {II}, and {III}.
\newblock {\em Inf. Comput.}, 163(2):285--408, 2000.
\newblock \href {https://doi.org/10.1006/inco.2000.2917}
  {\path{doi:10.1006/inco.2000.2917}}.

\bibitem[HS19]{hansen-shulman}
Linde~Wester Hansen and Michael Shulman.
\newblock Constructing symmetric monoidal bicategories functorially.
\newblock {\em arXiv preprint arXiv:1910.09240}, 2019.

\bibitem[Joy81]{joyal1981theorie}
Andr{\'e} Joyal.
\newblock Une th{\'e}orie combinatoire des s{\'e}ries formelles.
\newblock {\em Advances in mathematics}, 42(1):1--82, 1981.

\bibitem[KMO23]{ol:proof}
Axel Kerinec, Giulio Manzonetto, and Federico Olimpieri.
\newblock Why are proofs relevant in proof-relevant models?
\newblock {\em {PACMPL}}, 7({POPL}):8:1--8:31, 2023.
\newblock \href {https://doi.org/10.1145/3571201} {\path{doi:10.1145/3571201}}.

\bibitem[KNO02]{DBLP:journals/tcs/KerNO02}
Andrew~D. Ker, Hanno Nickau, and C.{-}H.~Luke Ong.
\newblock Innocent game models of untyped lambda-calculus.
\newblock {\em Theor. Comput. Sci.}, 272(1-2):247--292, 2002.

\bibitem[Lac00]{lack2000coherent}
Stephen Lack.
\newblock A coherent approach to pseudomonads.
\newblock {\em Advances in Mathematics}, 152(2):179--202, 2000.

\bibitem[Lai97]{DBLP:conf/lics/Laird97}
James Laird.
\newblock Full abstraction for functional languages with control.
\newblock In {\em {LICS}}, pages 58--67. {IEEE} Computer Society, 1997.

\bibitem[LMMP13]{DBLP:conf/lics/LairdMMP13}
Jim Laird, Giulio Manzonetto, Guy McCusker, and Michele Pagani.
\newblock Weighted relational models of typed lambda-calculi.
\newblock In {\em 28th Annual {ACM/IEEE} Symposium on Logic in Computer
  Science, {LICS} 2013, New Orleans, LA, USA, June 25-28, 2013}, pages
  301--310. {IEEE} Computer Society, 2013.
\newblock \href {https://doi.org/10.1109/LICS.2013.36}
  {\path{doi:10.1109/LICS.2013.36}}.

\bibitem[Mel03]{mellies2003asynchronous}
Paul-Andr{\'e} Mellies.
\newblock Asynchronous games 1: Uniformity by group invariance, 2003.

\bibitem[Mel05]{DBLP:conf/lics/Mellies05}
Paul{-}Andr{\'{e}} Melli{\`{e}}s.
\newblock Asynchronous games 4: {A} fully complete model of propositional
  linear logic.
\newblock In {\em {LICS}}, pages 386--395. {IEEE} Computer Society, 2005.

\bibitem[Mel06]{DBLP:journals/tcs/Mellies06}
Paul{-}Andr{\'{e}} Melli{\`{e}}s.
\newblock Asynchronous games 2: The true concurrency of innocence.
\newblock {\em Theor. Comput. Sci.}, 358(2-3):200--228, 2006.

\bibitem[Mel09]{panorama}
Paul-Andr{\'e} Mellies.
\newblock Categorical semantics of linear logic.
\newblock {\em Panoramas et syntheses}, 27:15--215, 2009.

\bibitem[Mel19]{DBLP:conf/lics/Mellies19}
Paul{-}Andr{\'{e}} Melli{\`{e}}s.
\newblock Template games and differential linear logic.
\newblock In {\em {LICS}}, pages 1--13. {IEEE}, 2019.

\bibitem[MM07]{DBLP:conf/concur/MelliesM07}
Paul{-}Andr{\'{e}} Melli{\`{e}}s and Samuel Mimram.
\newblock Asynchronous games: Innocence without alternation.
\newblock In {\em {CONCUR}}, volume 4703 of {\em Lecture Notes in Computer
  Science}, pages 395--411. Springer, 2007.

\bibitem[NPW79]{DBLP:conf/scc/NielsenPW79}
Mogens Nielsen, Gordon~D. Plotkin, and Glynn Winskel.
\newblock Petri nets, event structures and domains.
\newblock In {\em Semantics of Concurrent Computation}, volume~70 of {\em
  Lecture Notes in Computer Science}, pages 266--284. Springer, 1979.

\bibitem[Oli21]{ol:intdist}
Federico Olimpieri.
\newblock Intersection type distributors.
\newblock In {\em 36th Annual {ACM/IEEE} Symposium on Logic in Computer
  Science, {LICS} 2021, Rome, Italy, June 29 - July 2, 2021}, pages 1--15.
  {IEEE}, 2021.
\newblock \href {https://doi.org/10.1109/LICS52264.2021.9470617}
  {\path{doi:10.1109/LICS52264.2021.9470617}}.

\bibitem[Paq20]{paquet2020probabilistic}
Hugo Paquet.
\newblock {\em Probabilistic concurrent game semantics}.
\newblock PhD thesis, 2020.

\bibitem[Paq22]{mfps22}
Hugo Paquet.
\newblock Bi-invariance for uniform strategies on event structures.
\newblock In {\em {MFPS}}, 2022.

\bibitem[RW11]{DBLP:conf/lics/RideauW11}
Silvain Rideau and Glynn Winskel.
\newblock Concurrent strategies.
\newblock In {\em {LICS}}, pages 409--418. {IEEE} Computer Society, 2011.

\bibitem[Sav20]{philip:thesis}
Philip Saville.
\newblock Cartesian closed bicategories: type theory and coherence, 2020.
\newblock \href {https://doi.org/10.48550/ARXIV.2007.00624}
  {\path{doi:10.48550/ARXIV.2007.00624}}.

\bibitem[See87]{seely:two}
Robert A.~G. Seely.
\newblock Modelling computations: {A} 2-categorical framework.
\newblock In {\em Proceedings of the Symposium on Logic in Computer Science
  {(LICS} '87), Ithaca, New York, USA, June 22-25, 1987}, pages 65--71. {IEEE}
  Computer Society, 1987.

\bibitem[Shu10]{shulman}
Michael~A Shulman.
\newblock Constructing symmetric monoidal bicategories.
\newblock {\em arXiv preprint arXiv:1004.0993}, 2010.

\bibitem[Str72]{street1972formal}
Ross Street.
\newblock The formal theory of monads.
\newblock {\em Journal of Pure and Applied Algebra}, 2(2):149--168, 1972.

\bibitem[TAO17]{DBLP:conf/lics/TsukadaAO17}
Takeshi Tsukada, Kazuyuki Asada, and C.{-}H.~Luke Ong.
\newblock Generalised species of rigid resource terms.
\newblock In {\em 32nd Annual {ACM/IEEE} Symposium on Logic in Computer
  Science, {LICS} 2017, Reykjavik, Iceland, June 20-23, 2017}, pages 1--12.
  {IEEE} Computer Society, 2017.
\newblock \href {https://doi.org/10.1109/LICS.2017.8005093}
  {\path{doi:10.1109/LICS.2017.8005093}}.

\bibitem[TAO18]{DBLP:conf/lics/TsukadaAO18}
Takeshi Tsukada, Kazuyuki Asada, and C.{-}H.~Luke Ong.
\newblock Species, profunctors and {Taylor} expansion weighted by {SMCC:} {A}
  unified framework for modelling nondeterministic, probabilistic and quantum
  programs.
\newblock In {\em {LICS}}, pages 889--898. {ACM}, 2018.

\bibitem[TO15]{DBLP:conf/lics/TsukadaO15}
Takeshi Tsukada and C.{-}H.~Luke Ong.
\newblock Nondeterminism in game semantics via sheaves.
\newblock In {\em {LICS}}, pages 220--231. {IEEE} Computer Society, 2015.

\bibitem[TO16]{DBLP:conf/lics/TsukadaO16}
Takeshi Tsukada and C.{-}H.~Luke Ong.
\newblock Plays as resource terms via non-idempotent intersection types.
\newblock In {\em {LICS}}, pages 237--246. {ACM}, 2016.

\bibitem[Win07]{DBLP:journals/entcs/Winskel07}
Glynn Winskel.
\newblock Event structures with symmetry.
\newblock {\em Electron. Notes Theor. Comput. Sci.}, 172:611--652, 2007.

\bibitem[Woo82]{wood1982abstract}
Richard~J Wood.
\newblock Abstract proarrows {I}.
\newblock {\em Cahiers de topologie et g{\'e}om{\'e}trie diff{\'e}rentielle},
  23(3):279--290, 1982.

\bibitem[Yon60]{yoneda1960ext}
Nobuo Yoneda.
\newblock On ext and exact sequences.
\newblock {\em J. Fac. Sci. Univ. Tokyo Sect. I}, 8(507-576):1960, 1960.

\end{thebibliography}

\end{document}